\documentclass[aos]{imsart}
\usepackage{imsart}
\usepackage{amsthm}
\usepackage{amsmath}
\usepackage{amsfonts}
\usepackage{amssymb}
\usepackage{mathtools}
\usepackage[authoryear]{natbib}
\usepackage{graphicx}
\usepackage{booktabs}
\usepackage{multirow}
\usepackage{threeparttable}
\usepackage{adjustbox}
\usepackage{mathrsfs}
\usepackage{bm}
\usepackage{algorithm}
\usepackage{algorithmicx}
\usepackage{algpseudocode}
\usepackage[hidelinks]{hyperref}
\startlocaldefs
\theoremstyle{plain}
\newtheorem{theorem}{Theorem}
\newtheorem{proposition}[theorem]{Proposition}
\newtheorem{lemma}[theorem]{Lemma}
\newtheorem{corollary}[theorem]{Corollary}
\theoremstyle{definition}
\newtheorem{example}{Example}
\newtheorem{remark}{Remark}

\newtheorem{assumption}{Assumption}

\newcommand{\Assumption}[1][]{\begin{assumption}[#1]}
\newcommand{\EndAssumption}{\end{assumption}}

\newcommand{\R}{\mathbb{R}}
\newcommand{\diag}{\operatorname{diag}}

\theoremstyle{definition}
\newtheorem{proofcondition}{Proof condition}

\newcommand{\SupplementTocStart}{}
\providecommand{\tr}{\operatorname{tr}}
\providecommand{\Diag}{\operatorname{Diag}}
\providecommand{\op}{\operatorname{op}}

\providecommand{\cD}{\mathcal D}
\providecommand{\cM}{\mathcal M}
\providecommand{\cW}{\mathcal W}
\providecommand{\cR}{\mathcal R}
\providecommand{\norm}[1]{\left\lVert #1\right\rVert}
\providecommand{\ip}[2]{\left\langle #1,#2\right\rangle}
\providecommand{\Zh}{\widehat Z}
\providecommand{\Mh}{\widehat M_x}
\providecommand{\Seh}{\widehat\Sigma_e}
\providecommand{\Syh}{\widehat\Sigma_y}
\providecommand{\Ph}{\widehat P}
\providecommand{\Sh}{\widehat{\mathscr S}}
\providecommand{\Gh}{\widehat G_H}

\endlocaldefs
\begin{document}

\begin{frontmatter}
\title{Identification problem and quasi-maximum likelihood estimation for matrix-variate CP-factor models}
\runtitle{Identification and QMLE for CP-factor models}
\begin{aug}
\author[A]{\fnms{Hanzi}~\snm{Ye}}
\author[A]{\fnms{Chun Yip}~\snm{Yau}\ead[label=e1]{cyyau@sta.cuhk.edu.hk}}
\runauthor{H.YE AND C.Y.YAU}
\address[A]{Department of Statistics and Data Science, The Chinese University of Hong Kong\printead[presep={,\ }]{e1}}
\end{aug}
\begin{abstract}
Matrix-valued time series, arising in diverse fields such as economics, neuroscience, and recommender systems, have become increasingly prominent in modern data analysis. Among various modeling frameworks, the matrix-variate CP-factor model represents an important and widely applicable class for capturing low-rank structures in matrix time series. In this paper, we provide a unified treatment for the identification problem of CP factor models using both analytic and algebraic tools. 
In particular, we characterize the parameter space as the union of two subspaces, one identifiable and the other non-identifiable. We show that existing estimation methods only apply to an open proper subset of the identifiable subspace. In contrast, we propose a quasi-maximum likelihood estimation (QMLE) procedure for CP factor models, which allows consistent estimation on the whole identifiable subspace. Moreover, we show that the estimated loading matrices by QMLE achieve a faster convergence rate compared to existing approaches. A simulation study and a real application are conducted to demonstrate the finite-sample performance of the proposed method.
\end{abstract}
\end{frontmatter}

\section{Introduction}\label{intro}
Many modern time series are naturally matrix- or tensor-valued, where each time point corresponds to a multi-way array rather than a scalar or vector observation. Such data arises when measurements are recorded across multiple cross-sectional units, spatial locations, or other structured dimensions simultaneously. Examples include traffic flow matrices across origin-destination pairs, spatio-temporal environmental measurements on spatial grids, and demand arrays indexed by regions and product categories.

These data are inherently high-dimensional: at each time index, one observes an entire matrix or tensor. Modeling such high-dimensional data is intrinsically challenging. A common and effective approach is to impose a low-dimensional latent factor structure, where the underlying signal is driven by a small number $d$ of common factors. Modeling matrix- or tensor-valued data via a factor structure offers a decisive benefit: it substantially reduces the effective dimension of the problem, allowing us to capture the dominant patterns of variation with only a few latent factors, while still preserving the original matrix/tensor structure of the observations. To fix the notation, we state the matrix factor model proposed by \cite{WANG2019231}: the observation matrix $Y_t$ at each time point $t$ can be represented as \begin{equation}\label{factormodel}Y_t=AF_tB^{\top}+e_t\,, \quad F_t\in \mathbb{R}^{d_1\times d_2},\quad t\in\mathbb{Z}_+\,,\end{equation}
where $F_t$ is a $d_1\times d_2$-matrix time series representing unknown time-varying latent factors, $e_t$ is a $p\times q$ matrix white noise, and $A=(a^1,\cdots, a^{d_1})$ and $B=(b^1,\cdots, b^{d_2})$ are, respectively, $p\times d_1$ and $q\times d_2$ unknown parameters. Although the matrix factor model in \eqref{factormodel} allows flexible modeling of matrix-valued time series, the number (i.e., $d_1\times d_2$) of latent factors can still be large for moderate values of $d_1$ and $d_2$. 
To reduce model complexity and enhance interpretability, \cite{CHANGHEYANGYAO2021} proposed the CP factor model of order $d$: 
\begin{eqnarray}\label{eq:cpfactmodel}
    Y_t=AX_tB^{\top}+e_t &=& \operatorname{Mat}_{p\times q}\{(B\odot A)x_t\}+e_t\,,\\*
        &=&\sum_{\ell=1}^{d}\operatorname{Mat}_{p\times q}\{b^\ell\otimes a^{\ell}\cdot x_t^{\ell}\} +e_t\,,\nonumber
    \end{eqnarray} 
where $\ A=(a^1,\cdots,a^d)$ and $B=(b^1,\cdots,b^d)$ are unknown loading matrices, 
$x_t=(x_t^{1},\ldots,x_{t}^{d})\in\mathbb{R}^d$, the factor $X_t=\operatorname{diag}(x_t)$ is a latent vector time series, $e_t \in \mathbb{R}^{p\times q}$ is a white noise matrix, $\odot $ is the Khatri–Rao product defined as $B\odot A=(b^1\otimes a^1,\dots,b^d\otimes a^d)$ (see \cite{khatriraoprod}), and the operator $\operatorname{Mat}$ transforms a vector into a matrix: $(v_1^{\top},\cdots,v_q^{\top})^{\top}\in\mathbb{R}^{pq}\to (v_1,\cdots, v_q)\in\mathbb{R}^{p\times q}, \quad v_i\in \mathbb{R}^p .$

The name of Model \eqref{eq:cpfactmodel} originates from the fact that, by stacking $\{Y_t\}_{t=1,\ldots,T}$ into a three-way tensor $\mathcal{Y}$, the CP-factor model \eqref{eq:cpfactmodel} becomes a 
CANDECOMP/PARAFAC (CP) decomposition of $\mathcal{Y}=\sum_{\ell=1}^d a^\ell\otimes b^{\ell}\otimes x^\ell+e$ in \cite{KRUSKAL197795} and \cite{Lathauwer06}, where $x^\ell=(x_1^\ell,\ldots,x_T^\ell)^\top$. 
This classical framework provides a rich source of methodological references.

The identification problem is a long-standing challenge in the CP decomposition literature, and CP-factor models inherit this difficulty as a special case where the component \((x^{\ell})\) is a latent time series. 
To illustrate the identification issue, 
a basic observation is that the model is invariant to column permutations and scaling: for any non-zero constants $c_{\ell}, d_{\ell}$ and permutation 
$\{\sigma_\ell\}$ of $\{1,\ldots,d\}$, we always have \(\sum _{\ell=1}^da^{\ell}\otimes b^{\ell}\otimes x^\ell=\sum_{\ell=1}^{d}c_{\sigma_\ell}a^{\sigma_\ell}\otimes d_{\sigma_\ell}b^{\sigma_\ell}\otimes \frac{1}{c_{\sigma_\ell}d_{\sigma_\ell}}x^{\sigma_\ell}\). 
 Accordingly, we call a CP decomposition (and analogously the CP-factor model) \emph{identifiable} if its parameters are unique modulo the permutation and scaling ambiguities. We argue in Section \ref{sec:idCPmodel} that, 
all the parameters \((A,B)\) for which the CP-factor model is identifiable form a set \(\mathcal {K}_d\): \begin{equation}\label{identifiablespacewithoutquotient}
{\mathcal{K}}_d
=
\left\{
(A,B)\;\middle|\;
{
\begin{array}{l}
\big(\operatorname{Mat}(b^{\ell}\otimes a^{\ell}),\ell\in[d]\big)\text{ is the only rank-1 basis up to scaling}\\\text{of the linear subspace }\operatorname{span}\{\operatorname{Mat}(b^{\ell}\otimes a^{\ell}),\ell\in[d]\}\subset \mathbb{R}^{p\times q}
\end{array}}
\right\}\,.
\end{equation} 
We call $\mathcal{K}_d$, or more precisely $\mathcal{K}_d(p,q,d)$, an identifiable space of dimensions $(p,q,d)$.
In general, given a parameter \((A,B)\in \mathbb{R}^{p\times d}\times \mathbb{R}^{q\times d}\), deciding whether it belongs to an identifiable space such as \({\mathcal K}_d\) can be challenging. This motivates a line of research for verifiable sufficient conditions for identifiability.  

A classical, easily checkable sufficient condition is due to \cite{KRUSKAL197795}. In the CP decomposition setting, if the Kruskal ranks (denoted by $\operatorname{krank}(\cdot)$, which stands for the largest integer $k$ such that every subset of $k$ columns is linearly independent) satisfy \[\operatorname{krank}(A)+\operatorname{krank}(B)+\operatorname{krank}(X)\geq 2d+2\,,\quad X=(x^{1},\ldots,x^d)\,,\] then the decomposition is unique modulo column permutations and scaling. Applying this to the factor model time series setting, we have that the CP-factor model is identifiable (up to column permutations and scaling) if $(A,B)$ lies in the parameter subspace: 
\(
{\mathcal{K}}_{\operatorname{krank}}
\coloneqq \Big\{(A,B)\in\mathbb{R}^{p\times d}\times \mathbb{R}^{q\times d}:
\operatorname{krank}(A)+\operatorname{krank}(B)\ge d+2\Big\}\,.
\)
More recently, \cite{CHANGDUHUANGYAO2024} 
develop the ideas in \cite{Lathauwer06} and introduce a matrix $\Omega(B\odot A)$. They enlarge the parameter subspace of identifiable models from ${\mathcal{K}}_{\operatorname{krank}}$ to ${\mathcal{K}}_\Omega\coloneqq\{(A,B):\dim\ker \Omega(B\odot A)=d\}$. In summary, ${\mathcal{K}}_{\operatorname{krank}}$ and ${\mathcal{K}}_\Omega$ both provide efficiently checkable (and sufficient) conditions for identifiability, but they may not be necessary conditions.

One of our main contributions is to derive a complete, verifiable characterization of identifiability, and to study the topology of the parameter space $\mathcal{K}_d$. Specifically, we establish a necessary and sufficient condition for identifiability by applying methods in algebraic geometry. We also show that the identifiable space 
${\mathcal{K}}_{\Omega}$ in 
\cite{CHANGDUHUANGYAO2024}
is an open, proper subset of $\mathcal{K}_d$ (with respect to the relative Euclidean topology). 
To determine whether a given $(A,B)$ lies in $\mathcal{K}_d$, we construct a collection of quadratic forms $\{\Psi_{m,n;k,\ell}(\lambda)\}$ that map the vector $\lambda=(\lambda_1,\ldots,\lambda_d)^\top$ to the $(m,n;k,l)$-indexed $2\times 2$ minors of the matrix $\sum_{\ell\in[d]}\lambda_\ell \operatorname{Mat}(b^{\ell}\otimes a^{\ell})$. By adopting an algebraic-geometry perspective, the identifiable condition $(A,B)\in {\mathcal{K}}_d$ is shown to be equivalent to that the zero locus of equation system $\{\Psi_{m,n;k,\ell}(\lambda)=0\}$ is the real coordinate axes. This reformulation reduces checking identifiability to checking whether the ideal generated by \(\{\Psi_{m,n;k,\ell}\}\) defines the union of the coordinate axes,
which can be certified using Gröbner bases in standard computer algebra systems such as \cite{M2,Singular,CoCoA}.

Another main contribution of this paper is to develop a quasi-maximum likelihood estimation (QMLE) approach for estimating the loading matrices. Existing estimation methods, including the eigen-based method (CP-unified) in \cite{CHANGHEYANGYAO2021,CHANGDUHUANGYAO2024}, the Higher-Order Projection method (HOPE) in \cite{HANYANGZHANGCHEN2024} and the Iterative Simultaneous Orthogonalization (ISO) method in \cite{CHEN2026106167}, are predominantly based on spectral or eigenvalue-based procedures and are applicable only under additional structural restrictions on the loading matrices. In particular, the iterative procedures underlying HOPE and ISO apply only to a strict subset of ${\mathcal K}_{\mathrm{krank}}$, since their projection steps require the loading matrices to be of full column rank. The restrictions imposed by CP-unified are substantially weaker, allowing parameters in ${\mathcal K}_{\Omega}$. In contrast, our QMLE imposes no additional restrictions on the column ranks of $A$ and $B$ and is applicable over the entire parameter space ${\mathcal K}_d$. By extending projection techniques to the theoretical analysis of QMLE, we establish a dimension-free parametric convergence rate of order $1/T$ for mean squared error even when the loading matrices are rank-deficient. Furthermore, we establish a sharper rate of order $\frac{p+q}{pqT}$ for loading matrices under additional distribution assumptions without constraints on the rank of loading matrices. Our numerical studies also show that QMLE yields small estimation errors and remains stable as the dimensions increase, demonstrating the practical effectiveness of the likelihood-based approach.

In parallel to our work, QMLE has recently been studied in the context of matrix factor models \eqref{factormodel} by \cite{Xu03042025}. Since no additional structure is imposed on \(F_t\), 
the identification of $(A,B,F_t)$ can be achieved by imposing simple affine constraints, such as requiring the leading square loading matrices to be identity. The resulting parameter space of identifiable models is an open smooth manifold. This structure closely resembles that of the vector factor model studied by \cite{BAILI2012aos}, and hence permits a similar QMLE framework. However, the CP factor model is fundamentally different: \((A,B)\) may be rank deficient, and some of them can never be sent to the identifiable parameter space $\mathcal{K}_d$ via affine transformations. In addition, $\mathcal{K}_d$ has a more intricate geometry, with degenerate configurations (see Propositions \ref{prop:splittingoffleadingtorankdeficient}, \ref{prop:geometric meaning of rank deficiency}) giving rise to boundary singularities. Our new QMLE framework, based on the projection method, does not rely on simple affine identification constraints and addresses these difficulties.  

The rest of the paper is organized as follows. Section~\ref{sec:idCPmodel} studies the identification problem for CP-factor models via algebraic and analytic methods. The quasi-likelihood estimation and its asymptotic properties are presented in Section \ref{sec:likelihoodfunctions&asympthm}.  Sections~\ref{sec:simulationdesign} and \ref{sec:numericalresult} provide simulation studies and a real financial data analysis to demonstrate the finite-sample performance of the proposed method, respectively. Technical details are provided in the supplementary materials.

\section{Identification problem for CP-factor models}\label{sec:idCPmodel}

In this section, we derive the identifiable space $\mathcal{K}_d$ and analyze its properties. 

\subsection{Identifiable condition}\label{subsec:preliminary}
This section develops the preliminary setup as in \cite{CHANGDUHUANGYAO2024}, but is reformulated in a geometric/set-theoretic language tailored to our purposes.

First, besides the scaling and permutation conditions, note that $B\odot A$ should be of full rank $d$ to ensure the uniqueness of $d$. Otherwise, we have $b^d\otimes a^d=\sum_{\ell=1}^{d-1}\gamma_\ell\,(b^\ell\otimes a^\ell)$ for some $(\gamma_1,\ldots,\gamma_{d-1})\in\mathbb{R}^{d-1}$, and thus 
$${Y_t}= \sum_{\ell=1}^d\operatorname{Mat}_{p\times q}\{b^\ell\otimes a^\ell\cdot x_t^{\ell}\}+e_t=
\sum_{\ell=1}^{d-1}\operatorname{Mat}_{p\times q}\{b^\ell\otimes a^{\ell}\cdot (x_t^{\ell}+\gamma_{\ell}x_t^d)\}+e_t\,,$$
reduces to a CP factor model of order $d-1$.

To deduce the necessary conditions for identifiability, suppose that two pairs $(A,B)$ and $(A',B')$ can represent the same $\{{Y}_t\}$, i.e.,  $\sum_{\ell=1}^d\operatorname{Mat}_{p\times q}\{b^{\ell}\otimes a^{\ell}\cdot  x_t^{\ell}\}=\sum_{\ell=1}^d\operatorname{Mat}_{p\times q}\{b'^{\ell}\otimes  a'^{\ell}\cdot  x'^{\ell}_t\}$ for some $x_t$ and $x_t'$, $\forall t$. This implies
\begin{equation}\label{eq:unfoldingform}
(B\odot A)\,(x_1,x_2,\ldots,x_T)  = (B'\odot A')\,(x_1',\ldots,x_T')\,. 
\end{equation}
Denoting $X=(x_1,\ldots,x_T)^{\top}$, and multiplying both sides of \eqref{eq:unfoldingform} by the right inverse \(X^\dagger\in\mathbb{R}^{T\times d}\) of \(X^\top\), we have \begin{equation}\label{defofH}(B\odot A) = (B'\odot A')\,H,\end{equation} where \(H := X'^{\top}X^\dagger  \in \mathbb{R}^{d\times d}\). 
As $B\odot A$ has column rank $d$, the column spaces $\operatorname{col}(B\odot A)=\operatorname{span}\{b^{\ell}\otimes a^{\ell},\ell\in[d]\}$
and 
$\operatorname{col}(B'\odot A')=\operatorname{span}\{b'^{\ell}\otimes a'^{\ell},\ell\in[d]\}$ are the same $d$-dimensional space. Since $\operatorname{Mat}_{p\times q}(b^{\ell}\otimes a^{\ell})=a^{\ell}b^{\ell\top}$ is a bijective operation,
it follows that the linear matrix spaces $M(A,B)$ and $M(A',B')$ are the same, where
\begin{eqnarray}\label{MAB}
\mathrm{M}(A,B)=\operatorname{span}\{a^{\ell}b^{\ell\top},\ell\in[d]\}\subset \mathbb{R}^{p\times q}\,,
\end{eqnarray}
is the space spanned by rank-1 tensors $\{a^{\ell}b^{\ell\top}\}_{\ell\in[d]}$.
In conclusion, model non-identifiabilty occurs when 
there exist two pairs $(A,B)$, $(A',B')$, which cannot be obtained from each other by scaling and permutation, forming respectively 
two different rank-1 bases $\{a^{\ell}b^{\ell\top}\}_{\ell\in[d]}$ and
$\{a'^{\ell}b'^{\ell\top}\}_{\ell\in[d]}$ which span the same space $M(A,B)$. We summarize the results as follows.

\begin{proposition}\label{prop:CPidcond}
A CP-factor model determined by $(A,B)$ is identifiable (up to scaling and permutations) if and only if the rank-1 basis of the linear subspace $\mathrm{M}(A,B)\subset \mathbb{R}^{p\times q}$ is unique (up to scaling and permutations). In other words, each element in $\mathcal{K}_d$ defined in \eqref{identifiablespacewithoutquotient} is identifiable. 
\end{proposition}

\subsection{Algebraic characterization of identifiable condition}\label{subsec:algidcond}

This subsection studies the space $\mathcal{K}_d$ through algebraic methods and provides an easily checkable condition to determine if a given CP-factor model determined by parameter $(A,B)$ is identifiable or not. 
For simplicity, for the rest of the paper, the terms identifiability and uniqueness are understood to be up to scaling and permutation.

Proposition~\ref{prop:CPidcond} shows that the identifiability of a CP-factor model $(A,B)$ is equivalent to the uniqueness of the rank-1 basis of $\mathrm{M}(A,B)$. Note that, for any $\lambda=(\lambda_1,\ldots,\lambda_d)\in\mathbb{R}^d$, if the matrix $m(\lambda)=\operatorname{Mat}_{p\times q}\big(B\odot A\cdot\lambda\big)=\sum_{\ell=1}^{d}\lambda_{\ell}a^{\ell}b^{\ell\top}\in \mathrm{M}(A,B)$ is of rank 1, then it can be completed together with elements from the given basis $\{a^{1}b^{1\top},\ldots a^db^{d\top}\}$ to form a rank-1 basis. Consequently, if the rank-1 basis of $\mathrm{M}(A,B)$ is unique, then $\mathrm{M}(A,B)$ contains no extra rank-1 direction beyond these $d$ basis directions, i.e., 
\begin{eqnarray}\label{rank.m.eq.1}
(A,B)\text{ is identifiable} \  \ \Longleftrightarrow  \ \ \operatorname{rank}\big(m(\lambda)\big)=1 \text{ if and only if }\lambda\in \cup_{\ell\in[d]}\mathbb{R}\cdot e_{\ell}\,,
\end{eqnarray}
where $\{e_{\ell},\ell\in[d]\}$ denotes the standard unit basis in $\mathbb{R}^{d}$. 

A well-known result in linear algebra is that, a matrix is rank-1 if and only if all its $2\times 2$-minors are zero. 
Hence, we define the  family of matrix $2\times 2$-minors of $m(\lambda)$ indexed by $(m,n;k,l)$, i.e., the rows $m,n\in[p]$ and columns $k,l\in[q]$, as 
\begin{equation}\label{2minorquadraticequation}
    \begin{aligned}
        \Psi_{m,n;k,l}(\lambda) &=\det \big((\sum_{\ell=1}^d\lambda_\ell a^{\ell}b^{\ell\top})_{\{m,n\}\times\{k,l\}}\big)\\
        &=\sum_{i< j\in[d]}2\lambda_i\lambda_j\psi_{m,n;k,l}( a^ib^{i\top}, a^jb^{j\top})\,, 
    \end{aligned}
\end{equation}where \begin{equation}\label{def:Psimap}\psi_{m,n;k,l}(a^ib^{i\top},a^jb^{j\top}):=\frac{1}{2}\det\begin{bmatrix}
            a^i_{m}b^i_{k}&a^j_{m}b^j_{l}\\a^i_{n}b^i_{k}&a^j_{n}b^j_{l}
        \end{bmatrix}+\frac{1}{2}\det\begin{bmatrix}
            a^j_{m}b^j_{k}&a^i_{m}b^i_{l}\\a^j_{n}b^j_{k}&a^i_{n}b^i_{l}
        \end{bmatrix}\,,
\end{equation}
is the 4-way tensor defined in \cite{Lathauwer06}. 
Observe from \eqref{2minorquadraticequation} that every $2\times 2$-minor $\Psi_{m,n;k,l}(\lambda)$ is a quadratic polynomial in $\lambda$. Thus, $m(\lambda)$ is of rank-1 if and only if $\lambda$ is the common root of all polynomials $\Psi_{m,n;k,l}(\lambda)$, $m,n\in [p], k,l\in[q]$.
Together with Proposition~\ref{prop:CPidcond} and \eqref{rank.m.eq.1}, we 
have the following results.

\begin{proposition}\label{prop:algcond}
    A pair $(A,B)$ is identifiable, i.e. $(A,B)\in {\mathcal{K}}_d$, if and only if the common zero set of the quadratic forms only contains the coordinate axes: 
    \begin{equation}\label{eq:commonzerosetisaxis}
        \cap_{(m,n;k,l)} \mathrm{Z}\big(\Psi_{m,n;k,l}(\lambda)\big)=\cup_{\ell\in[d]} \mathbb{R}\cdot e_{\ell}, 
    \end{equation}
    where $\mathrm{Z}(\cdot)$ denotes the zero set of a polynomial, and $\{e_{\ell},\ell\in[d]\}$ is the standard unit basis. 
\end{proposition}

In practice, to determine whether \eqref{eq:commonzerosetisaxis} holds, one can compute the Gr\"obner basis of \(\Psi:=\{\Psi_{m,n;k,l}(\lambda):  m,n\in[p],\,k,l\in[q]\}\) over the polynomial ring $\mathbb{R}[\lambda]$, which is a set of polynomials $\{g_1,\ldots,g_s\}$ in $\lambda$, satisfying
\(
\bigcap_{m,n,k,l}\mathrm{Z}(\Psi_{m,n;k,l}(\lambda))
=
\bigcap_{i=1}^s \mathrm{Z}(g_i)\,, 
\)  
see, e.g., \cite{Buchberger1976CanonicalForms,AnIntrotoCompuAGCA}. 
Thus, it suffices to compute the common zero set of the $s$ polynomials in the 
Gr\"obner basis instead of the $(pq)^2$ polynomials in $\Psi$. Moreover, \(g_i\)'s are often much simpler than the original $\Psi_{m,n;k,l}(\lambda)$'s, making~\eqref{eq:commonzerosetisaxis} easier to verify theoretically.

\cite{CHANGDUHUANGYAO2024} show that  
${\mathcal{K}}_\Omega\coloneqq\{(A,B):\dim\ker \Omega(B\odot A)=d\}$ is identifiable, 
where $\Omega(B\odot A)
:= \Big[\,\operatorname{vec}\big(\psi( a^ib^{i\top},\ a^jb^{j\top})\big)\,\Big]_{1\le i\leq j\le d}
\in\mathbb{R}^{p^2q^2\times \frac{d(d+1)}{2}}$. 
The following corollary show that $\mathcal{K}_\Omega$ is smaller than $\mathcal{K}_d$, and thus our identification condition is more general.

\begin{corollary}\label{coro:relationtoyao24}
We have ${\mathcal{K}}_{\Omega}\subset \operatorname{int}({\mathcal{K}}_d)$.  
\end{corollary}

The following example provides an explicit example of 
$(A,B)$ which belongs to $\mathcal{K}_d$ but not 
$\mathcal{K}_{\Omega}$. 

\begin{example}\label{examplenumeric}
    Consider $p=q=3\,,d=5$, with 
    \begin{align}\label{exampleforA,Binbdry}A=\begin{pmatrix} 
        1&1&0&1&0\\0&2&1&-1&0\\0&1&-1&2&1
    \end{pmatrix},\quad B=\begin{pmatrix}
        2&1&-1&-1&0\\0&-1&5&3&0\\0&-1&3&1&2
    \end{pmatrix}.\end{align}
It is readily checked that, $\operatorname{rank}(A)=\operatorname{rank}(B)=3,\ \operatorname{rank}\Omega(B\odot A)=9$, 
implying that 
$\dim\ker \Omega(B\odot A)=6 \neq 5$, i.e., 
$(A,B)\notin \mathcal{K}_{\Omega}$. On the other hand, by symbolic computation using a computer algebra system (e.g. in Python via SymPy \citep{MeurerEtAl2017SymPy}), the Gr\"{o}bner basis for $\Psi_{m,n;k,l}$'s is
$$\mathcal{G}(A,B):=\big[ \lambda_1\lambda_5^2, (\lambda_1\lambda_2-\lambda_1\lambda_5), (\lambda_1\lambda_3-\lambda_1\lambda_5),(\lambda_1\lambda_4-\lambda_1\lambda_5),\lambda_i\lambda_j\big|2\leq i<j\leq5\big]\,.$$
It is straightforward to verify that the common zeros are the coordinate axes. Thus, $(A,B)\in\mathcal{K}_d$. 
\end{example}

As an alternative to the Gr\"obner basis approach, one may resort to efficient numerical algorithms to verify \eqref{eq:commonzerosetisaxis}. 
In algebraic geometry, to determine whether two polynomial systems have the same common zero set, one often uses Hilbert's Nullstellensatz to reduce the problem to checking whether the associated radical ideals coincide. Here,  
the \emph{real radical} of a polynomial system is defined as 
\[
\sqrt[\mathbb{R}]{F}
\;\coloneqq\;
\left\{f\in \mathbb{R}[\lambda_1,\ldots,\lambda_d] \middle|\;
{
\begin{array}{l}
\text{for }\forall p_i\in\mathbb{R}[\lambda],i\in[s],\ \exists\, m,k\in\mathbb{N},\ \exists\, s_j\in \mathbb{R}[\lambda]\\\text{s.t.}\
f^{2m}+\sum_{j=1}^k s_j^2 =\sum_{i}p_i f_i, f_i\in F
\end{array}}
\right\}\,,
\]
and the Hilbert's Nullstellensatz theorem is stated as follows. 
\begin{theorem}[Real Hilbert's Nullstellensatz \citep{AnIntrotoCompuAGCA,atiyahmacCOMMUTATIVEALG}]\label{thm.H.N.Thm}
    Two polynomial systems $F=(f_1,\ldots,f_s),G=(g_1,\ldots,g_r)$ have the same common zeros in $\mathbb{R}$, i.e. $\mathrm{Z}(F)=\mathrm{Z}(G)$, if and only if the corresponding real radicals are the same: $\sqrt[\mathbb{R}]{F}=\sqrt[\mathbb{R}]{G}$. 
\end{theorem}

 It is easy to see that the zero set of the polynomial system $\Lambda\coloneqq\{\lambda_i\lambda_j:i<j\in[d]\}$ is 
$\cup\mathbb{R}\cdot e_{\ell}$. Together with Proposition \ref{prop:algcond} and Theorem \ref{thm.H.N.Thm}, we have that $(A,B)\in \mathcal{K}_d$ is equivalent to 
 \begin{equation}\label{eq:sqrtideal}
        \sqrt[\mathbb{R}]{\Psi}=\sqrt[\mathbb{R}]{\Lambda}=\left\{\sum_{i<j\in[d]}p_{i,j}\lambda_i\lambda_j:\forall p_{i,j}\in\mathbb{R}[\lambda]\right\}\,,
    \end{equation}
which can be readily verified by software tools such as Macaulay2 \cite{M2}, Singular \cite{Singular}, or CoCoA \cite{CoCoA}.

In summary, to determine whether a pair $(A,B)$ is identifiable, we compute all $2\times2$-minors of $m(\lambda):=\operatorname{Mat}\big(B\odot A\cdot \lambda\big)\in\mathbb{R}^{d_1\times d_2}$, denoted by \(\{\Psi_{m,n;k,\ell}(\lambda)\}\), and then verify \eqref{eq:sqrtideal} for $B\odot A$ via Macaulay2\cite{M2}. For illustration, in Example \ref{examplenumeric}, the radical of its polynomial system is $\sqrt[\mathbb{R}]{\Lambda}$ since $\lambda_1\lambda_5^2\cdot \lambda_1=(\lambda_1\lambda_5)^2$ implies $\lambda_1\lambda_5\in\sqrt[\mathbb{R}]{\mathcal{G}}$, and hence by definition, $\Lambda\subset \sqrt[\mathbb{R}]{\mathcal{G}}$, and thus $\sqrt[\mathbb{R}]{\mathcal{G}}=\sqrt[\mathbb{R}]\Lambda$. 

To further simplify the computation, a QR decomposition $(A_0A_r,B_0B_r)$ of $(A,B)$ may be conducted, where $A_0^{\top}A_0=I_{d_1}$ and $B_0^{\top}B_0=I_{d_2}$. Since the equations \begin{equation}\label{eq:idisstableunderQR}\begin{aligned}
    &\sum_{\ell\in[d]}\lambda_\ell a^{\ell}b^{\ell\top}=A_0(\sum_{\ell\in[d]}\lambda_{\ell}a_r^{\ell}b_r^{\ell\top})B_0^{\top}\,, \ \ \  A_0^{\top}\sum_{\ell\in[d]}\lambda_\ell a^{\ell}b^{\ell\top}B_0=\sum_{\ell\in[d]}\lambda_{\ell}a_r^{\ell}b_r^{\ell\top};
\end{aligned}\end{equation} give a one-to-one correspondence between $\mathrm{M}(A,B)$ and $\mathrm{M}(A_r,B_r)$, we have that $(A,B)\in{\mathcal{K}}_d(p,q,d)$ if and only if $(A_r,B_r)\in{\mathcal{K}}_d(d_1,d_2,d)$. This reduces the number of polynomials in checking \eqref{eq:sqrtideal}.

\subsection{\texorpdfstring{Interior and boundary of \({\mathcal{K}}_d\)}{Interior and boundary of Ktilde_d}}\label{subsec:intvsbdry}

This section characterizes the boundary of the space ${\mathcal{K}}_d$. 
As will be shown in Section \ref{sec:asympproperty}, the asymptotic behavior of estimators differs depending on whether the true parameters lie on the boundary
or in the interior of ${\mathcal{K}}_d$.

Consider a point $p_0=(A_0,B_0)$ lying on the boundary of $\mathcal{K}_d$. 
By definition, for any sufficiently small $\epsilon>0$, there exists a direction $\delta=(\delta_A,\delta_B)\in\mathbb{R}^{p\times d}\times\mathbb{R}^{q\times d}$, $\|\delta_A\|_F=\|\delta_B\|_F=1$, such that the point $p_\delta(\epsilon)=(A_0+\epsilon\delta_A, B_0+\epsilon\delta_B)\notin\mathcal{K}_d$. Moreover, 
the following result shows that we can find an analytic family $\delta_t$ (with all entries real analytic functions) such that the whole path $\{p_{\delta_t}(t):0<t<\epsilon\}$ is outside $\mathcal{K}_d$.
\begin{lemma}\label{lemma:shortintervalexistoutside}
    For every point $p_0$ on the boundary of $\mathcal{K}_d$, there exist an analytic family of directions $\delta_t$ and a small positive number $\epsilon$ such that $p_{\delta_t}(t)=(A+t\delta_{A_t},B+t\delta_{B_t})\notin\mathcal{K}_d$ over $0<t<\epsilon$. 
\end{lemma}
As $p_{\delta_t}(0)=p_0\in\mathcal{K}_d$ and $p_{\delta_t}(t)\notin \mathcal{K}_d$ for any $t>0$, we 
can characterize the boundary of $\mathcal{K}_d$ by studying 
the differences between $p_{\delta_t}(0)$ and $p_{\delta_t}(t)$ as $t\to 0$. 
To be specific, for the $\delta_t$ in Lemma \ref{lemma:shortintervalexistoutside}, define 
$\boldsymbol{\Psi}_t(\lambda)$ as the family of 
2$\times$2-minors $\{\Psi_{m,n;k,l}(\lambda)\}$ at $p_{\delta_t}(t)$; see \eqref{2minorquadraticequation}. Also,   
define $\mathrm{Z}_t$ as the zero set of 
$\boldsymbol{\Psi}_t(\lambda)$, i.e.,
\begin{eqnarray*}
Z_{t}=\left\{\lambda\in\mathbb{R}^{d}: 
\Psi_{m,n;k,l}(\lambda)=0 \mbox{ at } p_{\delta_t}(t) \mbox{ for all } m,n\in [p], k,l\in[q]\right\}\,.
\end{eqnarray*}

As $p_{\delta_t}(0)\in\mathcal{K}_d$, by Proposition~\ref{prop:algcond}, $\mathrm{Z}_0$ is the coordinate-axes \(\bigcup_{\ell}\mathbb R\cdot e_\ell\).  For $t>0$, note from \eqref{2minorquadraticequation} that 
$\boldsymbol{\Psi}_t(\lambda)$ is always a family of quadratic forms with diagonal part zero. Thus,  
$\mathrm{Z}_{t}$ still contains the coordinate axes. As 
$p_{\delta_t}(t)\notin \mathcal{K}_d$, Proposition~\ref{prop:algcond} implies that 
$\mathrm{Z}_{t}$ must contain some sets outside 
$\bigcup_{\ell}\mathbb R\cdot e_\ell$.
The following lemma describes the behavior of these additional zeros as $t\to 0$: 
\begin{lemma}\label{lemma:limitingbehaviorofzerolocusinboundary}
There exists an analytic family $\delta_t$ satisfying Lemma \ref{lemma:shortintervalexistoutside}, such that, for a sufficiently small number $\epsilon>0$ and a positive integer $m$, the zero set can be represented as 
$\mathrm{Z}_{t}=\mathrm{Z}_0\cup \bigcup^{m}_{j=1}\mathbb{R}\cdot v^{(j)}_t$ over $t\in(0,\epsilon)$, where each $v_{t}^{(j)}$ is a $d$-dimensional vector satisfying $\lim_{t\to 0}v^{(j)}_t=e_{\ell}$ for some $\ell\in[d]$. 

\end{lemma}
 
Geometrically, Lemma \ref{lemma:limitingbehaviorofzerolocusinboundary} 
indicates that as we go outside $\mathcal{K}_d$ from the boundary point $p_0$ and keep track of the zero sets 
$\mathrm{Z}_{t}$, we can find at least one new branch of zeros, say   $\mathbb{R}\cdot v_t$, bifurcating from one of the coordinate axes $\mathbb{R}e_{l}$. 
For an illustration of a bifurcating branch of zeros of polynomials, consider the family of polynomials
\(\Psi_{t}=\big(x(y+z),z[y+tx],0\big)^{\top}\). Clearly,  
the zeros of $\Psi_{t}$ is $\mathrm{Z}_{t}=
\bigcup_{\ell=1}^{3}\mathbb R\cdot e_\ell
\cup \mathbb{R}\cdot v_{t}$, with $v_{t}=(1,-t,t)$ bifurcating away from the axis $e_1=v_{0}$ as $t>0$ increases.

From the definition of $v_{t}$, we have $\lim_{t\to 0}v_{t}= e_{\ell}$ and $\Psi_t(c v_t)=0$ for any $c>0$ and  $t\in(0,\epsilon)$. Thus, we can take the {\sl normalized bifurcating direction}: $\omega_t=\frac1{v^\top_{t}e_{\ell}} v_{t}-e_{\ell}$ so that
\begin{eqnarray}\label{eq:omega.path}
\Psi_t(e_\ell+\omega_t)=0\,, \ \|\omega_t\|\neq 0\,, \ \lim_{t\to 0}\omega_t =0\,, \ 
\ {\rm and} \  \ \omega_t\perp e_\ell\,.
\end{eqnarray}
The sequence $\{e_{\ell}+\omega_t\}_{t>0}$ of zero directions forms a path perpendicularly approaching the zero locus $\mathbb{R}\cdot e_\ell\subset \mathrm{Z}_0$ from the outside of $\mathrm{Z}_0$. 
By expanding $\Psi_t$ around $t=0$ and evaluating it along the curve $e_\ell+w_t$, the following proposition, which follows readily from Lemmas 
\ref{lemma:shortintervalexistoutside} and \ref{lemma:limitingbehaviorofzerolocusinboundary}, 
derives a property about the boundary point $p_0=(A,B)$ in terms of the Jacobian matrix evaluated at $p_0$, 
\begin{eqnarray}\label{eq:submatrixfori}
J_{\ell}(p_0)=
\Big[\partial_{\lambda_{i}}\Psi_{0}(e_\ell)
\Big]_{i\in[d]\backslash{\{\ell\}}}
=
\Big[
\operatorname{vec}\big(\psi(a^\ell b^{\ell\top},a^ib^{i\top})\big)
\Big]_{i\in[d]\backslash{\{\ell\}}}\in\mathbb{R}^{p^2q^2\times (d-1)}
\,.
\end{eqnarray}

\begin{proposition}\label{prop:splittingoffleadingtorankdeficient}
    For any point $p_0$ on the boundary of $\mathcal{K}_d$, there exist some $\ell\in[d]$ and a normalized bifurcating direction $\omega_t$ such that $\lim_{t\to 0}\frac{\omega_t}{\|\omega_t\|}\in \ker J_{\ell}(p_0)$.
        
   \end{proposition}

Since $\lim_{t\to 0}\frac{\omega_t}{\|\omega_t\|}$ is non-degenerate,   
Proposition \ref{prop:splittingoffleadingtorankdeficient} implies that a necessary condition for $p_0$ to be a boundary point in $\mathcal{K}_d$ is that $\ker J_{\ell}(p_0)$ is non-empty, i.e., $J_{\ell}(p_0)$ is rank deficient for some $\ell\in[d]$.  
This condition is easily verifiable in view of the latter equality in \eqref{eq:submatrixfori}, which follows from straightforward algebra.

Finally, the following proposition establishes the converse of Proposition \ref{prop:splittingoffleadingtorankdeficient}, i.e., the rank deficiency of $J_{\ell}(p)$ indicates $p$ is a boundary point. Therefore, a boundary point is completely characterized by the rank deficiency of $J_{\ell}(p)$. 
\begin{proposition}\label{prop:geometric meaning of rank deficiency}
    Suppose that $p\in \mathcal{K}_d$ and $J_\ell(p)$ is rank-deficient for some $\ell\in[d]$. Then, there exists a family of vectors $\delta(s_1,s_2)=(s_1ue_\ell^{\top},s_2ve_\ell^{\top})$ with $u\in\mathbb{R}^{p}, v\in\mathbb{R}^{q}$, $(s_1,s_2)\in\mathbb{R}^2,\,s_1+s_2\neq0$, 
        such that $p_{\delta(s_1,s_2)}(t)$ and $p_{\delta(-s_2,-s_1)}(t)$ represent the same CP factor model at each $t>0$, where $p_{\delta}(t)=p+t\delta$ for $\delta\in\mathbb{R}^{p\times d}\times\mathbb{R}^{q\times d}$.  
\end{proposition}
In Proposition \ref{prop:geometric meaning of rank deficiency}, as $p_{\delta(s_1,s_2)}(t)$ and $p_{\delta(-s_2,-s_1)}(t)$ are not identifiable, they are outside $\mathcal{K}_d$. However, these two points can be arbitrarily close to $p\in\mathcal{K}_d$ as $t\to 0$. Thus, $p$ is on the boundary of $\mathcal{K}_d$. In conclusion, we can determine whether a point $p$ is in the boundary of $\mathcal{K}_d$ by checking if $J_{\ell}(p)$ is rank deficient  for some $\ell\in[d]$.

\section{Likelihood functions and asymptotic properties}\label{sec:likelihoodfunctions&asympthm}

This section proposes a quasi-maximum likelihood estimation (QMLE) procedure for the matrix CP-factor model, and investigates the asymptotic properties of the QMLE estimators. Similar to \cite{Xu03042025}, the QMLE procedure assumes the latent dimension $d$ is known. The estimation of $d$ has been considered in \cite{CHANGHEYANGYAO2021}; see also Section \ref{sec:numericalresult}.

\subsection{Quasi likelihood functions and the estimation method}
By vectorizing the matrix-variate time series \( Y_t \in \mathbb{R}^{p \times q} \) in \eqref{eq:cpfactmodel} as 
$y_t={\rm vec}(Y_t)$, 
we represent the CP-factor model as a vector factor model, 
\begin{equation}\label{eq:vectmodelforCP}
y_t=(B\odot A )x_t+\operatorname{vec}(e_t)\,,
\end{equation}
where $A\in\mathbb{R}^{p\times d}$, $B\in\mathbb{R}^{q\times d}$. The covariance matrix of $y_t$ is given as
\[
\Sigma_y(\theta)=(B\odot A) M_x (B\odot A)^\top+\Sigma_e\,,
\]
where $M_x={\rm Var}(x_t)\in\mathrm{S}_d^+$ is a positive definite matrix, $\Sigma_e={\rm Var}({\rm vec} (e_t))\in\mathcal{D}_{pq}$ is assumed to be a diagonal matrix, and $\theta=(A,B,M_x,\Sigma_e)$ is the parameter of interest.

By considering a working model of Gaussian distributed 
$\{x_t\}$ and $\{e_t\}$, and ignoring time dependence 
among $\{y_t\}$, we obtain the quasi likelihood function
\begin{equation}\label{eq:likelihoodforcpfact}
\begin{aligned}
\mathcal{L}_T(\theta)
&= -\frac{1}{2pq}\log|\Sigma_y(\theta)|
   -\frac{1}{2Tpq}\sum_{t=1}^T (y_t-\bar{y})^\top \Sigma_y^{-1}(\theta)(y_t-\bar{y}) \\
&= -\frac{1}{2pq}\log|\Sigma_y(\theta)|
   -\frac{1}{2pq}\operatorname{tr}\!\big(\widehat{M}_y\,\Sigma_y^{-1}(\theta)\big)\,,
\end{aligned}
\end{equation}
where $\widehat{M}_y=T^{-1}\sum_{t=1}^T (y_t-\bar{y})(y_t-\bar{y})^\top$. 
We scale the criterion by $pq$ for convenience, which does not affect the maximizer. The formulation of the quasi-likelihood in \eqref{eq:likelihoodforcpfact} is similar to that in \cite{BAILI2012aos} for vector time series, which offers substantial computational advantages over the full likelihood. The vectorization here is mainly a re-indexing of the entries. The matrix structure is maintained since the loading matrix is constrained as $B\odot A$ instead of a general $pq\times d$ dimensional matrix.

Since the quasi-likelihood is generally non-concave, its maximization requires a numerical optimization procedure. Moreover, \(M_x\) and \(\Sigma_e\) are required to be positive definite. Thus, we use the reparameterization \(M_x=LL^\top,\,\Sigma_e=\exp(S)\) where \(L\) is lower triangular with positive diagonal entries and
\(S=\operatorname{diag}(s_1,\ldots,s_{pq})\), and define
\[
    \widetilde{\mathcal L}_T(A,B,L,S)
    =
    \mathcal L_T(A,B,LL^\top,\exp(S))\,.
\]
We compute the quasi-maximum likelihood estimator   via maximizing $\widetilde{\mathcal{L}}_T$ as 
\begin{equation}\label{eq:def of estimator}
\widehat\theta
=
(\widehat A,\widehat B,\widehat M_x=\widehat{L}\widehat{L}^{\top},\widehat\Sigma_e=\exp\widehat{S})\,,\,(\widehat{A},\widehat{B},\widehat{L},\widehat{S})=\arg\max_{(A,B,L,S)}{\widetilde{\mathcal{L}}}_T\,, 
\end{equation}
using numerical methods such as the analytic-gradient L-BFGS-B algorithm \citep{BYRDLUNOCEDALZHU1995}. 
To alleviate the effect of non-concavity, we run the algorithm from multiple initial values and retain the solution with the largest quasi-likelihood. To further reduce the computational cost of high-dimensional matrix inversion \(\Sigma_y^{-1}\), we apply the Woodbury matrix identity \cite{HendersonSearle1981}:\begin{equation}\label{eq:inversematrix}
\Sigma_y^{-1}
=
\Sigma_e^{-1}
-
\Sigma_e^{-1}Z
(M_x^{-1}+Z^\top\Sigma_e^{-1}Z)^{-1}
Z^\top\Sigma_e^{-1}.
\end{equation}

 We remark that, the (generalized) EM algorithm has long been the conventional choice for QMLE in high-dimensional factor models, largely because of its favorable performance in vector factor models \citep{Engle01121981,DozGiannoneReichlin12,BAILI2012aos}. In those settings, the M-step often admits closed-form updates, making EM computationally convenient and numerically stable. These advantages, however, do not carry over directly to the CP-factor model, since the Khatri--Rao constraint \(Z=B\odot A\) precludes a joint closed-form update of \(A\) and \(B\). More importantly, our simulation results show that EM converges particularly slowly when the parameters lie on or near the boundary of \(\mathcal{K}_d\). We therefore use L-BFGS-B throughout. (Its computational efficiency in large-scale bound-constrained optimization has been studied in  \cite{Zhu1997Algorithm7L,Morales2011remarkonalg778}). Additional comparisons with the generalized EM algorithm are reported in Section~S.1 and Figure~S.1 of the Supplementary Materials.

\subsection{Asymptotic properties of the likelihood estimators}\label{sec:asympproperty}

This section studies the asymptotic properties of the proposed QMLE for the matrix CP factor model. 
Let the true parameter value be 
$\theta^{*}=(A^{*},B^{*},M_x^{*},\Sigma_e^{*})$, and 
$Z^{*}=B^{*}\odot A^{*}$. Take the strong-mixing coefficient \( \alpha_x(h)
 =\sup_{A\in\sigma(x_t:t\leq0),\,B\in\sigma(x_t:t\geq h)}
 \left|P(A\cap B)-P(A)P(B)\right|\). We impose the following assumptions. 

\Assumption\label{assump:full-ranklatentfactor}
    Assume that the latent factor $\{x_t\}$ is a stationary stochastic process satisfying $\mathbb{E}(x_t)=0$ and $\mathbb{E}(x_tx_t^{\top})={M}_x$, where ${M}_x$ is invertible and $\widehat{M}_x=\frac{1}{T}\sum_tx_tx_t^{\top}$ satisfies
    $\lim_{t\to\infty}\widehat{M}_x={M}_x$. 
    There exists a constant $\delta>0$ such that the following conditions hold. We require \(  E\|{x_t}\|^{4+\delta}<\infty\) and the $\alpha$-mixing condition: \(\sum_{h=1}^{\infty}\alpha_x(h)^{\delta/(4+\delta)}<\infty\).

    \EndAssumption

\Assumption\label{assump:wnbdedcov}
    Denote $e_{i;t}^{j}$ as the $(i,j)$-th entry of the 
    $p\times q$ white noise matrix $e_t$. 
    Assume that $\{e_{i;t}^{j},i,j,t\}$ are independent for all $i,j,t$ and $\{e_t\}$ is identically distributed with mean $\mathbb{E}(e_t) = 0$ and variance $\mathbb{E}[\operatorname{vec}(e_t)\operatorname{vec}(e_t)^{\top}]  = \operatorname{diag}((\sigma_{1}^{1})^2,\cdots, (\sigma_{i}^{j})^2, \cdots, (\sigma_{p}^{q})^2)$. 
    Also, $C_e^{-2}<(\sigma_{i}^{j})^2<C_e^2$ and $\sup_{i,j,t}E|e_{i;t}^j|^4\le C_e^4$ for every $i,j$ and some constant $C_e>0$.  
\EndAssumption
\Assumption\label{assump:boundedsetting}
    Assume that the rows of true loadings $A^*$ and $B^*$ are universally bounded, i.e., there exists a constant $C>0$, such that $\|a_i^*\|<C$ and $\|b_j^*\|<C$ hold for all $i,j$. Assume that the limit $\lim_{T\to \infty}\frac{1}{pq}Z^{*\top}\Sigma_{e}^*Z^*$ exists and is positive definite. In addition, the parameter $\theta^*$ belongs to a bounded set $\Theta$, and the $\ell_2$-norm of the columns of $A,B$ are $\sqrt{p}$ and $\sqrt{q}$, respectively. 
\EndAssumption

Assumptions \ref{assump:full-ranklatentfactor}-\ref{assump:boundedsetting} are commonly required in the existing literature \citep{CHANGDUHUANGYAO2024,HANYANGZHANGCHEN2024}. The full-rank condition in Assumption \ref{assump:full-ranklatentfactor} for the latent factor $X_t$ is necessary for model identification. The $\ell_2$-norm restriction in Assumption \ref{assump:boundedsetting} removes the column-wise rescaling ambiguity. The $4+\delta$-moment condition and $\alpha$-mixing condition are required for the purpose of $1/\sqrt{T}$-convergence rate of $\widehat{M}_x-M_x$.

The theoretical results in this paper does not impose restrictions on the relative growth rate of \(p, q\) with respect to \(T\), i.e., they can be fixed or divergent as \(T\to\infty\). In the high-dimensional setting, the dimension \(p q=p_T q_T\) is allowed to increase with \(T\), and the observed process should be viewed as a triangular array indexed by the sample size: for each \(T\), we observe a time series \(\{Y_{T,t}:1\leq t\leq T\}\), where each \(Y_{T,t}\) is a \(p_T\times q_T\) matrix. The corresponding loading matrices also form a sequence of high-dimensional matrices (or $p_T,q_T$-indexed triangular arrays): \(\mathcal{A}=(A_T:T\in\mathbb{Z}_+) ,\, A_T\in \mathbb{R}^{p_T \times d}\) (and similar for $\mathcal{B}$ and $B_T$), where \(d\) denotes the fixed number of latent factors. Then all asymptotic properties can be interpreted along these triangular arrays ($\{Y_{T,t}\},\mathcal{A},\mathcal{B}$). 

In particular, for such a sequence of high-dimensional loading matrices, it is often useful to consider the QR decomposition: \(A_T = A_{T,0} A_{T,r}\), where $A_{T,r}$ is an upper triangular matrix of fixed (or bounded by $d\times d$) dimension. Therefore, in the high-dimensional asymptotic analysis, the bounded condition on the identifiability can be equivalently expressed as boundedness conditions on the low-dimensional matrix.

\Assumption[Boundedness condition on ${\mathcal{K}}_d$]\label{assump:bdedidcond}\leavevmode\\
Assume that the point $(\frac{1}{\sqrt{p}}A_{r},\frac{1}{\sqrt{q}}B_{r})$ lies in a fixed compact subset of the interior $\operatorname{int}\mathcal{K}_d(d_1,d_2,d)$ or the boundary part and their minimal nonzero singular values are bounded below. 
\EndAssumption

When $(A,B)\in\operatorname{int}\mathcal{K}_d$, by equation \eqref{eq:idisstableunderQR}, Assumption~\ref{assump:bdedidcond} is equivalent to the requirement that the minimal singular value of the corresponding Jacobian matrices $J_{\ell}\,,\ell\in[d]$ has a positive lower bound. This condition is consistent with the boundedness assumption (of the eigenvalues of $\Omega$) imposed in \cite{CHANGDUHUANGYAO2024} because $J_\ell$ consists of particular columns of $\Omega(B\odot A)$ and the boundedness of the eigenvalues of $\Omega$ (away from zero) implies that of each $J_\ell$.

For the CP factor model \eqref{eq:cpfactmodel} with $(A,B)\in\mathcal{K}_d$ and its log-likelihood function $\mathcal{L}_T$ \eqref{eq:likelihoodforcpfact}, the following results give the consistency and convergence rates of the QMLE. 

\begin{theorem}\label{thmforcpconvrate}
    Assume the true parameters $(A^*,B^*)$ are identifiable. Under Assumptions \ref{assump:full-ranklatentfactor}--\ref{assump:bdedidcond} and $(A^*,B^*)\in\mathcal{K}_d$, as $T\to \infty$, we have
     \begin{align}
    \frac{1}{pq}\left\| \operatorname{diag}(\widehat{\Sigma}_{e} \Sigma_{e}^{*-1} - I_{pq}) \right\|^2 = o_p(1),&\quad
    \|\widehat{M}_x - M_x^*\|^2 = o_p(1), \\
    \frac{1}{p}\|\widehat{A} - A^*\|^2 = o_p(1), \quad &\frac{1}{q}\|\widehat{B} - B^*\|^2 = o_p(1)\,.\label{eq:A.hat.B.hat}
    \end{align} 
    Furthermore, in the interior case, i.e., $(A^*,B^*)\in \operatorname{int}\mathcal{K}_d$, we have
    \begin{align}
    \frac{1}{pq}\left\| \operatorname{diag}(\widehat{\Sigma}_{e} \Sigma_{e}^{*-1} - I_{pq}) \right\|^2 = O_p(T^{-1}), &\quad
    \|\widehat{M}_x - M_x^*\|^2 = O_p(T^{-1}), \\
    \frac{1}{p}\|\widehat{A} - A^*\|^2 = O_p(T^{-1}), \quad &\frac{1}{q}\|\widehat{B} - B^*\|^2 = O_p(T^{-1}).
    \end{align}

\end{theorem}

When the true parameters $(A^{*},B^{*})$ are in the interior of $\mathcal{K}_d$, the estimator $(\widehat{A}_r,\widehat{B}_r)$ will asymptotically lie in the interior of $\mathcal{K}_d$, and thus is identifiable for sufficiently large $T$. However, for the boundary case, $(\widehat{A},\widehat{B})$ can be outside $\mathcal{K}_d$, and thus may not be uniquely identified (stated in Proposition~\ref{prop:geometric meaning of rank deficiency}). 
Nevertheless, \eqref{eq:A.hat.B.hat} guarantees that 
all possible $(\widehat{A},\widehat{B})$ are close to the true value. 
See Example S.2 in the Supplementary Materials. 

Under additional assumptions on $x_t$ and $e_t$, the QMLE method produces a sharper rate.  
\begin{theorem}\label{thm:sharper rate}
    Assume $x_t$ is a jointly sub-Gaussian process, $e_t\overset{\mathrm{i.i.d}}{\sim}N(0,\Sigma_e^*)$, and $(A^*,B^*)\in\operatorname{int}\mathcal{K}_d$. Under Assumptions \ref{assump:full-ranklatentfactor}--\ref{assump:bdedidcond}, as $p,q,T\to \infty$ with $\log (pq)=o(T)$, we have 
    \begin{equation}
        \frac{1}{pq}\|\widehat{B}\odot \widehat A-B^*\odot A^*\|_F^2=O_p\left(\frac{p+q}{pqT}\right).
    \end{equation}
\end{theorem}

\begin{remark}
    For QMLE in the CP factor model, the convergence analysis cannot directly follow \cite{BAILI2012aos} due to the nonlinear coupling between the loading matrices $A$ and $B$. In particular, the loading parameters $A,B$, $M_x$ and $\Sigma_e$ in the gradient equations (see (S.30a)-(S.30d) in the supplementary material) are no longer as cleanly separated as in equations \cite{BAILI2012aos}[(2.7)--(2.9)], while the Khatri--Rao structure also introduces non-affine identifiability constraints. To address these difficulties, we develop a projection-based argument that removes high-dimensional nuisance components and isolates the quantities of interest, thereby disentangling the interactions among these parameters.
\end{remark}
\begin{remark}
The simultaneous diagonalization method of \cite{CHANGDUHUANGYAO2024} requires the true parameter $(A^*,B^*)$ to lie in $\mathcal{K}_\Omega$ and achieves the rate $O_p\{s\log(pq)/T\}$ under the sparsity condition $s\log(pq)/T=o(1)$. In contrast, our quasi-likelihood approach applies to all identifiable models, namely any $(A^*,B^*)\in\mathcal{K}_d$, where $\mathcal{K}_d$ strictly contains $\mathcal{K}_\Omega$. Even on $\mathcal{K}_\Omega$, Theorem~\ref{thmforcpconvrate} gives the faster rate $O_p(1/T)$, independent of dimension growth. The QMLE is also applicable in fixed dimensions, unlike the eigen-based simultaneous diagonalization method. In sub-Gaussian settings, the ISO method \cite{HANYANGZHANGCHEN2024,CHEN2026106167} achieves a sharper convergence rate of $\frac{1}{pq}\|\widehat{Z}-Z^*\|_F^2=O_p\left(\frac{p+q}{pqT}\right)$ when $(A^*,B^*)$ is a pair of full-rank loading matrices. Theorem \ref{thm:sharper rate} shows that the QMLE method still achieves the same convergence rate even in the rank-deficient case. 
\end{remark}

\section{Simulation study}\label{sec:simulationdesign}

In this section, we conduct two sets of simulation studies to investigate the performance of the QMLE for CP factor models at interior parameter points and on the boundary. We also describe the numerical designs underlying these experiments.

\subsection{Numerical Design}

We generate data from the CP-factor model \eqref{eq:cpfactmodel} with parameters $(A,B,x_t,e_t)$ for pre-specified $p,q,d,d_1,d_2$ as follows. To generate loading matrices $(A,B)$, we first simulate $\tilde A \in \mathbb{R}^{p \times d}$ and $\tilde B \in \mathbb{R}^{q \times d}$ 
with entries independently drawn from $N[0,3]$. 
\begin{enumerate}
\item[i)] To generate $(A,B)$ in the interior of $\mathcal{K}_d$, 
we conduct singular value decompositions 
\(
\tilde A = U_A \Lambda_A V_A^{\top}\,, \,\,
\tilde B = U_B \Lambda_B V_B^{\top}\,,
\) and extract the first $d_1$ and $d_2$ left singular vectors from $U_A$ and $U_B$, respectively, as 
$U_A^{(d_1)}$ and $U_B^{(d_2)}$. 
Then, define
\(
A' = U_A^{(d_1)}U_A^{(d_1)\top} \tilde{A}\,, \,\,
B' = U_B^{(d_2)}U_B^{(d_2)\top} \tilde{B}
\). 
Finally, normalize the columns of 
$A',B'$:   
\(a^j = \sqrt{p}\,
\frac{a'^{j}}{\|a'^{j}\|_2}\,,\,\,
b^{j} = \sqrt{q}\,
\frac{b'^{j}}{\|b'^{j}\|_2}
\) to form the $j$-th column of $A,B$. 
Then,
\(\operatorname{rank}(A) = d_1,\,\operatorname{rank}(B) = d_2\). According to \cite{CHANGDUHUANGYAO2024}, we take $d_1=d_2=d-1,d\geq3$ to ensure that the simulated $(A,B)$ lies in $\mathcal{K}_\Omega\subset \operatorname{int}\mathcal{K}_d$.  
To compare the sharper convergence rate between QMLE and ISO methods, we also generate full-rank $(A,B)$ in a similar way.  

\item[ii)] 
To generate $(A,B)$ on the boundary of $\mathcal{K}_d$, 
for simplicity, we follow Example \ref{examplenumeric} and consider $d=5$. 
We conduct QR decompositions 
\(\tilde A = Q_A R_A\,, \,\,\tilde B = Q_B R_B,\,\, Q_A\in\mathbb{R}^{p\times d},\ Q_B\in\mathbb{R}^{q\times d}\). 
Next, take the point $(A,B)$ given in \eqref{exampleforA,Binbdry}, and normalize their column vectors to be unit vectors: \begin{align*}
   A_r=A\cdot \operatorname{diag}\left(1, {6}^{-\frac{1}{2}},2^{-\frac{1}{2}},6^{-\frac{1}{2}},{1} \right),\, B_r=B\cdot \operatorname{diag}\left(2,3^{-\frac{1}{2}},35^{-\frac{1}{2}},11^{-\frac{1}{2}},2\right)\,.  
\end{align*} 
Finally, set 
\(
A = \sqrt{p}Q_A A_r,
\quad
B = \sqrt{q}\, Q_B B_r.
\)
The direct computation of $J_1$ shows that its rank is $3$ with the direction of the kernel $\delta=(1,1,1,1)$. By Proposition \ref{prop:geometric meaning of rank deficiency}, the point $(A,B)$ lies on the boundary $\mathcal{K}_d\setminus\operatorname{int}\mathcal{K}_d$. 
\end{enumerate}

The latent vector process $\{x_t\}$ is generated from a stationary VARMA(1,1) model of dimension $d$ with independent normally distributed innovations. Specifically, we randomly generate the autoregressive and moving-average coefficient matrices with a spectral radius strictly smaller than $0.95$ to ensure stationarity. The innovation covariance matrix is diagonal with entries randomly chosen from $\{1,2\}$. After discarding the first 100 burn-in observations, we retain $T$ observations and rescale the factors by a constant multiplier.

The noise matrices $e_t$ is generated from $\operatorname{vec}(e_t) \sim N(0, \Sigma_{e})$, where 
\[\Sigma_{e} = \operatorname{diag}(\sigma_1^2, \dots, \sigma_{pq}^2),\] and each variance $\sigma_i^2$ is independently drawn from the uniform distribution on $[1,2]$. 
The noise sequences are independent across $t$.

\subsection{Estimation procedure.}

We apply the gradient-based L-BFGS-B algorithm, which is implemented in \texttt{stats::optim()} in R, to obtain a more accurately converged local optimum in the simulation study.

Since the loading matrices \(A\) and \(B\) are identifiable only up to a permutation of components and sign changes, after optimization, we align the raw estimators \(\widehat A\) and \(\widehat B\) to the true ones as follows.
Denote \(A=[a_1,\ldots,a_{d}]\in\mathbb{R}^{p\times d}\) and \(\widehat A=[\hat a_1,\ldots,\hat a_{d}]\in\mathbb{R}^{p\times d}\). We choose a permutation \(\pi_A\) of \(\{1,\cdots,d\}\) that maximizes the total absolute column-wise agreement, \(\pi_A
\;=\;
\arg\max_{\pi\in\mathcal{S}_{d}}
\sum_{k=1}^{d}\big|\langle \hat a_k,\ a_{\pi(k)}\rangle\big|
\), where \(\mathcal{S}_{d}\) is the set of permutations and \(\langle\cdot,\cdot\rangle\) denotes the Euclidean inner product. Given \(\pi_A\), we set the sign for each matched pair by
\(
s_{A,k}=\mathrm{sign}\,\!\big(\langle \hat a_k,\ a_{\pi_A(k)}\rangle\big)\in\{-1,+1\},
\)
and define the aligned estimator
\[
\widehat A^{\,\mathrm{al}}
=
\big[s_{A,1}\hat a_1,\ \ldots,\ s_{A,d}\hat a_{d}\big]\,P_{\pi_A},
\]
where \(P_{\pi_A}\) is the permutation matrix corresponding to \(\pi_A\).
We perform the same alignment for \(B\) by computing \(\pi_B\), \(\{s_{B,k}\}\), and \(\widehat B^{\,\mathrm{al}}\).

\subsection{Simulation results}

This subsection reports simulation results generated according to the design described in Section~\ref{sec:simulationdesign}. We compare the finite-sample performance of the proposed QMLE procedure with the CP-unified method (the eigen-based approach proposed by \cite{CHANGDUHUANGYAO2024}) and the Iteratively Simultaneous Orthogonalization method proposed by \cite{CHEN2026106167}. To summarize the accuracy, we report the root-mean-squared-error between the estimators and the true values: 
\begin{equation}
\begin{aligned}
\mathrm{RMSE}(\widehat{A}^{\,\mathrm{al}})&=\sqrt{\frac{1}{pd}\|\widehat{A}^{\,\mathrm{al}}-A^*\|^2}\,,\quad A\,,\widehat{A}^{\,\mathrm{al}}\in \mathbb{R}^{p\times d}\,,
\end{aligned}
\end{equation}
and $\mathrm{RMSE}(\widehat B^{\,\mathrm{al}})$ is defined analogously.

\begin{table}[tb]
\centering
\renewcommand{\arraystretch}{1.05}
\setlength{\tabcolsep}{4pt}

\caption{Performance of QMLE, CPunified, and ISO in the
\textbf{interior} of $\tilde{\mathcal{K}}_d$.}
\label{tab:freq-interior}

\par

\begin{adjustbox}{max width=\linewidth}
\begin{threeparttable}
\begin{tabular}{ll*{3}{c}c}
\toprule
& & \multicolumn{3}{c}{
Mean $(\operatorname{RMSE}(\widehat A^{\mathrm{al}}),
       \operatorname{RMSE}(\widehat B^{\mathrm{al}}))$}
& Failed proportions\\
\cmidrule(lr){3-5}
Method & $(p,q)$ & $T=100$ & $T=200$ & $T=400$
& $(T=100,200,400)$ \\
\midrule

\multirow{6}{*}{QMLE}
& $(20,20)$
& $(4.32,4.39)$ & $(2.97,3.06)$ & $(2.12,2.14)$
& --- \\
& $(20,60)$
& $(2.43,4.19)$ & $(1.71,2.95)$ & $(1.21,2.07)$
& --- \\
& $(20,100)$
& $(1.89,4.16)$ & $(1.33,2.93)$ & $(0.93,2.06)$
& --- \\
& $(60,60)$
& $(2.38,2.41)$ & $(1.67,1.69)$ & $(1.18,1.19)$
& --- \\
& $(60,100)$
& $(1.85,2.39)$ & $(1.29,1.68)$ & $(0.91,1.18)$
& --- \\
& $(100,100)$
& $(1.86,1.85)$ & $(1.30,1.30)$ & $(0.91,0.91)$
& --- \\

\midrule

\multirow{6}{*}{CPunified}
& $(20,20)$
& $(9.76,9.79)$ & $(8.25,8.32)$ & $(6.72,6.74)$
& --- \\
& $(20,60)$
& $(5.67,9.75)$ & $(4.75,8.09)$ & $(3.81,6.62)$
& --- \\
& $(20,100)$
& $(4.37,9.58)$ & $(3.58,8.12)$ & $(2.95,6.60)$
& --- \\
& $(60,60)$
& $(5.58,5.61)$ & $(4.62,4.64)$ & $(3.79,3.79)$
& --- \\
& $(60,100)$
& $(4.29,5.54)$ & $(3.51,4.62)$ & $(2.98,3.84)$
& --- \\
& $(100,100)$
& $(4.27,4.26)$ & $(3.60,3.57)$ & $(2.93,2.88)$
& --- \\

\midrule

\multirow{6}{*}{ISO}
& $(20,20)$
& $(696.53,690.43)$ & $(693.42,697.66)$ & $(694.59,695.57)$
& $(0.63,0.63,0.63)$ \\
& $(20,60)$
& $(692.54,702.75)$ & $(690.88,676.84)$ & $(686.35,689.07)$
& $(0.63,0.55,0.65)$ \\
& $(20,100)$
& $(729.29,724.21)$ & $(740.50,759.75)$ & $(743.73,769.18)$
& $(0.67,0.74,0.75)$ \\
& $(60,60)$
& $(678.25,706.13)$ & $(686.12,686.26)$ & $(660.46,657.72)$
& $(0.64,0.58,0.59)$ \\
& $(60,100)$
& $(709.10,725.35)$ & $(695.81,701.74)$ & $(677.64,685.08)$
& $(0.71,0.67,0.57)$ \\
& $(100,100)$
& $(628.01,673.94)$ & $(690.67,711.11)$ & $(715.42,702.17)$
& $(0.7,0.68,0.63)$ \\

\bottomrule
\end{tabular}

\begin{tablenotes}[flushleft]\footnotesize
\item[] \textit{Notes:}
All RMSE values are multiplied by $1000$.
The last column gives failure proportions $(n_{\operatorname{failed}}/200)$ for $T=100,200,400$,
respectively, each out of 200 replications.
For ISO, failed replications are excluded from the reported means,
which are calculated over the remaining $200-n_{\mathrm{failed}}$
replications. 
\end{tablenotes}
\end{threeparttable}
\end{adjustbox}

\par
\end{table}

\begin{table}[tb]
\centering
\renewcommand{\arraystretch}{0.95}
\setlength{\tabcolsep}{4pt}

\caption{Performance of QMLE, CPunified, and ISO in the
\textbf{full-rank} case.}
\label{tab:freq-fullrank}
\par

\begin{adjustbox}{max width=\linewidth}
\begin{threeparttable}
\begin{tabular}{ll*{3}{c}}
\toprule
& & \multicolumn{3}{c}{
Mean $(\operatorname{RMSE}(\widehat A^{\mathrm{al}}),
       \operatorname{RMSE}(\widehat B^{\mathrm{al}}))$} \\
\cmidrule(lr){3-5}
Method & $(p,q)$ & $T=100$ & $T=200$ & $T=400$ \\
\midrule

\multirow{6}{*}{QMLE}
& $(20,20)$   & $(4.50,4.54)$ & $(3.11,3.16)$ & $(2.22,2.23)$ \\
& $(20,60)$   & $(2.56,4.57)$ & $(1.80,3.21)$ & $(1.27,2.26)$ \\
& $(20,100)$  & $(1.99,4.58)$ & $(1.40,3.21)$ & $(0.98,2.26)$ \\
& $(60,60)$   & $(2.60,2.62)$ & $(1.83,1.84)$ & $(1.29,1.29)$ \\
& $(60,100)$  & $(2.02,2.62)$ & $(1.42,1.84)$ & $(1.00,1.30)$ \\
& $(100,100)$ & $(2.04,2.02)$ & $(1.42,1.42)$ & $(1.00,1.00)$ \\

\midrule

\multirow{6}{*}{CPunified}
& $(20,20)$   & $(10.64,10.56)$ & $(8.93,8.94)$ & $(7.37,7.31)$ \\
& $(20,60)$   & $(6.06,10.84)$  & $(5.09,9.16)$ & $(4.20,7.55)$ \\
& $(20,100)$  & $(4.61,10.84)$  & $(3.89,9.11)$ & $(3.20,7.59)$ \\
& $(60,60)$   & $(6.12,6.18)$   & $(5.18,5.19)$ & $(4.31,4.30)$ \\
& $(60,100)$  & $(4.78,6.17)$   & $(3.99,5.21)$ & $(3.31,4.28)$ \\
& $(100,100)$ & $(4.78,4.79)$   & $(4.00,4.01)$ & $(3.31,3.33)$ \\

\midrule

\multirow{6}{*}{ISO}
& $(20,20)$   & $(5.04,5.08)$ & $(3.50,3.54)$ & $(2.51,2.50)$ \\
& $(20,60)$   & $(2.66,5.17)$ & $(1.87,3.63)$ & $(1.32,2.56)$ \\
& $(20,100)$  & $(2.03,5.18)$ & $(1.43,3.63)$ & $(1.00,2.56)$ \\
& $(60,60)$   & $(2.69,2.72)$ & $(1.89,1.90)$ & $(1.33,1.34)$ \\
& $(60,100)$  & $(2.06,2.72)$ & $(1.45,1.91)$ & $(1.02,1.34)$ \\
& $(100,100)$ & $(2.08,2.06)$ & $(1.45,1.45)$ & $(1.02,1.02)$ \\

\bottomrule
\end{tabular}

\begin{tablenotes}[flushleft]
\footnotesize
\item[] \textit{Notes:}
Entries in the three error columns are the arithmetic means of
replicate-specific RMSEs over 200 replications, multiplied by $1000$.
\end{tablenotes}

\end{threeparttable}
\end{adjustbox}
\par
\end{table}

\begin{table}[tb]
\renewcommand{\arraystretch}{0.95}
\caption{Performance of QMLE and CPunified on the \textbf{boundary} of \(\tilde{\mathcal{K}}_d\)\label{tab:freq-boundary}}
\tabcolsep=6pt
\centering
\begin{threeparttable}
\begin{adjustbox}{max width=0.8\linewidth}
\begin{tabular}{l l *{3}{c}}
\toprule
& & \multicolumn{3}{c}{mean value of \(\operatorname{RMSE}(\widehat{A}^{\mathrm{al}}),\operatorname{RMSE}(\widehat{B}^{\mathrm{al}})\)} \\
\cmidrule(lr){3-5}
& \((p,q)\) & \(T=100\) & \(T=200\) & \(T=400\) \\
\midrule
\multirow{6}{*}{QMLE}
& \((20,20)\) & \((40.83,56.43)\)   & \((24.65,35.40)\)   & \((10.23,16.48)\) \\
& \((20,60)\) & \((19.51,28.31)\)    & \((7.12,12.22)\)  & \((3.20,6.24)\) \\
& \((20,100)\) & \((12.23,19.01)\)   & \((4.20,7.98)\)     & \((2.49,5.25)\) \\
& \((60,60)\) & \((7.81,12.13)\)   & \((3.91,6.61)\)     & \((2.18,3.62)\) \\
& \((60,100)\) & \((4.08,7.02)\)   & \((2.55,4.46)\)   & \((1.66,2.85)\) \\
& \((100,100)\) & \((3.37,5.33)\)   & \((2.07,3.25)\)    & \((1.41,2.09)\) \\
\midrule
\multirow{6}{*}{CPunified}
& \((20,20)\) & \((160.57,319.99)\) & \((153.17,304.66)\) & \((157.55,327.12)\) \\
& \((20,60)\) & \((124.70,307.03)\) & \((153.69,313.58)\) & \((149.23,312.45)\) \\
& \((20,100)\) & \((149.30,304.27)\)& \((151.72,307.29)\) & \((154.66,310.75)\) \\
& \((60,60)\) & \((134.14,283.80)\) & \((134.29,290.46)\) & \((135.37,287.35)\) \\
& \((60,100)\) & \((142.02,295.89)\) & \((152.72,313.68)\)  & \((146.60,320.25)\) \\
& \((100,100)\) & \((157.01,315.13)\) & \((161.09,310.59)\) & \((160.38,318.71)\) \\
\bottomrule
\end{tabular}
\end{adjustbox}

\begin{tablenotes}\footnotesize
\item Note: All numbers are mean values of \(\operatorname{RMSE}\) based on 200 replications and multiplied by 1000.
\end{tablenotes}
\end{threeparttable}
\end{table}

Throughout Tables~\ref{tab:freq-interior}--\ref{tab:freq-boundary}, all simulations are implemented in \texttt{R}, and all reported values are multiplied by \(1000\). The simulation results for each $(p,q,T)$ are based on $200$ replications. 

A symmetry property is built into the data-generating process: the two loading matrices on the left and right play symmetric roles. Specifically, swapping \((p,q)\) is equivalent to swapping the two loading matrices (and their random seeds used to generate them), leaving the data-generating mechanism unchanged. Consequently, we only report configurations with \(p\le q\).

For both the interior (rank-deficient and full-rank cases) and boundary regimes, we consider \(p,q\in\{20,60,100\}\) with \(p\le q\), fix \(d=5\). In the interior case, we choose \((d_1,d_2)=(4,4)\) to highlight the behavior of QMLE in a rank-deficient setting. We examine sample sizes \(T\in\{100,200,400\}\). We report the results in Table~\ref{tab:freq-interior} (rank-deficient interior case), Table~\ref{tab:freq-fullrank} (full-rank case), and Table~\ref{tab:freq-boundary} (boundary case). Here, the failure count refers to replications in which the ISO iterations terminated due to a numerical safeguard because an estimated loading matrix became nearly singular. Specifically, termination was triggered when the ratio of its smallest to largest singular value was at most $10^{-12}$. Runs that merely reached the iteration limit without satisfying the stopping criterion are not counted as numerical failures.

In our experiments, interior points often stabilize after roughly 80--150 iterations, whereas boundary points frequently require at least about 250 iterations and can take around 500 iterations in small samples (especially when $p,q$ are small). Table \ref{tab:freq-boundary} suggests that the estimator remains practically well-behaved even near the boundary. For comparison, we apply the CP-unified method to the same simulated datasets. We emphasize that the CP-unified method is primarily designed for points inside \({\mathcal{K}}_{\Omega}\); for identifiable points outside this set, the implementation often issues non-convergence warnings on joint diagonalization. In the reported boundary case we fix \((d,d_1,d_2)=(5,4,4)\). We also experimented with various $d$'s and the displayed configuration is among the best-performing cases for CP-unified in our trials.

\section{Application to financial data analysis}\label{sec:numericalresult}

We consider two distinct firm characteristics: book-to-market equity (BE/ME), which reflects relative firm valuation, and market equity (ME), which reflects firm size. Then, stocks are classified separately into 10 groups according to each characteristic. Each combination of a BE/ME group and an ME group defines a portfolio, resulting in \(10 \times 10 = 100\) portfolios. Recording their monthly returns from January 1964 to December 2021 (\(T=696\)) yields a \(10 \times 10\) matrix time series: the ``100 Portfolios Formed on Size and Book-to-Market'' dataset in the Kenneth R. French Data Library \cite{FrenchDataLibrary}. To illustrate the proposed method, we use this matrix time series. To remove common market movements, we fit the CAPM \cite{FamaMacBeth1973} separately to each return series. 
The resulting residuals are collected into the market-adjusted return matrices \(R_t = \{r_{ij,t}\} \in \mathbb{R}^{10 \times 10},\, t=1,\ldots,T\). See Figure~S.3 in the Supplementary Materials for the time-series plots of $R_t$. The CAPM-filtered series 
appear stationary.

We formally assess stationarity entrywise using both ADF and KPSS tests across all \(pq=100\) entries. For the ADF test (null: unit root), we reject \(H_0\) for all entries; all Benjamini--Hochberg (BH--FDR) adjusted \(q\)-values are below \(0.05\). For the KPSS test (null: level stationarity), we fail to reject \(H_0\) for all entries. Since ADF and KPSS impose opposite null hypotheses, their unanimous agreement provides strong evidence that \(R_t\) is stationary.

\subsection{Factor dimension selection.}
 A commonly used approach in the literature to determine the latent factor dimension is the maximum eigenvalue-ratio criterion (which is also used in \cite{CHANGDUHUANGYAO2024}). As illustrated in Figure \ref{fig:parallel_analysis_95}, this rule most strongly supports $d=1$, with $d=2$ being the next most plausible candidate.  At the same time, if one adopts a more information-preserving perspective, we also apply the Bai-Ng information criterion (IC) \cite{BaiNg2002} to $\operatorname{vec}R_t$ with sample size $T=696$ and $N=100$. The Bai-Ng IC (in Figure \ref{fig:BaiNgIC}) yields: using the correlation (standardized) version, we obtain $d=3$, while the covariance (non-standardized) version favors $d=4$ (with the more permissive penalty placing the minimum around $d=5$). Taken together, these results support retaining a modest number of latent factors beyond the leading eigenvalue ratio, broadly in the range $d=3$--$5$. As an additional diagnostic, we apply parallel analysis (PA) \cite{Horn1965ParallelAnalysis} (red points in Figure \ref{fig:parallel_analysis_95}) to $R_t$, to assess which leading eigenvalues are clearly separated from a noise-based reference spectrum.
  
 We therefore treat this range as our main set of candidates and assess robustness accordingly.

\begin{figure}[tb]
\centering
\includegraphics[width=0.8\linewidth,
  height=0.4\linewidth]{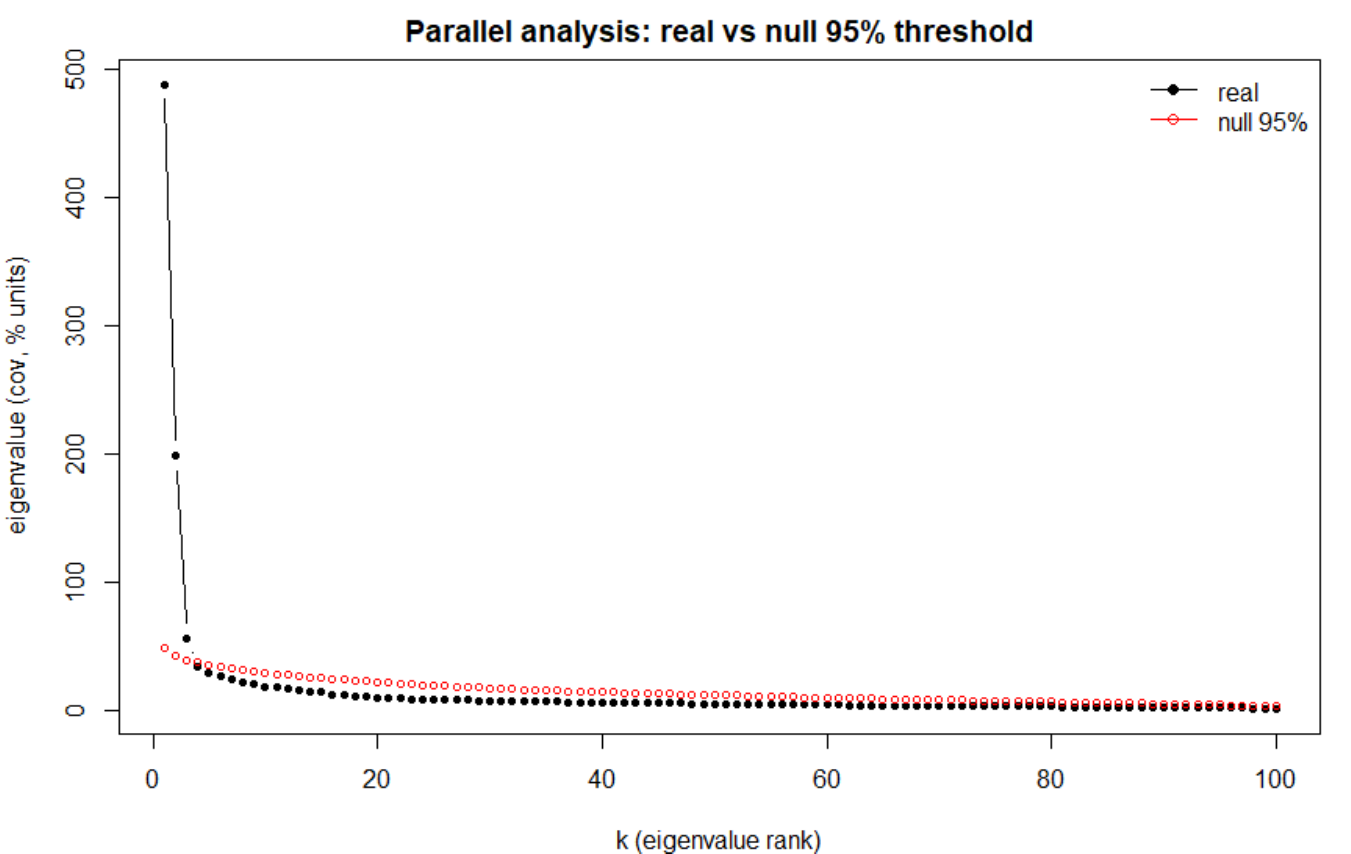}
\caption{Parallel analysis: empirical covariance eigenvalues (black) versus the 95\% null thresholds (red) obtained from permuting each column along the time dimension (with repetition $B=500$). }
\label{fig:parallel_analysis_95}
\end{figure}

\begin{figure}[tb]
    \centering
    \includegraphics[width=\linewidth]{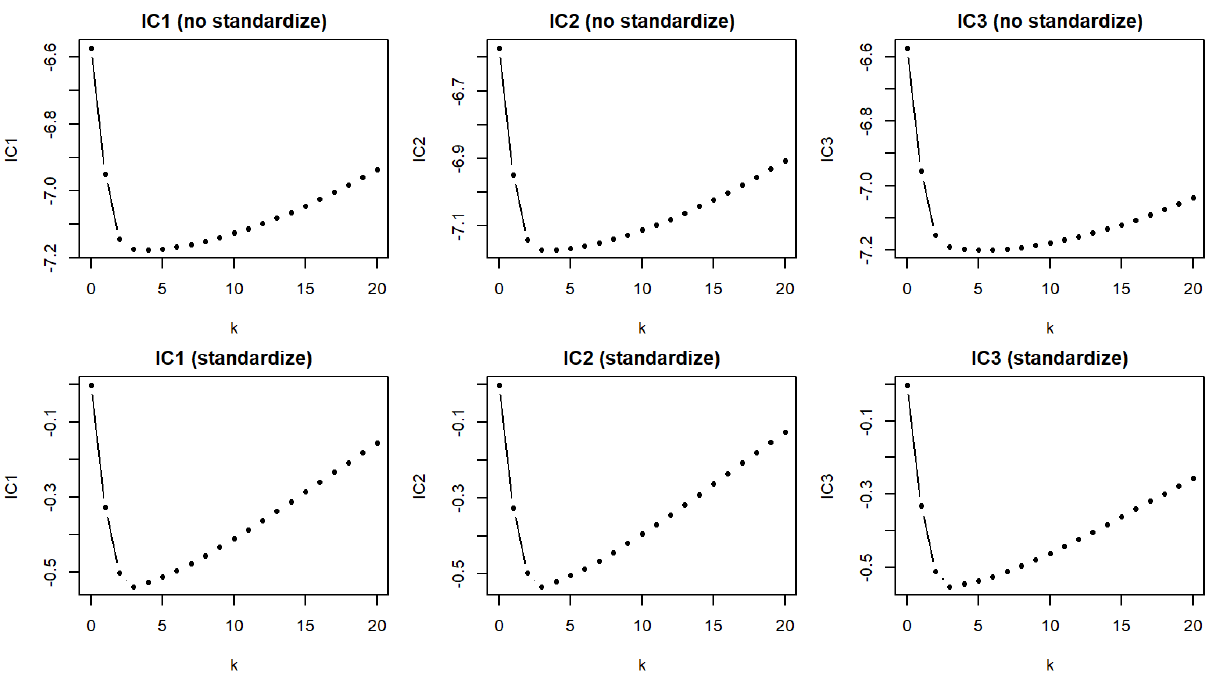}
    \caption{Bai-Ng information criteria.}
    \label{fig:BaiNgIC}
    \begin{tablenotes}\footnotesize
\item Note: $\mathrm{k}$ stands for the choice of latent dimension.
\end{tablenotes}
\end{figure}

\subsection{Latent factor extraction and forecasting evaluation.}
We fit a CP model to \(R_t\) and extract latent factors for each $d$ using both CP-unified and QMLE. We then model the dynamics of the estimated latent factor series via a VAR model, with the VAR order selected by AIC.

We evaluate out-of-sample predictive performance using a rolling-window scheme over the period 1964--2022. In each rolling iteration, we use a training window of length 456 observations and produce 1-step and 2-step ahead forecasts. Forecast accuracy is measured by the relative RMSE (rRMSE), defined by normalizing the entrywise prediction error using the entrywise variance estimated from the training window. Specifically, denote \(\widehat r_{ij,s+h|s}\) as the \(h\)-step ahead forecast for entry \((i,j)\) at time \(s\), and denote \(\widehat{\mathrm{Var}}_{s}(r_{ij,\cdot})\) as the sample variance of \(\{r_{ij,t}\}_{t=s-456+1}^{s}\), we compute
\[
\mathrm{rRMSE}_h(s)
=
\sqrt{
\frac{1}{pq}
\sum_{i=1}^{p}\sum_{j=1}^{q}
\frac{\big(r_{ij,s+h}-\widehat r_{ij,s+h|s}\big)^2}{\widehat{\mathrm{Var}}_{s}(r_{ij,\cdot})}
},
\qquad h\in\{1,2\}.
\]
We report the mean and standard deviation of the sample of \(\mathrm{rRMSE}_h(s)\) in the rolling windows for each method in Table~\ref{tab:rrmse_both}. 

\begin{table}[tb]
\caption{Out-of-sample rRMSE of market-adjusted returns (CAPM filtered), January 2002--December 2021.\label{tab:rrmse_both}}
\centering
\tabcolsep=6pt
\renewcommand{\arraystretch}{1.05}

\begin{adjustbox}{max width=\linewidth}
\begin{threeparttable}
\begin{tabular}{l cccc @{\qquad} l cccc}
\toprule
& \multicolumn{4}{c}{\textbf{(a) One-step ahead} (\(h=1\))} & & \multicolumn{4}{c}{\textbf{(b) Two-step ahead} (\(h=2\))} \\
\cmidrule(lr){2-5}\cmidrule(lr){7-10}
& \(d=2\) & \(d=3\) & \(d=4\) & \(d=5\) & & \(d=2\) & \(d=3\) & \(d=4\) & \(d=5\) \\
\midrule

\multicolumn{5}{l}{\textbf{Mean}} & \multicolumn{5}{l}{\textbf{Mean}} \\
\textsc{QMLE}      & 0.9201 & 0.9270 & 0.9198  & 0.9082  & \textsc{QMLE}      & 0.9185 & 0.9244 & 0.9152 & 0.9068 \\
\textsc{CPunified} & 0.9276 & 0.9339 & 0.94482 & 0.9579  & \textsc{CPunified} & 0.9260 & 0.9278 & 0.9434 & 0.9398 \\
\addlinespace[6pt]

\multicolumn{5}{l}{\textbf{SD}} & \multicolumn{5}{l}{\textbf{SD}} \\
\textsc{QMLE}      & 0.3785 & 0.3691 & 0.3794 & 0.3611 & \textsc{QMLE}      & 0.3780 & 0.3710 & 0.3792 & 0.3681 \\
\textsc{CPunified} & 0.3859 & 0.3917 & 0.3931 & 0.4306 & \textsc{CPunified} & 0.3851 & 0.3746 & 0.3852 & 0.3989 \\
\bottomrule
\end{tabular}

\begin{tablenotes}\footnotesize
\item Note: Values are out-of-sample rRMSE of market-adjusted returns (CAPM filtered) \\
for January 2002--December 2021.
\end{tablenotes}
\end{threeparttable}
\end{adjustbox}
\end{table}
\begin{figure}[tb]
\centering
\includegraphics[width=0.8\linewidth]{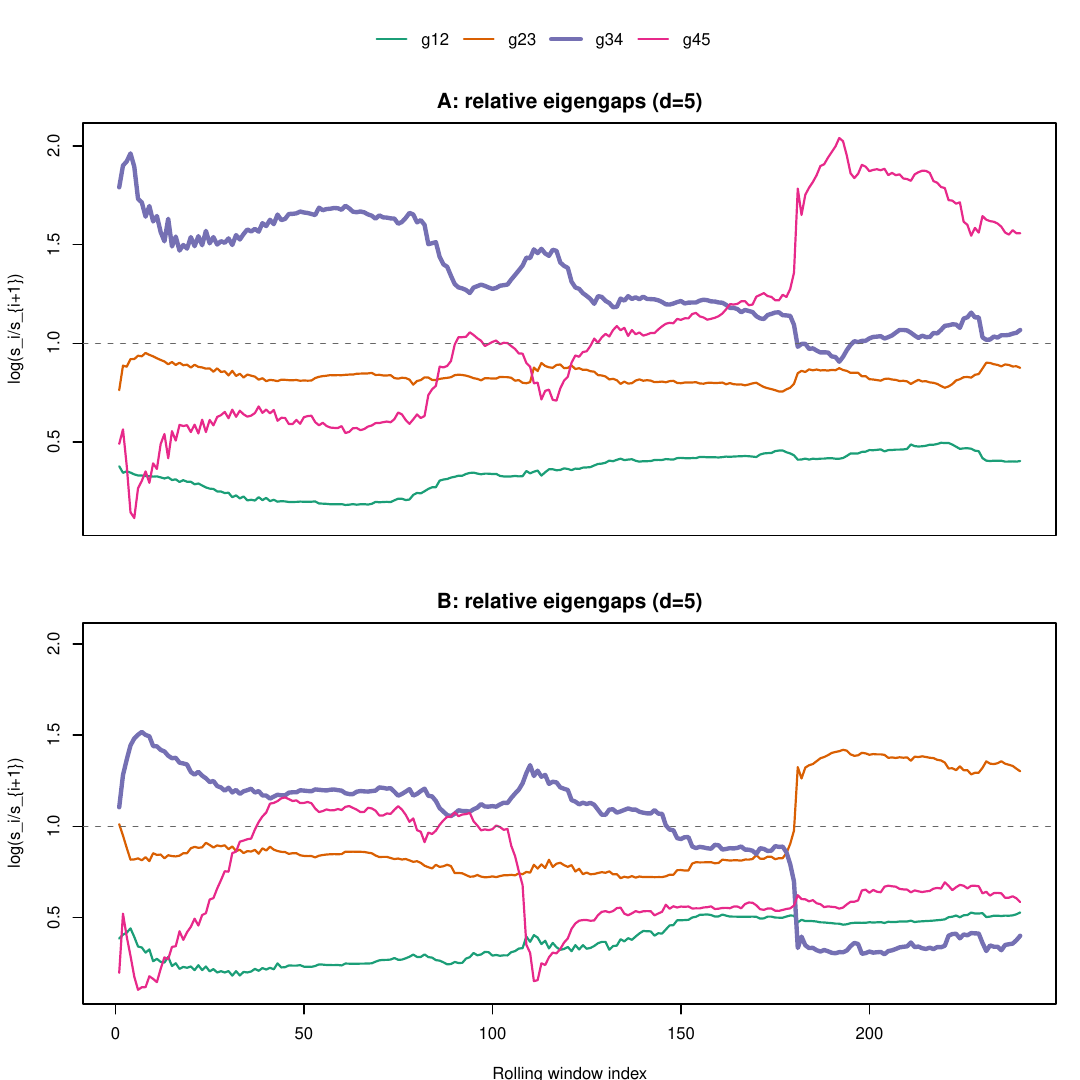}
\caption{Eigen-gaps for $d=5$: each row corresponds to the rolling window index $s\in[240]$.}
\label{fig:eigengaps}
\end{figure}
\begin{remark}[rank specification and numerical convergence for the CP-unified method]

For CP-unified, when the latent dimension is relatively large, the implementation also requires specifying the expected ranks \((d_1,d_2)\) of the loading matrices \(A\) and \(B\). We tried several rank-deficient choices and found that the out-of-sample rRMSEs were only marginally affected. However, for \(d=5\), the joint diagonalization step frequently produced non-convergence warnings, occurring in more than 50 out of 240 rolling windows. Thus, the CP-unified results for \(d=5\) should be interpreted with caution. Since QMLE does not require such rank inputs, we report CP-unified under the parsimonious symmetric setting \(d_1=d_2=d\) in Table~\ref{tab:rrmse_both}.
\end{remark}

\subsection{Further discussions on higher factor dimensions}
The rolling-window results show that \text{QMLE} consistently achieves smaller mean and SD of \(\mathrm{rRMSE}_h(s)\) than \text{CP-unified}. Moreover, the two methods exhibit different rank sensitivity: CP-unified performs best at relatively small ranks, particularly \(d=2\) or \(d=3\), whereas QMLE improves as \(d\) increases and attains its best performance at \(d=5\). Thus, the performance gap widens with \(d\), suggesting that QMLE benefits more from higher-factor-dimension specifications.

To understand the poorer stability of CP-unified at \(d=5\), we examine the QMLE estimates \(\{(\widehat A_s,\widehat B_s)\}_{s=1}^{240}\). The singular value profiles in Figure~\ref{fig:eigengaps} show a large eigengap between the third and fourth singular values for both \(\widehat A_s\) and \(\widehat B_s\), indicating that their effective ranks are often closer to \(3\) than \(5\). By the QR-based reduction in \eqref{eq:idisstableunderQR} and Section~\ref{subsec:algidcond}, this corresponds to studying \(A_r,B_r\in\mathbb{R}^{3\times 5}\). In this case, at most \(3\times 3=9\) non-redundant quadratic forms are available, fewer than the \(d(d-1)/2=10\) required when \(d=5\). Therefore, the sufficient identification condition for CP-unified fails in this effective rank-\(3\) regime, explaining its numerical instability.

Overall, whether under smaller latent dimensions or when we increase the latent dimension to enhance interpretability and to further reduce the 1- and 2-step-ahead rRMSE, QMLE exhibits consistently stronger robustness than CP-unified in our rolling-window experiments.


\subsection{Modeling data via a near-boundary CP-factor model}

According to Proposition~\ref{prop:geometric meaning of rank deficiency},
Theorem~\ref{thmforcpconvrate}, and the continuity of $J_\ell(A,B)$
around the true parameter $(A^*,B^*)$, we have the following result.

\begin{proposition}\label{prop:empirical-jacobian-diagnostic}
Suppose that $p,q,d$ are fixed and that $(A^*,B^*)\in\mathcal K_d$
satisfies the assumptions ensuring loading consistency in
Theorem~\ref{thmforcpconvrate}. For jointly aligned loading estimates,
define the squared smallest singular value
\[
\widehat\lambda_{\ell,\min}
:=\lambda_{\min}\!\left\{
J_\ell(\widehat A,\widehat B)^\top
J_\ell(\widehat A,\widehat B)\right\}
=\sigma_{\min}^2\!\left\{J_\ell(\widehat A,\widehat B)\right\}.
\]
If $(A^*,B^*)\in\mathcal K_d\setminus\operatorname{int}(\mathcal K_d)$,
there exists a fixed $\ell\in[d]$ such that
$\widehat\lambda_{\ell,\min}\xrightarrow{p}0$ as $T\to\infty$.
If the true parameter is a fixed point in
$\operatorname{int}(\mathcal K_d)$, there exists $c_*>0$ such that
\[
\mathbb P\!\left(
\min_{\ell\in[d]}\widehat\lambda_{\ell,\min}>c_*
\right)\longrightarrow1.
\]
\end{proposition}

The proof of Proposition~\ref{prop:geometric meaning of rank deficiency}
in the Supplementary Material associates a nontrivial vector
$c\in\ker J_\ell(A^*,B^*)$ with the matrix
\[
S_\ell=\sum_{i\ne\ell}c_i a_i^*(b_i^*)^\top.
\]
Under identifiability, this nonzero matrix must have rank greater
than one. We use this necessary condition to complement the
Jacobian diagnostic: besides reporting the components of an estimated
kernel direction, we inspect whether the resulting loading combination
is effectively more than rank one.

Write $\widehat J_\ell=J_\ell(\widehat A,\widehat B)$. To obtain an estimated kernel, let $\widehat Q_{\ell,+}$ contain orthonormal eigenvectors of $\widehat J_\ell^\top\widehat J_\ell$ corresponding to its $d-2$ largest eigenvalues. Projecting away the weakest singular component gives
\[
\widehat J_\ell^{(0)}
=\widehat J_\ell\widehat Q_{\ell,+}\widehat Q_{\ell,+}^{\top},
\qquad
\widehat c_\ell\in\ker\widehat J_\ell^{(0)},
\qquad \|\widehat c_\ell\|_2=1.
\]
When the smallest root is simple and the retained roots are positive, this kernel is one-dimensional, and $\widehat c_\ell$ is determined up to its overall sign. We report the component magnitudes
$|\widehat c_{\ell i}|$ and form
\[
\widehat S_\ell
=\sum_{i\ne\ell}\widehat c_{\ell i}
  \widehat a_i\widehat b_i^\top
\]
using the signed coefficients. 

These diagnostics provide a further examination of the $d=5$ fits in the preceding subsection. The common specification $d=5$ is supported by covariance-based Bai--Ng IC3 on the full sample \citep{BaiNg2002}; this does not mean that every window separately selects five factors.

Two types of behavior are of interest. In collision-associated configurations, several columns of $\widehat{Z}$ are nearly proportional, or some linear combination results in a nearly rank-1 matrix in $\mathrm{M}(\widehat A,\widehat B)$. Related degenerate behavior has been studied in CP approximation \citep{Stegeman2006,Stegeman2007}. These phenomena should not be equated with the identifiable boundary considered here. We examine the kernel components together with loading similarities and the matrix combination $\widehat S_{\ell}$ below. 

We revisit the seven candidate windows retained from the monthly scan. A representative is January 1965--December 2002, with $T=456$. For the refined estimate obtained in the expanding-sample analysis, the complete $\widehat J_4$ gives 
\[
\widehat\lambda_{4,\min}=0.00120702,
\qquad
\widehat c_4\approx
(-0.6249,\ 0.5494,\ -0.3558,\ 0.4255)^\top,
\]
where the coordinate order is $(1,2,3,5)$, since loading 4 is the base.


Using the same fitted loadings and the unrounded signed coefficients, the resulting matrix $\widehat S_4$ has leading singular values
\[
(s_1,s_2,s_3,s_4)
\approx(0.104807,\ 0.0271279,\ 0.00266636,\ 0.00144732),
\]
with the remaining singular values negligible at numerical precision.
The relative Frobenius errors of its best rank-one and rank-two approximations are, respectively,
\[
E_1=\frac{\bigl(\sum_{j>1}s_j^2\bigr)^{1/2}}{\|\widehat S_4\|_F}
\approx25.20\%,
\qquad
E_2=\frac{\bigl(\sum_{j>2}s_j^2\bigr)^{1/2}}{\|\widehat S_4\|_F}
\approx2.80\%.
\]
Thus a second singular component substantially improves the approximation, while the subsequent components are much smaller. This supports an approximate rank-two shape for the matrix generated
by the estimated kernel.

A complementary example with more stable loading and kernel directions is January 1968--December 1994. Here we fit CAPM once on this entire 324-month reference sample. We then freeze these residuals and
fit prefixes of $T=228,276,324$, ending in December 1986, 1990 and 1994, respectively. The covariance-based Bai--Ng IC1, IC2 and IC3 choices on the complete reference sample are $4$, $3$ and $5$; we use $d=5$, as selected by IC3 \citep{BaiNg2002}, at every prefix.

We first describe the estimated kernel direction of the complete, jointly matched $\widehat J_3$. Its four coordinates correspond to loadings $(1,2,4,5)$ with
\[
\widehat c_3\approx
(-0.1977,\ -0.7486,\ -0.0444,\ 0.6313)^\top.
\]
Panel (a) of Table~\ref{tab:boundary-reference-path} reports the component magnitudes $|\widehat c_{3i}|$ of the estimated kernel direction at different sample sizes. Its overall composition remains relatively stable as $T$ increases.

For each fixed fitted Jacobian, define
\[
\widehat L_k
=\min_{\substack{\|w\|_2=1\\\|w\|_0\le k}}
\|\widehat J_3w\|_2^2,
\qquad k=2,3,4,
\]
where $\|w\|_0$ counts the nonzero coordinates. Each restricted
minimum enumerates all allowed column subsets and reoptimizes the
unit coefficient vector within that subset; the fitted loadings
remain unchanged. Since $\widehat J_3$ has four non-base columns,
$\widehat L_4=\widehat\lambda_{3,\min}$. Panel (b) compares the minimum values of
$\|\widehat J_3w\|_2^2$ attainable when a unit coefficient vector $w$ has at most two, three, or four nonzero coordinates. This quantity measures the deviation from the kernel equation $\widehat J_3w=0$ for the original Jacobian, before projection. Allowing three participating directions substantially lowers the minimum compared with using at most two, while allowing all four gives a further reduction. Thus, the additional
directions contribute to attaining the full minimum
$\widehat L_4$. 

\begin{table}[tb]
\centering
\caption{Estimated kernel components and Jacobian roots for the
January 1968--December 1994 reference sample. CAPM is fitted once
on all 324 months. Panel (a) gives scalar component magnitudes;
panel (b) compares squared constraint residuals using the same
original Jacobian and unit coefficient norm.}
\label{tab:boundary-reference-path}
\small
\begin{tabular}{crrrr}
\toprule
\multicolumn{5}{l}{\textit{(a) Kernel component magnitudes}}\\
\addlinespace[2pt]
$T$ & $|\widehat c_{31}|$ & $|\widehat c_{32}|$
    & $|\widehat c_{34}|$ & $|\widehat c_{35}|$ \\
\midrule
228 & 0.2170 & 0.7321 & 0.2114 & 0.6101 \\
276 & 0.1673 & 0.7531 & 0.0330 & 0.6355 \\
324 & 0.1977 & 0.7486 & 0.0444 & 0.6313 \\
\bottomrule
\end{tabular}
\par\medskip
\begin{tabular}{crrrr}
\toprule
\multicolumn{5}{l}{\textit{(b) Weak roots and restricted residual ratios}}\\
\addlinespace[2pt]
$T$ & End year & $\widehat\lambda_{3,\min}$
    & $\widehat L_2/\widehat L_4$
    & $\widehat L_3/\widehat L_4$ \\
\midrule
228 & 1986 & 0.00117243 & 4.20 & 1.149 \\
276 & 1990 & 0.00048650 & 22.84 & 1.027 \\
324 & 1994 & 0.00067539 & 30.75 & 1.060 \\
\bottomrule
\end{tabular}
\end{table}

Here the residual $\|\widehat J_3w\|_2$ measures how closely a unit
coefficient vector $w$ satisfies the kernel condition
$\widehat J_3w=0$, or equivalently, the linearized rank-one
constraints. Panel (b) shows that allowing a third participating
direction substantially reduces this residual. The fourth direction
is needed to attain the full minimum $\widehat L_4$, although its
additional improvement is smaller. Thus, the smaller coefficients
cannot all be discarded without affecting the estimated near-kernel
relation.

We further examine the matrix generated by the signed kernel
coefficients,
\[
\widehat S_3
=\sum_{i\in\{1,2,4,5\}}
\widehat c_{3i}\widehat a_i\widehat b_i^\top.
\]
Table~\ref{tab:boundary-S3-spectrum} reports its four leading singular values $s_1\ge s_2\ge s_3\ge s_4$. To assess its approximate rank, we also report the relative Frobenius errors of the best rank-one and rank-two approximations \(E_r,r=1,2\). 

\begin{table}[tb]
\centering
\caption{Singular values of the signed loading combination
$\widehat S_3$ for the January 1968--December 1994 reference sample.
The last two columns report relative Frobenius approximation
errors in percent.}
\label{tab:boundary-S3-spectrum}
\small
\begin{tabular}{crrrrrr}
\toprule
$T$ & $s_1$ & $s_2$ & $s_3$ & $s_4$
    & $100E_1$ & $100E_2$ \\
\midrule
228 & 0.422164 & 0.0268655 & 0.0263857 & 0.0106668
    & 9.23 & 6.71 \\
276 & 0.170871 & 0.0208714 & 0.0147995 & 0.0036337
    & 14.95 & 8.82 \\
324 & 0.196074 & 0.0233587 & 0.0188992 & 0.0059467
    & 15.43 & 9.98 \\
\bottomrule
\end{tabular}
\end{table}

The first singular value dominates at all three sample sizes. However, the second and third singular values are comparable, so the second component does not clearly separate from the third. 

After aligning the loading columns across sample sizes, we compare the estimates of the same rank-one loading $\widehat a_3\widehat b_3^\top$ between $T=228$ and $T=276$, and between $T=276$ and $T=324$. The absolute cosine similarity exceeds $0.993$ in both comparisons. The corresponding unit kernel coefficient vectors have absolute inner products of approximately $0.9823$ and $0.9995$, respectively. Thus, both the estimated loading direction and its associated kernel direction remain relatively stable across these increasing sample sizes.

The two examples emphasize complementary aspects of identification for parameters near the boundary. The fixed 456-month estimate exhibits substantial participation by four directions and an approximate rank-two matrix combination. The historical example exhibits a more stable kernel direction, a first-to-last decrease in the Jacobian root, and an improvement requiring more than two directions. Together they motivate studying near non-identifiable and identifiable boundary neighborhoods. 

Taken together, these diagnostics provide complementary perspectives for assessing whether an estimated CP-factor model is consistent with the identifiable boundary of $\mathcal K_d$ or is dominated by proximity to non-identifiable configurations. The smallest Jacobian roots, the composition of the estimated kernel directions, and the rank structure of the associated loading combinations reveal aspects of the local geometry that column co-linearity checks or $\mathrm{rank}$ information alone do not fully capture.

This distinction is particularly relevant when retaining a richer latent structure: failure to satisfy a sufficient $k$-rank condition does not by itself establish non-identifiability or imply that the latent dimension must be reduced. Our analysis complements such criteria by providing additional geometric diagnostics for examining identification near the boundary and its stability across sample sizes.

\begin{acks}[Acknowledgments]

\noindent Department of Statistics and Data Science, The Chinese University of Hong Kong, Shatin, N.T., 999077, Hong Kong SAR, China.

\noindent\textit{Correspondence:} Chun Yip Yau, Department of Statistics and Data Science, The Chinese University of Hong Kong, Hong Kong SAR, China. \href{Email:cyyau@sta.cuhk.edu.hk}{cyyau@sta.cuhk.edu.hk}
\end{acks}

\begin{supplement}
\stitle{Supplementary Material for ``Identification problem and quasi-maximum likelihood estimation for matrix-variate CP-factor models''}
\sdescription{Technical details are provided in the supplementary materials.}
\end{supplement}

\bibliographystyle{spr-ims-nameyear}
\bibliography{reference}

\clearpage
\setcounter{section}{0}
\setcounter{subsection}{0}
\setcounter{subsubsection}{0}
\setcounter{equation}{0}
\setcounter{figure}{0}
\setcounter{table}{0}
\setcounter{theorem}{0}
\setcounter{example}{0}
\setcounter{remark}{0}
\setcounter{definition}{0}
\setcounter{assumption}{0}
\setcounter{proofcondition}{0}
\setcounter{algorithm}{0}
\renewcommand{\thesection}{S.\arabic{section}}
\renewcommand{\theequation}{S.\arabic{equation}}
\renewcommand{\thefigure}{S.\arabic{figure}}
\renewcommand{\thetable}{S.\arabic{table}}
\renewcommand{\thetheorem}{S.\arabic{theorem}}
\renewcommand{\theexample}{S.\arabic{example}}
\renewcommand{\theremark}{S.\arabic{remark}}
\renewcommand{\thealgorithm}{S.\arabic{algorithm}}
\renewcommand{\theHsection}{S.\arabic{section}}
\renewcommand{\theHsubsection}{\theHsection.\arabic{subsection}}
\renewcommand{\theHsubsubsection}{\theHsubsection.\arabic{subsubsection}}
\renewcommand{\theHequation}{S.\arabic{equation}}
\renewcommand{\theHfigure}{S.\arabic{figure}}
\renewcommand{\theHtable}{S.\arabic{table}}
\renewcommand{\theHtheorem}{S.\arabic{theorem}}
\providecommand{\theHexample}{}
\renewcommand{\theHexample}{S.\arabic{example}}
\providecommand{\theHremark}{}
\renewcommand{\theHremark}{S.\arabic{remark}}
\providecommand{\theHdefinition}{}
\renewcommand{\theHdefinition}{S.\arabic{definition}}
\providecommand{\theHassumption}{}
\renewcommand{\theHassumption}{S.\arabic{assumption}}
\providecommand{\theHproofcondition}{}
\renewcommand{\theHproofcondition}{S.\arabic{proofcondition}}
\providecommand{\theHalgorithm}{}
\renewcommand{\theHalgorithm}{S.\arabic{algorithm}}
\allowdisplaybreaks[1]
\runtitle{Supplement to identification and QMLE}
\runauthor{H. Ye and C. Y. Yau}
\makeatletter
\markboth{\@runauthor}{\@runtitle}
\thispagestyle{copyright}
\makeatother
\pdfbookmark[0]{Supplementary material}{arxiv.supplement}
\addtocontents{toc}{\protect\SupplementTocStart}

\makeatletter
{\title@style\title@size\title@case{Supplementary Material for ``Identification problem and quasi-maximum likelihood estimation for matrix-variate CP-factor models in high dimension''}\par}
\vskip\title@skip
{\author@style\author@size By Hanzi Ye and Chun Yip Yau\par}
\makeatother
\begin{abstract}
This supplement provides the EM-type algorithm, an identification example,
and proofs of the identification and likelihood results. References to
numbers without the prefix S refer to the main article. The supplementary
conditions used in the likelihood proofs are stated explicitly.
\end{abstract}
\makeatletter
\abstract@fmt
\frontmatter@cmd
\makeatother

\begingroup
\let\ArxivSavedContentsLine\contentsline
\renewcommand\contentsline[4]{}
\renewcommand\SupplementTocStart{\let\contentsline\ArxivSavedContentsLine}
\tableofcontents
\endgroup

\section{EM-type Algorithm}\label{sec:EMalg}

Let \(y_t=\operatorname{vec}(Y_t)\), \(Z=B\odot A\), and
\[
    \widehat M_y=\frac{1}{T}\sum_{t=1}^T (y_t-\bar y)(y_t-\bar y)^{\top}.
\]
At iteration \(m\), define
\[
\begin{split}
G^{(m)}
&=
\left\{
(M_x^{(m)})^{-1}
+(Z^{(m)})^\top(\Sigma_e^{(m)})^{-1}Z^{(m)}
\right\}^{-1},\\
\mu_t^{(m)}
&=
G^{(m)}(Z^{(m)})^\top(\Sigma_e^{(m)})^{-1}(y_t-\bar y),\\
H_t^{(m)}
&=
G^{(m)}+\mu_t^{(m)}(\mu_t^{(m)})^\top,
\end{split}
\]
where \(Z^{(m)}=B^{(m)}\odot A^{(m)}\). Further, let
\[
C_{yx}^{(m)}
=
\frac{1}{T}\sum_{t=1}^T (y_t-\bar y)(\mu_t^{(m)})^\top,
\qquad
C_{xx}^{(m)}
=
\frac{1}{T}\sum_{t=1}^T H_t^{(m)},
\]
and, for any \(Z\), define
\[
\mathcal R^{(m)}(Z)
=
\widehat M_y-C_{yx}^{(m)}Z^\top
-Z(C_{yx}^{(m)})^\top
+ZC_{xx}^{(m)}Z^\top.
\]

\begin{algorithm}[t]
\caption{Generalized EM Algorithm for the CP Factor Model}
\label{Salg:CP-GEM}
\begin{algorithmic}[1]

\Statex \textbf{Initialization:}
Choose
\(\theta^{(0)}
=(A^{(0)},B^{(0)},M_x^{(0)},\Sigma_e^{(0)})\in\Theta_N\),
a tolerance \(\varepsilon>0\), and set \(m=0\).

\Statex \textbf{Step 1: E-step.}
Using \(\theta^{(m)}\), compute
\(G^{(m)}\), \(\mu_t^{(m)}\), \(H_t^{(m)}\),
\(C_{yx}^{(m)}\), and \(C_{xx}^{(m)}\) as defined above.

\Statex \textbf{Step 2: Update the loading matrices.}
Numerically increase
\[
Q_{AB}^{(m)}(A,B)
=
-\frac{T}{2}
\operatorname{tr}\left\{
(\Sigma_e^{(m)})^{-1}
\mathcal R^{(m)}(B\odot A)
\right\}
\]
using a quasi-Newton or block-alternating method. Obtain a feasible pair $(A^{(m+1)},B^{(m+1)})$ satisfying
the prescribed loading normalization \(\frac{1}{p}\operatorname{diag}(A^{(m+1)\top}A^{(m+1)})=\frac{1}{q}\operatorname{diag}(B^{(m+1)\top}B^{(m+1)})=\mathbf 1_d\) and the product-Gram constraint in \eqref{supp:qmle-domain}, and denote the
result by \((A^{(m+1)},B^{(m+1)})\), with $Z^{(m+1)}=B^{(m+1)}\odot A^{(m+1)}$. 

\Statex \textbf{Step 3: Update the covariance matrices.}
Set
\[
\begin{split}
M_x^{(m+1)}
&=U\operatorname{Diag}\{[\lambda_\ell]_{[m_-,m_+]}\}U^\top,
\quad C_{xx}^{(m)}=U\operatorname{Diag}(\lambda_\ell)U^\top,\\
(\Sigma_e^{(m+1)})_{rr}
&=[\mathcal R^{(m)}(Z^{(m+1)})_{rr}]_{[d_-,d_+]}.
\end{split}
\]

Here $[z]_{[a,b]}=\min\{b,\max\{a,z\}\}$ denotes clipping to the covariance bounds in \eqref{supp:qmle-domain}.

\Statex \textbf{Step 4: Check convergence.}
Let
\(\theta^{(m+1)}
=(A^{(m+1)},B^{(m+1)},M_x^{(m+1)},\Sigma_e^{(m+1)})\).
Stop if
\[
\left|
\mathcal L_T(\theta^{(m+1)})
-
\mathcal L_T(\theta^{(m)})
\right|<\varepsilon.
\]
Otherwise, set \(m\leftarrow m+1\) and return to Step~1.

\end{algorithmic}
\end{algorithm}

Algorithm~\ref{Salg:CP-GEM} adapts the standard EM construction
for Gaussian factor analysis \cite{EMalgforML}; see also \cite[Section~8]{BAILI2012aos} for its implementation in high-dimensional factor models.
All conditional expectations are evaluated under the working
Gaussian model that ignores temporal dependence, consistently
with the quasi-likelihood criterion.
The CP constraint $Z=B\odot A$ is imposed in the loading update,
which is performed subject to the prescribed normalization.
Provided that Step~2 does not decrease $Q_{AB}^{(m)}$,
the algorithm does not decrease the EM auxiliary function: the clipped
updates in Step~3 maximize the two covariance blocks over their prescribed boxes.
The standard GEM ascent property
\cite[Theorem~1]{EMalgforML} therefore gives
\[
\mathcal L_T(\theta^{(m+1)})
\ge
\mathcal L_T(\theta^{(m)}).
\]
The routine derivation is omitted.

\begin{figure}[t]
    \centering
    \includegraphics[width=\linewidth]{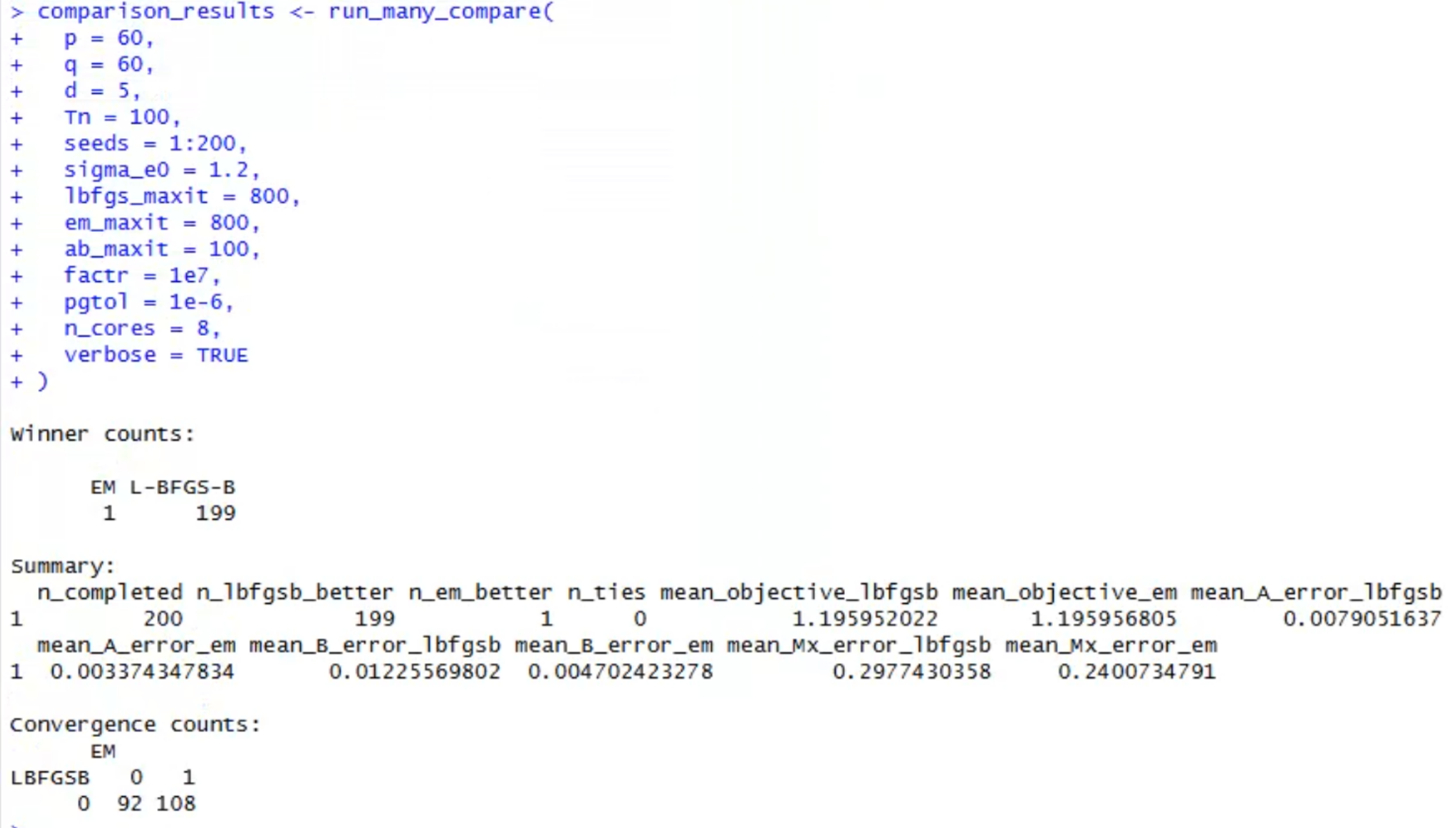}
    \caption{Boundary-case: Comparison for (60,60,100)}
    \label{fig:comparison results}
\end{figure}

We observed a substantial difference in convergence speed between the two algorithms. In the boundary-case simulation with \(p=q=60\), \(T=100\), and 200 replications, the EM algorithm frequently failed to reduce the relative change in the objective function below \(10^{-10}\), even after 800 iterations. In contrast, the L-BFGS-B algorithm converged within 500 iterations in all replications and attained a higher quasi-likelihood in 199 out of 200 cases. Figure \ref{fig:comparison results} shows the experiment result. 

To examine whether the computational difference persists at smaller observation dimensions, we repeated the boundary-case experiment with $p=q=20$, $d=5$, and $T\in\{100,200,400\}$, using 200 paired replications for each sample size. Both algorithms used the same oracle initialization, fixed parameter bounds, and relative objective-change \[\Delta_m
=
\frac{
\mathcal{L}_T(\theta^{(m)})-\mathcal{L}_T(\theta^{(m-1)})
}{
\max\left\{
1,\,
|\mathcal{L}_T(\theta^{(m-1)})|,\,
|\mathcal{L}_T(\theta^{(m)})|
\right\}
}\] tolerance of $10^{-10}$. L-BFGS-B was allocated a budget of 500 iterations, whereas generalized EM was evaluated at budgets of 500, 800, and 1000 outer iterations. As shown in Table~\ref{tab:lowdim-em-lbfgsb}, L-BFGS-B satisfied the stopping criterion in 144, 163, and 166 replications for $T=100,200,400$, respectively. The corresponding counts for EM were only 11, 25, and 50 at a budget of 500 iterations, increasing to 14, 28, and 59 at a budget of 1000 iterations. The mean recorded runtime of L-BFGS-B was 0.84--1.78 seconds, whereas that of EM at the 1000-iteration budget was 10.34--13.86 seconds, approximately 7.8--12.4 times longer at the corresponding sample size. Despite this larger budget, EM attained a lower quasi-likelihood than L-BFGS-B in 187, 161, and 134 of the 200 replications, respectively. These results support the greater computational efficiency of L-BFGS-B for optimizing the working Gaussian quasi-likelihood in this experimental design.
Moreover, EM yielded smaller mean loading RMSEs despite attaining
lower quasi-likelihoods on average. Under oracle initialization,
this pattern is consistent with incomplete optimization retaining
more information from the true starting values.
The comparison therefore concerns optimization performance,
not a general advantage in loading-estimation accuracy.

\begin{table}[!t]
\centering
\renewcommand{\arraystretch}{1.12}
\setlength{\tabcolsep}{3.5pt}

\caption{Optimization performance of L-BFGS-B and generalized EM
in the boundary case with $p=q=20$ and $d=5$.}
\label{tab:lowdim-em-lbfgsb}
\par

\begin{adjustbox}{max width=\linewidth}
\begin{threeparttable}
\begin{tabular}{clrrrccc}
\toprule
$T$
& \shortstack{Method\\(budget)}
& \shortstack{Stopping\\test met}
& \shortstack{Mean\\iterations}
& \shortstack{Mean\\time (s)}
& \shortstack{$10^5\overline{\Delta}_f$\\(MCSE)}
& \shortstack{L-BFGS-B\\wins}
& \shortstack{Mean RMSE\\$(A,B)\times10^3$} \\
\midrule

\multirow{4}{*}{100}
& L-BFGS-B (500) & $144/200$ & 382.0 & 0.84
& --- & --- & $(15.96,15.67)$ \\
& EM (500) & $11/200$ & 483.1 & 5.37
& $3.74\ (0.28)$ & $190/200$ & $(7.09,6.42)$ \\
& EM (800) & $12/200$ & 766.0 & 8.39
& $2.82\ (0.24)$ & $190/200$ & $(8.38,7.80)$ \\
& EM (1000) & $14/200$ & 953.3 & 10.34
& $2.29\ (0.21)$ & $187/200$ & $(9.29,8.76)$ \\

\midrule

\multirow{4}{*}{200}
& L-BFGS-B (500) & $163/200$ & 337.1 & 1.21
& --- & --- & $(11.30,10.99)$ \\
& EM (500) & $25/200$ & 465.1 & 5.75
& $1.55\ (0.14)$ & $167/200$ & $(4.59,4.02)$ \\
& EM (800) & $27/200$ & 725.4 & 8.90
& $1.35\ (0.13)$ & $162/200$ & $(5.07,4.54)$ \\
& EM (1000) & $28/200$ & 897.8 & 10.96
& $1.21\ (0.13)$ & $161/200$ & $(5.44,4.94)$ \\

\midrule

\multirow{4}{*}{400}
& L-BFGS-B (500) & $166/200$ & 273.9 & 1.78
& --- & --- & $(7.85,7.57)$ \\
& EM (500) & $50/200$ & 425.2 & 7.62
& $0.67\ (0.07)$ & $143/200$ & $(3.10,2.69)$ \\
& EM (800) & $58/200$ & 644.3 & 11.51
& $0.61\ (0.07)$ & $139/200$ & $(3.32,2.92)$ \\
& EM (1000) & $59/200$ & 786.2 & 13.86
& $0.57\ (0.06)$ & $134/200$ & $(3.49,3.11)$ \\

\bottomrule
\end{tabular}

\begin{tablenotes}[flushleft]
\footnotesize
\item[] \textit{Notes:}
Results are based on 200 paired replications per sample size, with
the same data, oracle initialization, and fixed parameter bounds
for both algorithms. The budget is the maximum permitted number of iterations. 
The stopping test requires a relative objective change $|\Delta_m|$ of at most $10^{-10}$.
The reported mean iterations are calculated over
all 200 replications, including runs that reached the iteration
budget without satisfying the stopping criterion. 

\item[]
For each EM budget,
$\Delta_\mathcal{L}=2[\mathcal{L}_T(\hat\theta_{\mathrm{LB}})-\mathcal{L}_T(\hat\theta_{\mathrm{EM}})]$ compares against L-BFGS-B with a budget of 500 iterations. Positive values favor L-BFGS-B. MCSE denotes the Monte Carlo standard error of the paired mean; both the mean difference and its MCSE are multiplied by $10^5$. ``L-BFGS-B wins'' counts replications with $\Delta_\mathcal{L}>10^{-10}$. 

\item[]
RMSEs are computed as
$\|\hat A^{\mathrm{al}}-A_0\|_{\mathrm F}/\sqrt{pd}$
and
$\|\hat B^{\mathrm{al}}-B_0\|_{\mathrm F}/\sqrt{qd}$
after joint column permutation, sign alignment, and normalization.
Their arithmetic means are multiplied by $1000$.
\end{tablenotes}

\end{threeparttable}
\end{adjustbox}
\par
\end{table}

\section{Example: A non-identifiable Sequence Converges to an identifiable limit}\label{example: convseqnonidentifiable}
We give an example in which two different sequences converge to the same boundary point, and the corresponding $n$-th elements of these sequences represent the same CP factor model. We take the numeric example 1: 

Let $(A,B)$ be the pair of matrices given in Example 1. We take $u=(0,1,0)^{\top}$ and $v=(0,2,0)^{\top}$ and consider the following two convergent sequences with limit $(A,B)$: 
\begin{equation}\label{example: sequenceofboundary}
\begin{aligned}\{(A_n,B_n):A_n=A,B_n=(b_1-\frac{1}{n}v,b_2,\ldots,b_5),n\in \mathbb{Z}_+\};\\
\{(A'_n,B'_n):A'_n=(a_1+\frac{1}{n}u,a_2,\ldots,a_5),B'_n=B,n\in \mathbb{Z}_+\}. \end{aligned}\end{equation}

To see both sequences living in $\mathbb{R}^{3\times 5}\times \mathbb{R}^{3\times 5 }\setminus{\mathcal{K}}_5$, let us look at the following linear combination given by $\lambda_n=(1,\frac{1}{n},\frac{1}{n},\frac{1}{n},\frac{1}{n})^{\top}\in \mathbb{R}^5$: \begin{equation}\begin{aligned}
B_n\odot A_n\cdot \lambda_n=b_1\otimes(a_1+\frac{1}{n}u)=[B'_n\odot A'_n]\cdot e_1
\end{aligned}
\end{equation}which is just the first column of $B'_n\odot A'_n$. By proposition 1, $(A_n,B_n,x_{n,t})$ and $(A'_n,B'_n,x'_{n,t})$ represent the same CP-factor model, where $x_{n,t}=(\lambda_n,e_2,\cdots,e_5)x_{n,t}'\in\mathbb{R}^{5}$. 

\begin{proposition}\label{prop:example-boundary-inner-directions}
Consider the pair $(A,B)$ in Example~1, and let
\(u=(0,1,0)^\top,\quad
v=(0,2,0)^\top\). 
Let $z_j=b_j\otimes a_j$, $j\in[5]$. Then: 
The path
    \[
    A_{\rm bd}(t)=A-tu e_1^\top,\qquad
    B_{\rm bd}(t)=B+tv e_1^\top
    \]
    remains on the boundary of $\mathcal K_5$ for all sufficiently
    small $|t|$:
    \(
    \bigl(A_{\rm bd}(t),B_{\rm bd}(t)\bigr)
    \in
    \mathcal K_5\setminus\operatorname{int}(\mathcal K_5).
    \)
    In particular,
    \(\delta_{\rm bd}=(-u e_1^\top,\,v e_1^\top)\)
    is a boundary-tangent direction at $(A,B)$.

\end{proposition}

\begin{proof}
For the matrices in Example~1, direct calculation gives
\[
z_2+z_3+z_4+z_5
=
v\otimes a_1+b_1\otimes u.
\]
Denote this nonzero vector by
\[
z(s):=z_2+z_3+z_4+z_5.
\]
Along this path,
\[
a_1(t)=a_1-tu,\qquad
b_1(t)=b_1+tv,
\]
while all the remaining columns are unchanged. Therefore,
\begin{align*}
v\otimes a_1(t)+b_1(t)\otimes u
&=
v\otimes(a_1-tu)+(b_1+tv)\otimes u\\
&=
v\otimes a_1+b_1\otimes u\\
&=
z(s).
\end{align*}
Consequently,
\[
0\neq z(s)
\in
\bigl\{
b_1(t)\otimes x+y\otimes a_1(t):
x,y\in\mathbb R^3
\bigr\}
\cap
\operatorname{span}\{z_2,z_3,z_4,z_5\}.
\]
By Proposition~8, this is equivalent to
\[
\ker
\left(
\left.
\partial_\lambda\Psi_t(e_1)
\right|_{e_1^\perp}
\right)
\neq\{0\}.
\]
Hence
\[
\operatorname{rank}
\left(
\left.
\partial_\lambda\Psi_t(e_1)
\right|_{e_1^\perp}
\right)
<4
\]
for every sufficiently small $t$.

It remains to verify that the perturbed point is still identifiable.
Let
\[
Z_{\rm bd}(t)=B_{\rm bd}(t)\odot A_{\rm bd}(t),
\]
and let $\Psi_t(\lambda)$ denote the collection of all $2\times2$
minors of
\[
\operatorname{Mat}\{Z_{\rm bd}(t)\lambda\},
\qquad
\lambda=(\lambda_1,\ldots,\lambda_5)^\top.
\]
Over $\mathbb Q(t)$, a direct elimination of these minors gives the following
Gr\"obner basis for their polynomial ideal:
\[
\begin{aligned}
&\lambda_1(\lambda_2-\lambda_5),\qquad
\lambda_1(\lambda_3-\lambda_5),\qquad
\lambda_1(\lambda_4-\lambda_5),\qquad
\lambda_1\lambda_5^2,\\
&\lambda_2\lambda_3,\quad
\lambda_2\lambda_4,\quad
\lambda_2\lambda_5,\quad
\lambda_3\lambda_4,\quad
\lambda_3\lambda_5,\quad
\lambda_4\lambda_5.
\end{aligned}
\]
The rational elimination uses finitely many nonzero polynomial denominators.
By excluding their nonzero roots, this basis remains valid for all sufficiently
small nonzero $t$; at $t=0$, identification follows from Example~1.
It follows that 
\[
\{\lambda:\Psi_t(\lambda)=0\}
=
\bigcup_{j=1}^5\mathbb R e_j.
\]
Since $Z_{\rm bd}(0)$ has full column rank, $Z_{\rm bd}(t)$ remains
full column rank for all sufficiently small $|t|$. It follows from
Proposition~1 that
\[
(A_{\rm bd}(t),B_{\rm bd}(t))\in\mathcal K_5.
\]
Together with the rank deficiency above, we obtain
\[
(A_{\rm bd}(t),B_{\rm bd}(t))
\in
\mathcal K_5\setminus\operatorname{int}(\mathcal K_5).
\]
\end{proof}

This reveals that, when $s_1+s_2\neq 0$ ($\delta(s_1,s_2)\neq \delta(-s_2,-s_1)$) in Proposition~8, the tangent direction $\delta(s_1,s_2)$ goes outside, but when $s_1=-s_2=s$, the path $p_{\delta(s,-s)}$ would sometimes stay along the boundary part of $\mathcal{K}_d$. 
Since the ambient space is \(d(p+q)\)-dimensional, we project a neighborhood of the boundary point \(p_0\in \mathcal{K}_d\setminus\operatorname{int}\mathcal{K}_d\subset \mathbb{R}^{p\times d}\times \mathbb{R}^{q\times d}\) onto the Euclidean space spanned by tangent directions $\delta(1,0)$, $\delta(0,-1)$ and $\delta_{in}$ to provide a more intuitive visualization of the local geometry near \(p_0\) in Figure \ref{fig:localgeometry}, where \(\delta_{\mathrm{in}}\) is an arbitrary auxiliary direction, implemented here by adding a pair of full-rank matrix perturbations to \(A_0\) and \(B_0\) respectively. The red lines and yellow regions, except for the point $p_0$ (the blue line in $(a)$), form the complement $\mathcal{K}_d^\complement$, while the purple region denotes the interior. The figure is a schematic projection of these local paths, rather than a computed section of the full parameter space. The curves and regions are visually thickened for clarity, and their displayed thickness has no geometric significance.    
    
\begin{figure}[tb]
  \centering
    \includegraphics[width=0.75\linewidth]{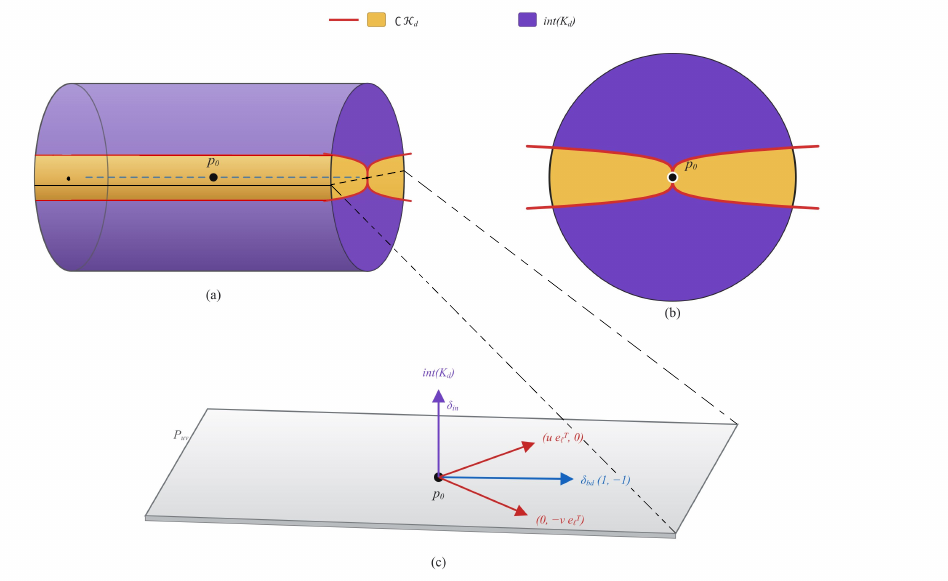}
    \caption{Projection of a small tube around $p_0$ in $\mathbb{R}^{p\times d}\times \mathbb{R}^{q\times d}$ onto the three directions \(\delta(1,0)\), \(\delta(0,-1)\), and \(\delta_{\mathrm{in}}\).}
    \label{fig:localgeometry}
\end{figure}

Figure~\ref{fig:market_adjusted_100portfolios} displays the market-adjusted return series for the 100 size--book-to-market portfolios.

\begin{figure}[t]
\centering
\includegraphics[width=\linewidth]{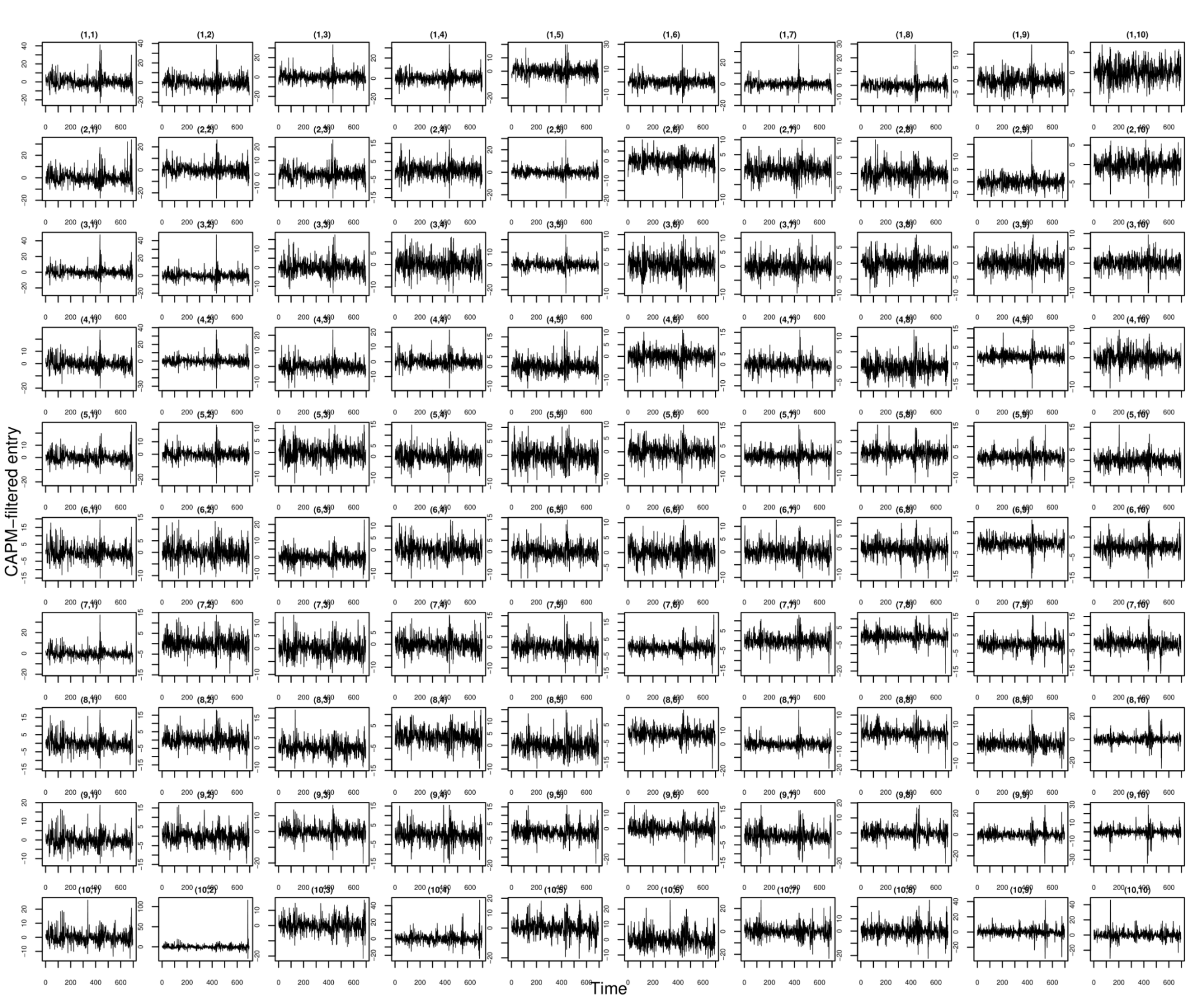}
\caption{Time-series plots of the market-adjusted return series for the 100 size--book-to-market portfolios (CAPM filtering). Rows correspond to the 10 BE/ME deciles, and columns correspond to the 10 size deciles.}
\label{fig:market_adjusted_100portfolios}
\end{figure}

\section{Proof of Propositions and Theorems related to the Identifiability Problem}\label{sec:geometry-proofs}

We briefly outline the organization of the proofs in this supplement.
We first prove Propositions in Section 2. Although part of Proposition~1 has been justified in the main text, the argument given there is not fully rigorous. We therefore provide a complete proof below. We then describe the topology of \(\mathcal{K}_d\), which serves as the theoretical prerequisite for our analysis of boundary versus interior points. However, since, in the problems of matrix space, one typically works with the subspace topology inherited from the ambient matrix space $\mathcal{K}_d\subset\mathbb{R}^{p\times d}\times \mathbb{R}^{q\times d}$, we do not emphasize this point in the main text. Then we prove Proposition~2 and Lemmas and Propositions~5-8. 

Then, the proof of Theorem~9 is presented. We defer the auxiliary lemmas for this result to Section \ref{sec:lemma}, along with an independent estimation procedure for the rate of $R$ in Section \ref{subsec:Est of R}. Finally, we prove the sharper rate under the Gaussian assumption in Section \ref{sec:gaussian-sharp-rate}. 

For notational convenience in the proofs, we will invoke the notion of an ideal in a polynomial ring; see, e.g., \cite{AnIntrotoCompuAGCA,atiyahmacCOMMUTATIVEALG}. 
In fact, no machinery from commutative algebra is needed beyond the definition: throughout, an ideal generated by an equation system $\Psi$ can be viewed simply as a collection of polynomials that is closed under addition and under multiplication by arbitrary polynomials. Equivalently, for polynomials \(\Psi_1,\dots,\Psi_m\), the ideal they generate consists of all polynomial combinations \(\sum_{k=1}^m h_k \Psi_k\) with polynomial coefficients \(h_k\).
We denote $I(B\odot A)$ as the ideal of $\Psi$ over the polynomial ring $\mathbb{R}[\lambda]$, i.e. $I(B\odot A)=\mathbb{R}[\lambda]\langle\Psi_{m,n;k,l},m,n\in[p],k,l\in[q]\rangle$

We take $K_d$ to be the subspace of $\mathbb{R}^{p\times d}\times \mathbb{R}^{q\times d}$ denoting for all the pairs $(A,B)$ such that $\operatorname{rank} B\odot A=d$. Let $\tilde{K}_d$ be the image of $K_d\subset \mathbb{R}^{p\times d}\times \mathbb{R}^{q\times d}$:  
\begin{align*}
    \odot:\mathbb{R}^{p\times d}\times \mathbb{R}^{q\times d}&\to \mathbb{R}^{pq\times d}\\
    (A,B)&\mapsto B\odot A. 
\end{align*}

If $b\otimes a=z\ne0$, then
$\{(a',b'):b'\otimes a'=z\}=\{(ca,c^{-1}b):c\in\mathbb R^*\}$:
the nonzero rank-one matrix $ab^\top$ determines its column and row spaces.
Thus the fiber above $Z=B\odot A\in\tilde K_d$ is
$\{(AD,BD^{-1}):D\in\mathcal D_d\}$, where $\mathcal D_d$ is the set of
invertible diagonal real matrices. This yields a local product description.
For each column matrix $M_\ell=\operatorname{Mat}(z_\ell)$, choose a nonzero
entry $(i_\ell,j_\ell)$. On the neighborhood where these entries remain nonzero, put
\[
 a_\ell^0=M_\ell[:,j_\ell],\qquad
 b_\ell^0=M_\ell[i_\ell,:]^\top/(M_\ell)_{i_\ell j_\ell}.
\]
Then $M_\ell=a_\ell^0b_\ell^{0\top}$, $b^0_{\ell,j_\ell}=1$, and every
factor pair there is uniquely $(A^0D^{-1},B^0D)$, with
$D=\operatorname{Diag}(b^1_{j_1},\ldots,b^d_{j_d})$. Hence the preimage of this
neighborhood is homeomorphic to its product with $\mathcal D_d$, and
$\odot$ is open. No global choice of normalized factors is needed.

\begin{proof}[proof of Proposition~1]

    In the main context, we proved that, every triple $(A',B',X')$ represents the same model $Y$ in (2) given by $(A,B,X)$ shares the same $\mathrm{M}(A,B)$. Every rank-1 basis $\mathrm{M}(A,B)=\operatorname{span}\{u_iv_i^{\top}, i\in[d]\}$ provides a transformation such that $(v_1\otimes u_1,\ldots,v_d\otimes u_d)^{\top}=H(B\odot A)^\top$ and produces a representation $$Y_t-e_t=(u_1,\ldots, u_d)\operatorname{Diag}(x'_t)(v_1,\ldots,v_d)^{\top}=AX_tB^{\top},x_t'=(H^{\top})^{-1}x_t.$$
    This suggests us to collect all base-change transformations that map one ordered rank-one basis $\{b^{1}\otimes a^{1},\cdots, b^{d}\otimes a^{d}\}$ of $\mathrm{col}(B\odot A)$ to another (where rank-one means the corresponding matrix $\operatorname{Mat}(z),z\in\mathrm{col}(B\odot A)$ is rank-one): \begin{equation}\label{eq:Gab}\mathcal{G}(B\odot A)\coloneqq\{U\in\mathrm{GL}_d:(B\odot A)U\in \tilde{K}_d\}.\end{equation}Each pair in $\{(A',B'):B'\odot A'=(B\odot A)U, U\in \mathcal{G}(B\odot A)\}$ represents the same CP-factor model (2). If the rank-1 basis of $\mathcal{S}(A,B)$ is unique up to scaling, the base-change transformations can only be permutations and scaling, and $\mathcal{G}(B\odot A)=\mathcal{D}_d\rtimes \mathcal{S}_d$. And if there exist a nonzero rank-1 tensor $v\otimes u\in\mathcal{S}(A,B)\setminus \cup _{\ell\in[d]}\mathbb{R}\cdot \{b^{\ell}\otimes a^{\ell}\}$, we can then extend it to a new basis using $\{b^{\ell}\otimes a^{\ell},\ell\in[d]\}$. The corresponding base-change does not belong to $\mathcal{D}_d\rtimes\mathcal{S}_d$ but still produces a new presentation. This proved equation~(3). 

\end{proof}

The identifiable region $\mathcal{K}_d$, by definition, is a subset of $K_d$. We take its topology to be the subspace topology via $\mathcal{K}_d\subset K_d\subset \mathbb{R}^{p\times d}\times \mathbb{R}^{q\times d}$. The preceding local product description identifies interior and boundary points of $\mathcal K_d$ with those of its image $\tilde{\mathcal K}_d=\odot(\mathcal K_d)$. Because $K_d$ is open in the ambient matrix-pair space, its relative interior agrees with the ambient interior. Throughout, an identifiable boundary point means a point of $\mathcal K_d\setminus\operatorname{int}\mathcal K_d$. 

Taking $\Omega(B\odot A)=\big[\operatorname{vec}\big(\psi(b^{i}\otimes a^i,b^j\otimes a^j)\big)\big]_{1\leq i\leq j\leq d}$ and $\Omega_{i<j}=\big[\operatorname{vec}\big(\psi(b^{i}\otimes a^i,b^j\otimes a^j)\big)\big]_{1\leq i< j\leq d}$, equation~(8) can be expressed as \begin{equation*}
    \psi\big(s(\lambda),s(\lambda)\big)=2\Omega_{i<j}(B\odot A){k}(\lambda),\quad k(\lambda)=(\lambda_1\lambda_2,\ldots, \lambda_1\lambda_d,\lambda_2\lambda_3,\ldots,\lambda_{d-1}\lambda_d)^\top. 
\end{equation*}

\begin{lemma}\label{lemma:omegafullrank}
    For any matrix $B\odot A\in \tilde{K}_d$, the following inequality holds: 
    \begin{equation}\label{eq:rankofOmegahasupperbound}
        \operatorname{rank}\big(\Omega(B\odot A)\big)\leq \frac{d(d-1)}{2}. 
    \end{equation}
\end{lemma}
\begin{proof}
    It is easy to see that $\Omega(B\odot A)$ has $\frac{d(d+1)}{2}$ columns. Because each tensor in the form $b\otimes a$ is rank-1, the columns indexed by $(i=j)$ are zero, i.e. $\psi(b^i\otimes a^i,b^i\otimes a^i)=0$. Only the $d(d-1)/2$ off-diagonal columns can contribute to the rank, which proves the inequality. 
\end{proof}

\begin{lemma}\label{lemma:Omegalambda}
    For each fixed pair $(p,q)$, the map $\Omega$ defined on $\tilde K_d$ can be naturally extended to be a continuous map defined on $\mathbb{R}^{pq\times d}$. The bi-linearity of $\psi$ ensures the following equation: 
    \begin{equation}
        \begin{aligned}
 \psi\big(v(\lambda),v(\lambda')\big)&=\Omega(M)\mathrm k(\lambda,\lambda'),\\
 v(\lambda)&=\sum_{i=1}^d\lambda_iv_i,\quad
 M=(v_1,\ldots,v_d)\in\mathbb R^{pq\times d}.
\end{aligned}
    \end{equation}
where $\mathrm{k}(\lambda,\lambda')$ is defined as: $$\mathrm{k}(\lambda,\lambda')=\begin{pmatrix}\lambda_1\lambda_1'\\\lambda_1\lambda'_2+\lambda_2\lambda_1'\\\vdots\\\lambda_{1}\lambda_d'+\lambda_1'\lambda_d\\\vdots\\\lambda_{d}\lambda_d'
        \end{pmatrix}.$$ Moreover, for a matrix $M$ and a transformation $U\in\mathrm{GL}_{d}$, we have 
\begin{equation}
    \Omega(MU)=\Omega(M)\cdot \mathrm{k}(U), \quad \text{where }\mathrm{k}(U)\coloneqq \big(\mathrm{k}(u_i,u_j):i\leq j\big)_{\frac{d(d+1)}{2}\times \frac{d(d+1)}{2}}.
\end{equation}
\end{lemma}
\begin{proof}
    For a linear combination $v({\lambda})=\sum_i\lambda_iv_i$ of $\{v_i,i\in[d]\}$ given by $\lambda=(\lambda_i)_{i\in[d]}$, its $2\times 2$ minor with index $(m,n,k,l)$ is \begin{align*}&v({\lambda})_{mk}v({\lambda})_{nl}-v({\lambda})_{nk}v({\lambda})_{ml}\\&=\sum_{i, j}\lambda_i\lambda_j[(v_i)_{mk}(v_j)_{nl}-(v_i)_{nk}(v_j)_{ml}]\\&=\sum_{i< j}2\lambda_i\lambda_j(\psi_{ij})_{m,n;k,l}+\sum_i\lambda_i^2\psi_{m,n;k,l}(v_i,v_i).\end{align*}
    Therefore, we take $\psi_{m,n;k,l}(v({\lambda}),v(\lambda))=v({\lambda})_{mk}v({\lambda})_{nl}-v({\lambda})_{nk}v({\lambda})_{ml}$. 
    Similarly, for two vectors $\lambda,\lambda'$, we take
\begin{align*}
 \psi_{m,n;k,l}\big(v(\lambda),v(\lambda')\big)
 &=\tfrac12\{v(\lambda)_{mk}v(\lambda')_{nl}
             -v(\lambda)_{nk}v(\lambda')_{ml}\}\\
 &\quad+\tfrac12\{v(\lambda')_{mk}v(\lambda)_{nl}
             -v(\lambda')_{nk}v(\lambda)_{ml}\}\\
 &=\tfrac12\sum_{i,j}\lambda_i\lambda_j'
       \{(v_i)_{mk}(v_j)_{nl}-(v_i)_{nk}(v_j)_{ml}\}\\
 &\quad+\tfrac12\sum_{i,j}\lambda_j\lambda_i'
       \{(v_i)_{mk}(v_j)_{nl}-(v_i)_{nk}(v_j)_{ml}\}\\
 &=\sum_{i,j}\lambda_i\lambda_j'\psi_{m,n;k,l}(v_i,v_j)\\
 &=\bigl[\Omega(M)\mathrm k(\lambda,\lambda')\bigr]_{m,n;k,l},
\end{align*}
    where the matrix $\Omega(M)$ is defined as  $$\Omega(M)=\bigg(\psi\big(v_i,v_j\big):i\leq j\in[d]\bigg).$$By taking $$\mathrm{k}(\lambda,\lambda')=(\lambda_1\lambda_1',\lambda_1\lambda'_2+\lambda_2\lambda_1',\cdots,\lambda_{1}\lambda_d'+\lambda_1'\lambda_d,\cdots,\lambda_{d}\lambda_d')^{\top},$$ it's easy to check that $$\psi\big(v(\lambda),v(\lambda')\big)=\Omega\big(M\big)\mathrm{k}(\lambda,\lambda').$$

    Because $\Omega$ is a polynomial map, it is continuous on $\mathbb{R}^{pq\times d}$. 

\end{proof}

In the following proof, we take $\psi^d(z)=\psi(\operatorname{Mat}(z),\operatorname{Mat}(z))$. 

\begin{proof}[proof of Proposition 2]

    Let $\lambda\in \mathbb{R}^{d}$ be a vector with $\|\lambda\|=1$ and $\lambda\notin\{\pm e_\ell:\ell\in[d]\}$. By Proposition 1, $(A,B)\in\mathcal{K}_d$ if and only if each tensor $s(\lambda)=\sum_{\ell=1}^d\lambda_{\ell}b^{\ell}\otimes a^{\ell}$ has its rank strictly more than one. Therefore, we have: 
    \begin{align*}
        (A,B)\in \mathcal{K}_d \Leftrightarrow \{\lambda\in\mathbb{R}^d\big|\|\lambda\|=1,\psi\big(s(\lambda),s(\lambda)\big)=0\}=\cup _{\ell=1}^d\pm e_{\ell}. 
    \end{align*}
    The solution set of equation $\psi\big(s(\lambda),s(\lambda)\big)=0$ equals that of the equation system: 
    \begin{align*}
        \Psi_{m,n;k,l}(\lambda)\coloneqq\psi_{m,n;k,l}\big(s(\lambda),s(\lambda)\big)=0,\forall m,n\in[p],k,l\in[q].
    \end{align*}
    Real Hilbert's Nullstellensatz theorem states that, two polynomial systems share the same zero set, if and only if the corresponding ideals have the same real radical. Therefore, the equation (12) holds. 
\end{proof}

\begin{proof}[proof of Corollary~3]

   Assume $B\odot A\in \tilde{\mathcal{K}}_d$ is a matrix such that $\operatorname{rank}\big(\Omega(B\odot A)\big)=\frac{d(d-1)}{2}$. By Lemma \ref{lemma:omegafullrank}, the sub-matrix $\Omega_{i<j}(B\odot A)$ is full-rank. Therefore, there are $\frac{d(d-1)}{2}$ linearly independent rows in $\Omega(B\odot A)$ corresponding to $\frac{d(d-1)}{2}$ quadratic forms which are linearly independent over $\mathbb{R}$. Let's denote them by $\Psi_1,\cdots, \Psi_{s},s=d(d-1)/2$. Because the vector spaces $\mathbb{R}\langle\lambda_i\lambda_j\big|i<j\in[d]\rangle$ and $\mathbb{R}\langle\Psi_1,\cdots, \Psi_s\rangle$ have the same dimension and are both spanned by $\{\lambda_i\lambda_j,i<j\in[d]\}$, they are the same vector space and the equation $I(B\odot A)=\mathbb{R}[\lambda]\langle\lambda_i\lambda_j\big|i<j\in[d]\rangle$ holds. 

    Conversely, let $B\odot A\in \tilde{\mathcal K}_d$ be a matrix satisfying $I(B\odot A)=\mathbb{R}[\lambda]\langle\lambda_i\lambda_j\big|i<j\in[d]\rangle$. By Lemma \ref{lemma:omegafullrank}, we can find $s\leq d(d-1)/2$ linearly independent rows of $\Omega(B\odot A)$ and hence $s$ linearly independent (over $\mathbb{R}$) quadratic forms $\Psi_1,\cdots, \Psi_s$ such that $I(B\odot A)=\mathbb{R}[\lambda]\langle\Psi_1,\cdots,\Psi_s\rangle$. Therefore this equation implies that, there exists polynomials $f_l^{ij}\in \mathbb{R}[\lambda],\ i<j\in[d],\ l\in[s]$ such that 
    \begin{align*}
        \lambda_i\lambda_j=\sum_{l=1}^s f_{l}^{ij}\Psi_l,\quad i<j\in[d]. 
    \end{align*}
    By comparing the coefficients of the quadratic terms on both sides, we obtain the same representation with the constant terms of $f_l^{ij}$, so $\lambda_i\lambda_j\in \mathbb{R}\langle \Psi_1,\cdots,\Psi_s\rangle,\forall i<j\in[d]$. Therefore $\mathbb{R}\langle\lambda_i\lambda_j\big|i<j\in[d]\rangle$ and $\mathbb{R}\langle\Psi_1,\cdots, \Psi_s\rangle$ are the same vector space and have the same dimension $d(d-1)/2$ over $\mathbb{R}$. This implies $\operatorname{rank}\Omega(B\odot A)=\frac{d(d-1)}{2}$. 

    In summary, the space $\mathcal{K}_{\Omega}$ is a subspace of $\mathcal{K}_d$. Moreover, we know that $\mathcal{K}_{\Omega}$ is an open subset, because the full-rank condition $\operatorname{rank}\Omega_{i<j}=\frac{d(d-1)}{2}$ is an open condition. Therefore we have $\mathcal{K}_{\Omega}\subset \operatorname{int}(\mathcal{K}_d)$. 
\end{proof}

A direct observation is that, $4\|\Omega_{i<j}k(\lambda)\|^2=\sum\Psi_{m,n;k,l}(\lambda)^2$. This links the quadratic forms and the matrix $\Omega$ defined by \cite{CHANGDUHUANGYAO2024}. 

\begin{proof}[proof of lemma 5]

Denote $\mathcal M=\mathbb R^{p\times d}\times\mathbb R^{q\times d}$.
The set $K_d$ is open and semi-algebraic. The incidence set
\[
\mathcal I=\{(A,B,\lambda):(A,B)\in K_d,\ \|\lambda\|=1,
 \lambda\notin\{\pm e_\ell:\ell\in[d]\},\ \Psi^{(A,B)}(\lambda)=0\}
\]
is semi-algebraic. Its projection onto $(A,B)$ is $K_d\setminus\mathcal K_d$
by Proposition~2. The Tarski--Seidenberg projection theorem
\cite{BochnakCosteRoy1998} therefore shows that $\mathcal K_d$ is semi-algebraic.

Set $E=K_d\setminus\mathcal K_d$. Since $\mathcal M$ and $\mathcal K_d$ are semi-algebraic, so is $E$. Moreover, $p_0\in\mathcal K_d\cap\partial_{\mathcal M}\mathcal K_d$
implies
\[
p_0\in\overline E\setminus E.
\]
By the Nash curve selection lemma~\cite[Proposition~8.1.13]{BochnakCosteRoy1998}, there exist $\varepsilon>0$ and a Nash arc \(\gamma(t)=(A(t),B(t)),\,t\in[0,\varepsilon)\), such that
\[
\gamma(0)=p_0,
\qquad
\gamma(t)\in E
\quad\text{for all }t\in(0,\varepsilon).
\]

For $t\neq0$, define
\[
\delta_A(t)=\frac{A(t)-A}{t},
\qquad
\delta_B(t)=\frac{B(t)-B}{t}.
\]
Because $A(t)$ and $B(t)$ are real analytic at $t=0$ and
$A(0)=A$, $B(0)=B$, these maps extend real analytically to
$t=0$ by setting \(\delta_A(0)=A'(0),\,\delta_B(0)=B'(0)\).
Consequently, \(\bigl(A+t\delta_A(t),\,B+t\delta_B(t)\bigr)=\gamma(t)\in E\,\text{ for all }t\in(0,\varepsilon)\). The two components of $\delta_t$ need not have equal or unit norms. 

\end{proof}

\begin{proof}[proof of lemma 6]
Let $\delta_t$ be the semi-algebraic family chosen above. For every sufficiently
small $t>0$, the set
\[
 \{u\in\mathbb S^{d-1}:\Psi_t(u)=0,\ u\notin\{\pm e_\ell:\ell\in[d]\}\}
\]
is nonempty. Semi-algebraic selection gives a choice $u_t$ which, after
shrinking the interval, is continuous and has a limit $u_0$ as $t\to0$.
Continuity gives $\Psi_0(u_0)=0$, so $u_0=\pm e_\ell$ for some $\ell$.
Changing the sign and dividing by the $\ell$-th coordinate, put
\[
 v_t=u_t/(u_t)_\ell,\qquad \omega_t=v_t-e_\ell.
\]
Then $v_t\to e_\ell$ and (14) holds.
The bounded semi-algebraic vector $\omega_t/\|\omega_t\|$ has a limit of norm one.
After a positive-integer power reparameterization, the selected branch
is analytic at zero, and the parameter arc remains analytic.
Thus there exists at least one additional branch tending to an axis.
\end{proof}

\begin{proof}[proof of proposition 7]
    We expand $\Psi_t$ around $t=0$:  $\Psi_t=\Psi_0+\partial_t\Psi_0\cdot t+\operatorname{h.o.t}$ where $$\partial_t\Psi_{0,m,n;k,l}=\sum_{i<j\in[d]}\partial_t\psi(a^i_tb^{i\top}_t,a^jb^{j\top}_t)\cdot 2\lambda_i\lambda_j$$ and $\operatorname{h.o.t}$ are quadratic forms of $\lambda$ with diagonal part zero, and take values at zero locus: $0=\Psi_t(e_\ell+\omega_t)=\Psi_0(e_\ell+\omega_t)+\partial_t\Psi_0(e_\ell+\omega_t)\cdot t+o(t\cdot \|\omega_t\|).$ Because $\lim_{t\to0}\omega_t=0$ and $\partial_t\Psi_0$ is a quadratic form with $\partial_t\Psi_0(e_\ell)=0$, $\|\partial_t\Psi_0(e_\ell+\omega_t)\|\leq C\left(\sum_{i\ne\ell}|\omega_i|+\sum_{i<j}|\omega_i\omega_j|\right)= O(\|\omega_t\|)$, where $C$ is the upper bound of absolute values of coefficients of $\partial_t\Psi_0$, and $$\|\Psi_0(e_\ell+\omega_t)\|=\|\partial_t\Psi_0(e_\ell+\omega_t)\cdot t\|+o(t\cdot \|\omega_t\|)\leq O(t\|\omega_t\|)=o(\|\omega_t\|).$$ This reveals that every boundary point in $\mathcal{K}_d$ produces a quadratic form $\Psi$ such that $\exists \ell\in[d],\ \lim_{\epsilon\to0}\frac{\|\Psi_0(e_\ell+\epsilon\delta_\lambda)\|}{\|\epsilon\delta_\lambda\|}=0$ along some direction $\delta_\lambda\perp e_\ell$. This is just the directional derivative of $\Psi_0$ at $e_\ell$ along $\delta_\lambda\in e_{\ell}^{\perp}$. Expressing it in terms of the Jacobian matrix, we have $\delta_\lambda\in\ker\partial_{\lambda}\Psi_0(e_\ell)$. 

    Note that the limit $\lim_{t\to 0}\frac{\|\Psi_0(e_\ell+\omega_t)\|_2}{\|\omega_t\|_2}=0$ can be expressed as a derivative along $\omega_t$: \begin{equation}\label{eq:Jacmat}\lim_{t\to0}\frac{\|\Psi_0(e_\ell+\omega_t)-\Psi_0(e_\ell)\|}{\|\omega_t\|}=\lim_{t\to0}\bigg\|\partial_{\lambda}\Psi_0(e_\ell)\cdot \frac{\omega_{t}}{\|\omega_t\|}\bigg\|=0
\,,\end{equation}
where $\partial_\lambda\Psi_0(e_\ell)\in\mathbb{R}^{p^2q^2\times d}$ denotes the Jacobian matrix of $\Psi_0$.
Since $\omega_t\perp e_\ell$ for $t\in(0,\epsilon)$, the $\ell$-th entry of $\omega_t$ is always zero. Thus, \eqref{eq:Jacmat} implies that the matrix 
$\partial_\lambda\Psi_0(e_\ell)|_{e_{\ell}^\perp}:=[\partial_{\lambda_i}\Psi_0(e_\ell)]_{i\neq \ell}\in \mathbb{R}^{p^2q^2\times (d-1)}$, 
which is the matrix $\partial_\lambda\Psi_0(e_\ell)$ with the $\ell$-th column deleted, is rank deficient. 
\end{proof}

\begin{proof}[proof of proposition 8]
    We prove that every $Z\in\tilde{\mathcal{K}}_d$ such that \begin{equation}\label{eq:1ordervanishingcondition}
    \exists \ell\in[d]\ \text{and }0\ne x\in e_\ell^\perp \text{ such that }\lim_{\epsilon\to 0}\frac{\|\Omega_{i<j}(Z)k(e_\ell+\epsilon x)\|^2}{\|k(e_\ell+\epsilon x)\|^2}=0
\end{equation}is a boundary point of $\tilde{\mathcal{K}}_d$.  
Let's assume $Z=B\odot A\in\tilde{\mathcal K}_d$ and that \eqref{eq:1ordervanishingcondition} holds. 
Therefore, $$\|\Omega_{i<j}(Z)k(e_\ell+\epsilon x)\|=o\big(\epsilon\|x\|\big).$$ By equation (8), we have \begin{equation}\label{eq:jacobianvanishforphi}
\begin{aligned}
\psi^d\bigl(Z(e_\ell+\epsilon x)\bigr)
 &=2\Omega_{i<j}(Z)k(e_\ell+\epsilon x)\\
 &=\psi^d(z_\ell)+\epsilon\operatorname{Jac}(\psi^d)_{z_\ell}Zx
       +\epsilon^2\psi^d(Zx).
\end{aligned}
\end{equation}
Since $\psi^d(z_\ell)=0$, division by $\epsilon$ and passage to the limit give
$Zx\in\ker\operatorname{Jac}(\psi^d)_{z_\ell}$. Identification also gives
$\psi^d(Zx)\ne0$: otherwise $e_\ell+\epsilon x$ would be an off-axis zero
for every sufficiently small nonzero $\epsilon$. Hence
\begin{equation}\label{eq:secondordernonvanishing}
 \|2\Omega_{i<j}(Z)k(e_\ell+\epsilon x)\|
 =\epsilon^2\|\psi^d(Zx)\|,\qquad \|\psi^d(Zx)\|>0.
\end{equation}
This proved equation \eqref{eq:1ordervanishingcondition} holds if and only if $$\operatorname{span}\{z_i:i\neq \ell\}\cap \ker\big(\operatorname{Jac}(\psi^d)_{z_\ell}\big)\neq \{0\}.$$Basic algebraic manipulation shows that $\ker\big(\operatorname{Jac}(\psi^d)_{z_\ell}\big)=b^{\ell}\otimes \mathbb{R}^p+\mathbb{R}^q\otimes a^\ell$: after invertible row and column changes send $a^\ell b^{\ell\top}$ to a matrix supported at $(1,1)$, the minor derivatives require exactly that the lower-right block of a tangent matrix vanish. 

We then give a Cauchy sequence in $K_d\setminus\mathcal{K}_d$ with limit $(A,B)\in \mathcal{K}_d$. There is a linear combination $z(s)=\sum_{i\neq \ell}s_iz_i$ such that $z(s)=v\otimes a^{\ell}+b^{\ell}\otimes u$ where $u,v$ are not colinear with $a^{\ell},b^{\ell}$ respectively; otherwise $z(s)$ would have rank at most one, contradicting \eqref{eq:secondordernonvanishing}. This means that, $(b^{\ell}-\frac{1}{n}v)\otimes a^{\ell}+\frac{1}{n}z(s)=b^{\ell}\otimes(a^{\ell}+\frac{1}{n}u)$ and the sequence of matrices 
\begin{equation}\label{sequenceofboundary}
Z_n=(z_1,\cdots,z_\ell-\frac{1}{n}v\otimes a^{\ell},\cdots,z_d)
\in \tilde K_d\setminus\tilde{\mathcal K}_d\quad(n\ge n_0),\qquad
\lim_{n\to\infty}Z_n=Z,
\end{equation} is, for all sufficiently large $n$, a Cauchy sequence in $\tilde K_d\setminus\tilde{\mathcal K}_d$ with limit $Z\in\tilde{\mathcal K}_d$; full column rank persists because $\tilde K_d$ is relatively open among rank-one-column matrices. This prove $Z\in \tilde{\mathcal{K}}_d$ with the property \eqref{eq:1ordervanishingcondition} is a boundary point. Similarly, we have the following two convergent sequences 
\begin{align*}
    \{Z_n=(z_1,\cdots,[b_\ell+\frac{1}{n}s_2v]\otimes[ a^{\ell}+\frac{1}{n}s_1u],\cdots,z_d),n\in \mathbb{Z}_+\}; \\
    \{Z_n'=(z_1,\cdots,[b_\ell-\frac{1}{n}s_1v]\otimes[ a^{\ell}-\frac{1}{n}s_2u],\cdots,z_d),n\in \mathbb{Z}_+\}, 
\end{align*}
with the property: 
\begin{align*}
    [b_\ell+\frac{1}{n}s_2v]\otimes[ a^{\ell}+\frac{1}{n}s_1u]-\frac{1}{n}(s_1+s_2)z(s)=[b_\ell-\frac{1}{n}s_1v]\otimes[ a^{\ell}-\frac{1}{n}s_2u].  
\end{align*}
Therefore, $Z_n$ and $Z_n'$ represent the same model for every fixed $s_1+s_2\neq 0$ and all sufficiently large $n$. 

We still need to prove that, the intersection $\operatorname{span}\{z_i:i\neq \ell\}\cap \ker\big(\operatorname{Jac}(\psi^d)_{z_\ell}\big)$ is the image under $Z$ of the kernel space of $\partial_\lambda\Psi(e_\ell)|_{e_\ell^{\perp}}=\Big[\partial_{\lambda_{i}}\Psi_{0}(e_\ell)
\Big]_{i\in[d]\backslash{\{\ell\}}}$. This is a basic result by algebraic manipulation. Because $\psi$ is a bilinear map and $\psi(z_\ell,z_\ell)=0$, the kernel of $\partial_\lambda\Psi(e_\ell)|_{e_\ell^\perp}$ is all the vectors $\lambda$ such that $\sum_j\psi(z_j,z_\ell)\lambda_j=\psi(\sum_j\lambda_jz_j,z_\ell)=0$ and therefore $\psi^d(z_\ell+c\sum_j\lambda_jz_j)=c^2\psi^d(\sum_j\lambda_jz_j)$ for all $c\in\mathbb{R}$. Thus $\psi^d(z_\ell+c\sum_j\lambda_jz_j)=o(c)$ as $c\to0$, which is precisely the tangent-kernel condition; the map $\lambda\mapsto Z\lambda$ is injective because $Z$ has full column rank. This proves the proposition. 
\end{proof}

We then summarize the above Lemma 5-6 and Propositions 7-8 as the following theorem: 

\begin{theorem}\label{thm:mainthmforidcond}

For $(A,B)\in\mathcal K_d$, if all $\partial_\lambda\Psi_0(e_\ell)|_{e_\ell^{\perp}}=\Big[\partial_{\lambda_{i}}\Psi_{0}(e_\ell)
\Big]_{i\in[d]\backslash{\{\ell\}}}$'s are full-rank: $\operatorname{rank}\partial_{\lambda}\Psi(e_{\ell})\big|_{e_{\ell}^{\perp}}=d-1,\forall\ell\in[d]$, then $(A,B)\in \operatorname{int}\mathcal{K}_d$. Otherwise, $\exists\ell\in[d],\operatorname{rank}\partial_{\lambda}\Psi(e_{\ell})\big|_{e_{\ell}^{\perp}}<d-1$ and $(A,B)\in{\mathcal{K}}_d\setminus \operatorname{int}({\mathcal{K}}_d)$ lives on the boundary. 
\end{theorem}

In the proof of the consistency and convergence rate, giving a quantity description for a point in $\mathcal{K}_d$ to measure how close it is to the boundary is also important. To this purpose, we give a lower bound which is an analog of the {\L}ojasiewicz inequality (\cite{Lojasiewicz1965EnsemblesS,globallojasiewicsineq92,BochnakCosteRoy1998,Kurdyka2016}) with the exponent {\L}ojasiewicz $w$: 

\begin{theorem}\label{thm:mainthm2foridcond}
For every element $(A,B)\in\mathcal K_d$, there exist a positive integer $w$
and a strictly positive polynomial $C_w(\lambda)$ of degree $4(w-1)$ such that
\begin{equation}\label{ineq:mnormnondegeneracy}
\frac{\sum_{m,n;k,l}|\Psi_{m,n;k,l}(\lambda)|^2}
     {[\sum_{i<j}|\lambda_i\lambda_j|^2]^w}
\ge\frac1{C_w(\lambda)},\qquad
\lambda\notin\bigcup_{\ell\in[d]}\mathbb R e_\ell.
\end{equation}
Let $w_*$ be the smallest admissible integer exponent. Then $w_*<\infty$,
$w_*=1$ if and only if $(A,B)\in\operatorname{int}\mathcal K_d$, and
$2\le w_*<\infty$ at every identifiable boundary point.
For a singular transverse Jacobian $J_\ell$, write
$N_\ell=\ker J_\ell$ and let $P_\ell$ be the orthogonal projection onto
$(\operatorname{im}J_\ell)^\perp$. The exponent-two bound holds if and only if
\begin{equation}\label{eq:geometry-second-order-condition}
\min_{\substack{h\in N_\ell\\\|h\|=1}}
       \|P_\ell\Psi(h)\|>0
\quad\text{for every singular }J_\ell.
\end{equation}
In particular, $w_*=2$ at an identifiable boundary point satisfying this condition.
As before, a pair is called \textbf{$w$-nondegenerate} if
\eqref{ineq:mnormnondegeneracy} holds for that $w$.
\end{theorem}

\begin{proof}
Write
\[
 F_Z(\lambda)=\sum_{m,n;k,l}|\Psi_{m,n;k,l}(\lambda)|^2
             =4\|\Omega_{i<j}(Z)k(\lambda)\|^2,
 \qquad K(\lambda)=\|k(\lambda)\|^2.
\]
For $d=1$ the assertion is vacuous, so take $d\ge2$.
We retain the three steps: obtain a finite exponent from the real-radical
identity, characterize $w_*=1$ through the transverse Jacobians, and then
examine the additional condition needed for $w_*=2$.

\medskip
\noindent
\hypertarget{step:main:1}{\textbf{Step 1:}}
We first prove that (12) implies a finite-order inequality.
If $I(Z)=\mathbb R[\lambda]\langle\lambda_i\lambda_j:i<j\rangle$,
Corollary~3 implies that $\Omega_{i<j}(Z)$ has full
column rank, and
$F_Z(\lambda)\ge4s_{\min}(\Omega_{i<j}(Z))^2K(\lambda)$.
Thus $w=1$ is admissible.

In general, the definition of the real radical gives a common integer $n\ge1$
and finitely many polynomials $f_l^{ij},g_r^{ij}$ such that
\begin{equation}
 (\lambda_i\lambda_j)^{2n}+\sum_r(g_r^{ij})^2
 =\sum_{l=1}^s f_l^{ij}\Psi_l,\qquad i<j,
\end{equation}
where $\Psi_1,\ldots,\Psi_s$ are linearly independent generators chosen from
the minor forms, and $s\le D:=d(d-1)/2$.
A common $n$ is obtained by multiplying each certificate by an even power
of $\lambda_i\lambda_j$ if necessary.
Discarding the nonnegative sum of squares and applying H\"older's inequality gives
\begin{equation}
\begin{aligned}
 (\lambda_i\lambda_j)^{2n}
 &\le\sum_{l=1}^s |f_l^{ij}|\,|\Psi_l|
 \le\left\{\sum_l(f_l^{ij})^2\right\}^{1/2}F_Z(\lambda)^{1/2},\\
 K(\lambda)^{2n}
 &\le D^{2n-1}\sum_{i<j}(\lambda_i\lambda_j)^{4n}
 \le D^{2n-1}\left\{1+\sum_{i<j,l}(f_l^{ij})^2\right\}F_Z(\lambda).
\end{aligned}
\end{equation}
The polynomial in braces is bounded on the unit sphere. Therefore some $c_w>0$
gives $F_Z(u)\ge c_wK(u)^w$ there, with $w=2n$.
For any admissible $w$, homogeneity with $\lambda=ru$ yields
\begin{equation}\label{eq:analyticversioninequality}
 F_Z(\lambda)\|\lambda\|^{4(w-1)}
 \ge c_wK(\lambda)^w,\qquad\lambda\ne0.
\end{equation}
Consequently \eqref{ineq:mnormnondegeneracy} holds with
$C_w(\lambda)=c_w^{-1}(1+\|\lambda\|^2)^{2(w-1)}$, which is strictly positive
and has the stated degree. Conversely, a positive polynomial bound restricts
to a constant bound on the sphere by compactness.

\medskip
\noindent
\hypertarget{step:main:2}{\textbf{Step 2:}}
We next prove that
\begin{equation}\label{inequalitywhenm=1}
 F_Z(\lambda)\ge cK(\lambda),\qquad\lambda\in\mathbb R^d,
\end{equation}
holds for some $c>0$ if and only if every $J_\ell$ has full column rank.
Equivalently, $\tilde\Omega=J_1\oplus\cdots\oplus J_d$ has full
column rank. Here the convention (8) gives
$J_\ell=2[\operatorname{vec}\psi(z_\ell,z_i)]_{i\ne\ell}$.

For $x\in e_\ell^\perp$, expansion at $e_\ell$ gives the exact formula
\[
 \Psi(e_\ell+\epsilon x)=\epsilon J_\ell x+\epsilon^2\Psi(x).
\]
Thus the numerator and denominator in the ratio satisfy
\begin{align*}
 F_Z(e_\ell+\epsilon x)&=\epsilon^2\|J_\ell x\|^2+O(\epsilon^3),\\
 K(e_\ell+\epsilon x)&=\epsilon^2\|x\|^2+\epsilon^4K(x),
\end{align*}
and, for $x\ne0$,
\[
 \lim_{\epsilon\to0}\frac{F_Z(e_\ell+\epsilon x)}{K(e_\ell+\epsilon x)}
 =\frac{\|J_\ell x\|^2}{\|x\|^2}.
\]
The lower bound \eqref{inequalitywhenm=1} therefore forces every $J_\ell$
to have a strictly positive smallest singular value.

Conversely, put $\sigma=\min_\ell s_{\min}(J_\ell)>0$ and choose $L\ge1$
with $\|\Psi(v)\|\le L\|v\|^2$. On the unit sphere write
$u=ae_\ell+v$, $v\perp e_\ell$. When
$\|v\|\le\delta:=\min\{1/2,\sigma/(4L)\}$, we have $|a|\ge1/2$ and
\[
 \|\Psi(u)\|=\|aJ_\ell v+\Psi(v)\|
 \ge\tfrac\sigma4\|v\|,\qquad
 K(u)=a^2\|v\|^2+K(v)\le\|v\|^2.
\]
On the remaining compact part of the sphere, identification gives
$F_Z(u)>0$. Combining these regions gives a common lower bound for
$F_Z(u)/K(u)$, and homogeneity proves \eqref{inequalitywhenm=1}.
Theorem~\ref{thm:mainthmforidcond} now gives
$w_*=1\Longleftrightarrow(A,B)\in\operatorname{int}\mathcal K_d$.
At a boundary point, $w_*\ge2$. The use of the smallest exponent is essential:
once an exponent is admissible on the sphere, every larger integer is also
admissible because $K(u)<1$ there.

\medskip
\noindent
\hypertarget{step:main:3}{\textbf{Step 3:}}
The fixed-direction calculation \eqref{eq:secondordernonvanishing} does not
by itself give a uniform exponent-two bound, because a curved path may
cancel its quadratic term. We now check precisely when that bound holds.
In one coordinate chart write
\[
 f(h)=\Psi(e_\ell+h)=Jh+Q(h),\quad h\in e_\ell^\perp,\qquad
 J=J_\ell,\quad Q(h)=\Psi(h).
\]
Let $N=\ker J$, $H=N^\perp$, and $P=\operatorname{Proj}_{(\operatorname{im}J)^\perp}$.
If $PQ(v)=0$ for a unit $v\in N$, choose $u\in H$ with $Ju=-Q(v)$.
By bilinearity,
\[
 f(tv+t^2u)=2t^3\psi(Zv,Zu)+t^4Q(u)=O(t^3),\qquad
 \|tv+t^2u\|\asymp|t|.
\]
Hence no bound $\|f(h)\|\ge c\|h\|^2$ can hold in this chart.

Conversely, suppose the minimum in \eqref{eq:geometry-second-order-condition}
is positive, but this quadratic lower bound fails. There are $h_n\to0$,
$r_n=\|h_n\|>0$, with $\|f(h_n)\|=o(r_n^2)$.
Decompose $h_n=x_n+y_n$ with $x_n\in N$ and $y_n\in H$. Since
$Jy_n=f(h_n)-Q(h_n)=O(r_n^2)$ and $J|_H$ is injective,
$y_n=O(r_n^2)$. A subsequence therefore satisfies
$h_n/r_n\to v\in N$, $\|v\|=1$, while
\[
 0=\lim_n\frac{Pf(h_n)}{r_n^2}
   =\lim_nPQ(h_n/r_n)=PQ(v),
\]
a contradiction. If $N=\{0\}$, Step~2 gives a local linear bound instead.
Finally, $K(e_\ell+h)=\|h\|^2+K(h)\asymp\|h\|^2$ near each axis.
The chart bounds and compactness away from the axes give
$F_Z(u)\ge cK(u)^2$ on the sphere exactly under
\eqref{eq:geometry-second-order-condition}. This proves the last assertion.
\end{proof}

For use in the consistency proof, these arguments also give the following
uniform statements on a fixed compact identifiable family $\mathcal C$ of
normalized cores. For $M<\infty$ and $\varepsilon>0$, whenever the constrained
set is nonempty,
\begin{equation}\label{eq:geometry-uniform-separation}
 \rho_{\mathcal C,M}(\varepsilon)
 :=\min_{\substack{\theta\in\mathcal C,\ \|u\|\le M\\
                   \|k(u)\|\ge\varepsilon}}\|\Psi_\theta(u)\|>0.
\end{equation}
Indeed, the minimized function is continuous on a compact set and cannot
vanish there by identification. If $\mathcal C$ lies in the identifiable
interior, compactness also gives a common positive lower bound for the
smallest singular values of all transverse Jacobians, so Step~2 gives
\begin{equation}\label{eq:geometry-uniform-linear}
 \|\Psi_\theta(u)\|^2\ge c_{\mathcal C}\|k(u)\|^2,
 \qquad\theta\in\mathcal C,\quad u\in\mathbb R^d.
\end{equation}

\begin{proposition}\label{prop:polardecomposition}
Suppose that $p^{-1/2}\|\widehat A-A\|_F=o_p(1)$, and define
$A_s=(A^\top A)^{1/2}$ and $\widehat A_s=(\widehat A^\top\widehat A)^{1/2}$.
Then $p^{-1/2}\|\widehat A_s-A_s\|_F=o_p(1)$.
When $p\ge d$, there exist compatible decompositions
$A=A_0A_s$, $\widehat A=\widehat A_0\widehat A_s$ with
$A_0,\widehat A_0\in\mathbb R^{p\times d}$ and
$A_0^\top A_0=\widehat A_0^\top\widehat A_0=I_d$.
The Gram-core conclusion holds for all $p$, without a rank assumption.
\end{proposition}

\begin{proof}
When $p\ge d$, take singular value decompositions
\[
 A=U\Sigma V^\top,\qquad
 \widehat A=\widehat U\widehat\Sigma\widehat V^\top,
\]
where $U,\widehat U\in\mathbb R^{p\times d}$ have orthonormal columns and
$V,\widehat V\in\mathbb R^{d\times d}$ are orthogonal. Zero singular directions
may be completed by arbitrary orthonormal columns. Define
\[
 A_0=UV^\top,\qquad A_s=V\Sigma V^\top=(A^\top A)^{1/2},
\]
and
\[
 \widehat A_0=\widehat U\widehat V^\top,\qquad
 \widehat A_s=\widehat V\widehat\Sigma\widehat V^\top
             =(\widehat A^\top\widehat A)^{1/2}.
\]
Then
\[
 A_0^\top A_0=\widehat A_0^\top\widehat A_0=I_d,\qquad
 A=A_0A_s,\quad\widehat A=\widehat A_0\widehat A_s.
\]
For arbitrary $p$, the Hilbert--Schmidt perturbation inequality for matrix
absolute values yields
\[
 \|\widehat A_s-A_s\|_F
 =\|(\widehat A^\top\widehat A)^{1/2}-(A^\top A)^{1/2}\|_F
 \le\sqrt2\|\widehat A-A\|_F.
\]
For completeness, apply $\||H|-|K|\|_F\le\|H-K\|_F$ to the symmetric
dilations $H=\left(\begin{smallmatrix}0&A\\A^\top&0\end{smallmatrix}\right)$
and its estimated counterpart. The lower-right block of $|H|$ is $A_s$,
and the squared dilation difference is $2\|\widehat A-A\|_F^2$.
The inequality for symmetric $H,K$ follows from their spectral decompositions
and $(|x|-|y|)^2\le(x-y)^2$.
Consequently,
\[
 p^{-1/2}\|\widehat A_s-A_s\|_F
 \le\sqrt2\,p^{-1/2}\|\widehat A-A\|_F=o_p(1).
\]
The thin SVD also shows that passage from $A,B$ to their Gram cores preserves
the rank of $A\operatorname{Diag}(\lambda)B^\top$ for every $\lambda$.
\end{proof}

\section{Proof of consistency and convergence rate}
\label{rev9:section}

\paragraph*{Proof conditions and scope}
The arguments below establish the conclusions associated with
Theorems~9 and~10 subject to the
supplementary conditions stated here. In particular, the estimator domain,
its margins, and the uniformity requirements below supplement the
assumptions stated in the main article; they are not consequences of those
statements alone. The $4+\delta$ moment and mixing summability conditions
are already present in Assumption~A1; their
uniform interpretation along a triangular array is made explicit here.

\begin{proofcondition}[Likelihood domain and uniformity]
\label{supp:likelihood-conditions}
Keep the model and Assumptions~A1--A4
of the main article, with the following specifications. The factor dimension
$d$ is fixed and $N=pq$. For fixed constants
$0<d_-<d_+<\infty$, $0<m_-<m_+<\infty$, and $g>0$, let
\begin{equation}\label{supp:qmle-domain}
 \Theta_N=\left\{(A,B,M_x,\Sigma_e):
 \begin{array}{l}
 \|A_{\cdot\ell}\|_2^2=p,\quad
 \|B_{\cdot\ell}\|_2^2=q\quad(\ell\le d),\\
 \lambda_{\min}\{(B\odot A)^\top(B\odot A)/N\}\ge g,\\
 m_-I_d\preceq M_x\preceq m_+I_d,\\
 \Sigma_e\text{ diagonal},\quad
 d_-\le(\Sigma_e)_{rr}\le d_+\quad(r\le N)
 \end{array}\right\}.
\end{equation}
The theoretical estimator is a measurable global maximizer of the
centered quasi-likelihood over $\Theta_N$. The true parameter belongs to
this domain and, for fixed $\epsilon_M,\epsilon_g>0$, satisfies
\[
 (m_-+\epsilon_M)I_d\preceq M_x^*
       \preceq(m_+-\epsilon_M)I_d,
 \qquad
 \lambda_{\min}\{Z^{*\top}Z^*/N\}\ge g+\epsilon_g.
\]
The bounds on the rows of the true loading matrices in
Assumption~A3 hold uniformly; no bound on every
fitted loading row is imposed.

For each $T$, the factor process is stationary, has mean zero and
covariance $M_x^*$, and for some fixed $\delta>0$,
\[
 \sup_T E\|x_{T,0}\|^{4+\delta}<\infty,
 \qquad
 \sup_T\sum_{h=1}^{\infty}
       \alpha_{x,T}(h)^{\delta/(4+\delta)}<\infty.
\]
The entire factor process is independent of the entire error process.
The error coordinates are independent across coordinates and time, as in
Assumption~A2, with their fourth moments uniformly
bounded along the array. All constants in the domain and the foregoing
bounds are independent of $N,T$.

For thin orthonormal frames $Q_A,Q_B$ of the true loading spaces, and
each possible pair of mode ranks $(r_A,r_B)$, the normalized true core
pairs
\[
 (Q_A^\top A^*/\sqrt p,\ Q_B^\top B^*/\sqrt q)
\]
lie in a fixed compact subset of $\mathcal K_d(r_A,r_B,d)$. This permits
identifiable boundary points for consistency. For the interior rate,
additionally require
\begin{equation}\label{supp:uniform-interior}
 \inf_{T,\ell\le d}
 \sigma_{\min}\!\left\{
 J_\ell\!\left[
 (Q_B^\top B^*/\sqrt q)\odot(Q_A^\top A^*/\sqrt p)
 \right]\right\}>0,
\end{equation}
where $J_\ell$ is the transverse Jacobian used in
Theorem~\ref{thm:mainthm2foridcond}. This condition is vacuous when $d=1$.
\end{proofcondition}

For $N,T\to\infty$, the consistency argument uses
Proof condition~\ref{supp:likelihood-conditions}; the $T^{-1}$ squared-error
rates also use \eqref{supp:uniform-interior}. Either mode dimension may
remain fixed, and no restriction on $N/T$ is imposed. Parameter comparisons
are made after a common column permutation and compatible signs, with
the corresponding conjugation of $M_x$. The notation $\widehat M_x$
below denotes the QMLE; the empirical factor covariance is denoted by
$M_{X,T}$.

\begin{proofcondition}[Fixed-dimensional identification and regularity]
\label{supp:fixed-covariance}
When $p,q$ are fixed, the true parameter and normalized domain $\Theta_N$
are fixed as $T$ varies.
\begin{enumerate}
\item[(i)] There is no additional covariance identification problem:
$\Sigma_y(\theta)=\Sigma_y(\theta^*)$ on $\Theta_N$ implies equality up to
a common permutation and compatible signs.
\item[(ii)] For the rate, the truth is interior in smooth normalized local
coordinates, and the population information matrix is nonsingular there.
\end{enumerate}
\end{proofcondition}
The fixed-dimensional calculations in Step~6 use part~(i) for consistency
and both parts for the $T^{-1/2}$ parameter rate, together with the domain
and stochastic conditions of Proof condition~\ref{supp:likelihood-conditions}.

In this proof, write $N=pq$, and keep the population factor covariance
$M_x^*$ distinct from its empirical version.  For
$X=[x_1^*,\ldots,x_T^*]$, $E=[e_1,\ldots,e_T]$, and
$P_T=I_T-\mathbf1\mathbf1^\top/T$, put
\begin{align}
 M_{X,T}&=T^{-1}XP_TX^\top=M_x^*+F_T, &
 G_T&=T^{-1}XP_TE^\top, & E_T&=T^{-1}EP_TE^\top-\Sigma_e^*.
 \label{rev9:primitives-definition}
\end{align}
For the centered sample covariance $\widehat M_y$,
\begin{equation}\label{rev9:sample-decomposition}
 \cW_T:=\widehat M_y-\Sigma_y^*
 =Z^*F_TZ^{*\top}+Z^*G_T+G_T^\top Z^{*\top}+E_T.
\end{equation}
The bound $F_T=O_p(T^{-1/2})$ follows from
Proof condition~\ref{supp:likelihood-conditions}; its contribution is retained
below. Choose $C_e\ge1$ with $[d_-,d_+]\subset[C_e^{-2},C_e^2]$.
Below $\|\cdot\|$ is the operator norm for matrices and $\|\cdot\|_F$
is the Frobenius norm. We use the original estimated-parameter notation and recall
$\widehat W=\Seh^{-1}$ and $\Ph=\Syh^{-1}$.
For a block index $j$, let $J_j$ select rows $(j-1)p+1,\ldots,jp$,
and put $\widehat\Sigma_{ej}^{-1}=J_j\widehat WJ_j^\top$ and
$\widehat b_j=\operatorname{Diag}(\widehat B_{j,\cdot})$.
Thus $J_j\Zh=\widehat A\widehat b_j$. For a vectorized coordinate
$r=(j-1)p+i$, $\widehat\sigma_r^2=\widehat\sigma_{ij}^2$ and $\sigma_r^{*2}=\sigma_{ij}^{*2}$.
Then
\begin{align}
 \Gh&=(\Mh^{-1}+\sum_j\widehat b_j\widehat A^\top
              \widehat\Sigma_{ej}^{-1}\widehat A\widehat b_j)^{-1},\\
 \widehat H&=(\Zh^\top\widehat W\Zh)^{-1}
   =(\sum_j\widehat b_j\widehat A^\top
             \widehat\Sigma_{ej}^{-1}\widehat A\widehat b_j)^{-1},\\
 \Seh&=\operatorname{Diag}(\widehat\sigma_{11}^{2},\ldots,\widehat\sigma_{pq}^{2}),\\
 \mathcal Z&=(\Zh-Z^*)^\top\widehat W\Zh\widehat H,
 \qquad (I-\mathcal Z)^\top=\widehat H\Zh^\top\widehat WZ^*.
\end{align}
For invertible $U+V$ and $V$, the equation
\begin{equation}\label{eq:initialinverse}
 (U+V)^{-1}=V^{-1}-V^{-1}U(U+V)^{-1},\qquad \widehat H =\Gh (I-\Mh^{-1}\Gh )^{-1}
\end{equation}
holds. Woodbury's formula gives
\begin{equation}\label{eq:syinverse}
 \begin{split}
 \Ph &=\widehat W -\widehat W \Zh \Gh  \Zh^\top \widehat W ,\\
 J_i\Ph J_j^\top&=\mathbf1_{\{i=j\}}\widehat\Sigma_{ei}^{-1}
   -\widehat\Sigma_{ei}^{-1}\widehat A\widehat b_i\Gh \widehat b_j\widehat A^\top \widehat\Sigma_{ej}^{-1},\\
 \Zh^\top \Ph &=\Mh^{-1}\Gh  \Zh^\top \widehat W ,\qquad \Ph \Zh \Mh=\widehat W \Zh \Gh .
 \end{split}
\end{equation}

The derivative of $\mathcal L_T$ gives the following first-order
conditions. The applicability of the loading and $M_x$ equalities is
verified at the end of Step 4. For the diagonal variance box, the
corresponding condition is a variational inequality. Write
$\Sh=\Ph (\widehat\Sigma_y-\widehat M_y)\Ph $, $g=\operatorname{diag}(\Sh)$,
and $\mathcal E_e=\Seh-\Sigma_e^*$. Then
\begin{subequations}\label{gradconds}
\begin{align}
 \sum_{j=1}^q \widehat\Sigma_{ej}^{-1}J_j(\widehat\Sigma_y-\widehat M_y)\widehat W \Zh \Gh \widehat b_j&=0,
       \label{gradcondforCP-A}\\
 \operatorname{diag}\{\widehat A^\top \widehat\Sigma_{ej}^{-1}J_j
              (\widehat\Sigma_y-\widehat M_y)\widehat W \Zh \Gh \}&=0,\quad j\le q,
       \label{gradcondforCP-B}\\
 \Zh^\top \widehat W (\widehat\Sigma_y-\widehat M_y)\widehat W \Zh &=0,
       \label{gradcondforCP-X}\\
 g^\top\operatorname{diag}(-\mathcal E_e)&\ge0,\quad
 \begin{cases}
 g_r=0,&d_-<\widehat\sigma_r^2<d_+,\\
 g_r\ge0,&\widehat\sigma_r^2=d_-,\\
 g_r\le0,&\widehat\sigma_r^2=d_+.
 \end{cases}
       \label{gradcondforCP-e}
\end{align}
\end{subequations}
Here $\widehat\Sigma_{ej}^{-1}$ is a $p\times p$ diagonal block and $J_j\widehat W =\widehat\Sigma_{ej}^{-1}J_j$
is the corresponding $p\times N$ block of rows.

\begin{proof}[Proof of Theorem 9]
The centered likelihood function can be written, up to its constant at
$\theta^*$, as
\[
 \mathcal L_T(\theta)-\mathcal L_T(\theta^*)
 =\bar{\mathcal L}(\theta)+R_T(\theta)-R_T(\theta^*),
\]
where
\[
 \begin{aligned}
 \bar{\mathcal L}(\theta)&=-\frac{1}{2N}\log|\Sigma_y|
 -\frac{1}{2N}\operatorname{tr}(\Sigma_y^*\Sigma_y^{-1})
 +\frac12+\frac{1}{2N}\log|\Sigma_y^*|,\\
 R_T(\theta)&=-\frac1{2N}\operatorname{tr}(\cW_T\Sigma_y^{-1}).
 \end{aligned}
\]

    \textbf{Step 1: }We firstly prove the error term $\mathcal{E}^2\coloneqq\frac{1}{pq}\operatorname{tr}[(\hat\Sigma_{e}\Sigma_{e}^{*-1}-I_{pq})^2]=o_p(1)$. 
    The lemma \ref{lemmaRformatrixfactor} implies $R_T(\theta)=O_p(T^{-\frac{1}{2}})$ uniformly in $\Theta_N$, and therefore $|R_T(\theta^*)-R_T(\hat\theta)|=O_p(T^{-\frac{1}{2}})$. Since $\hat\theta$ maximizes $\mathcal L_T$ and $\theta^*$ maximizes
$\bar{\mathcal L}$,
\[
 0\ge\bar{\mathcal L}(\widehat\theta)
 \ge R_T(\theta^*)-R_T(\widehat\theta)\ge-O_p(T^{-1/2}).
\]

    Noticing $|\Sigma_y|=|\Sigma_{e}|\cdot |I_{d}+M_xZ^{\top}\Sigma_{e}^{-1}Z|$ and $|I_{d}+M_xZ^{\top}\Sigma_{e}^{-1}Z|=O((pq)^d)$, we have \begin{equation}\label{eq:errorlnSigmay}\frac{1}{2pq}\ln|\Sigma_{y}^*|-\frac{1}{2pq}\ln|\Sigma_y|=\frac{1}{2pq}(\ln|\Sigma_{e}^*|-\ln|\Sigma_{e}|)+O_p(\frac{\ln(pq)}{pq}).\end{equation}
    For the term $\operatorname{tr}\big(\Sigma_{y}^*\Sigma_y^{-1}\big)$, we use equation \eqref{eq:syinverse} and derive \begin{align*}\Sigma_y^*\Sigma_{y}^{-1}&=Z^*M_x^* Z^{*\top}\Sigma_{y}^{-1}+\Sigma_{e}^*\Sigma_{e}^{-1}-\Sigma_{e}^*\Sigma_{e}^{-1}Z(M_x^{-1}+Z^\top\Sigma_{e}^{-1}Z)^{-1}Z^{\top}\Sigma_{e}^{-1}. 
    \end{align*}
    Because \begin{align*}&\operatorname{tr}[\Sigma_{e}^*\Sigma_{e}^{-1}Z(M_x^{-1}+Z^\top\Sigma_{e}^{-1}Z)^{-1}Z^{\top}\Sigma_{e}^{-1}]\\&=\operatorname{tr}[Z^{\top}\Sigma_{e}^{-1}\Sigma_{e}^*\Sigma_{e}^{-1}Z(M_x^{-1}+Z^\top\Sigma_{e}^{-1}Z)^{-1}]\\&\leq \operatorname{tr}[C_e^4Z^{\top}\Sigma_{e}^{-1}Z(M_x^{-1}+Z^\top\Sigma_{e}^{-1}Z)^{-1}]\le C_e^4d,\end{align*} combining equation \ref{eq:errorlnSigmay}, we have \begin{align*}
        \bar{\mathcal{L}}(\hat\theta)&=-\frac{1}{2pq}\bigg[\sum_{i,j}\bigg(\ln(\hat{\sigma}_{ij}^2/\sigma_{ij}^{*2})+\frac{\sigma_{ij}^{*2}}{\hat{\sigma}_{ij}^2}-1\bigg)\bigg]-\frac{1}{2pq}\operatorname{tr}(Z^*M_x^* Z^{*\top}\widehat\Sigma_{y}^{-1})\\&+O_p(\frac{\ln (pq)}{pq})+O_p(\frac{1}{pq}). 
    \end{align*}
    Notice that $\ln\ x+\frac{1}{x}-1\geq 0$, $\frac{1}{2pq}\sum_{i,j}\bigg(\ln(\hat{\sigma}_{ij}^2/\sigma_{ij}^{*2})+\frac{\sigma_{ij}^{*2}}{\hat{\sigma}_{ij}^2}-1\bigg)$ and $\frac{1}{2pq}\operatorname{tr}(Z^*M_x^* Z^{*\top}\widehat\Sigma_{y}^{-1})$ are both nonnegative. By $\bar{\mathcal{L}}(\hat{\theta})=O_p(T^{-\frac{1}{2}})$, both terms are $o_p(1)$. Since $x=1$ is the only zero of $f(x)=\ln\ x+\frac{1}{x}-1$, for every compact interval $I\subset(0,\infty)$, there exists $c_I>0$ such that $f(x)\geq c_I(x-1)^2$ holds on $I$. Because $C_e^{-4}\leq\hat{\sigma}_{ij}^2/\sigma_{ij}^{*2} \leq C_e^4$ holds for all $i\in[p],j\in[q]$, the error term \begin{equation}\label{errorforwninCPfactmod}
        \frac{1}{pq}\sum_{i,j}(\hat{\sigma}_{ij}^2/\sigma_{ij}^{*2}-1)^2\xrightarrow{p}0,  
    \end{equation}
    goes to zero as $T\to \infty$. 

\textbf{Step 2: }Secondly, we prove $\widehat G_H=o_p(1)$ and
$\widehat H=o_p(1)$, and control the orthogonal residual
$R=\Zh (I-\mathcal Z)^\top-Z^*$. The candidate Gram condition in the parameter domain gives
$\|\widehat H \|\le C_0/N$ and $\|\Gh \|\le C_0/N$ directly. Regarding the other
nonnegative term from Step 1, we expand
\[
 \operatorname{tr}(Z^*M_x^*Z^{*\top}\Ph )
 =\operatorname{tr}\{M_x^*Z^{*\top}\widehat W Z^*
          -M_x^*Z^{*\top}\widehat W \Zh \Gh  \Zh^\top \widehat W Z^*\}.
\]
The two-part signal expansion is
\[
 \begin{split}
 \operatorname{tr}(Z^*M_x^*Z^{*\top}\Ph )
 &=\operatorname{tr}\{M_x^*Z^{*\top}(\widehat W -\widehat W \Zh \widehat H \Zh^\top \widehat W )Z^*\}\\
 &\quad+\operatorname{tr}
       \{M_x^*Z^{*\top}\widehat W \Zh (\widehat H -\Gh )\Zh^\top \widehat W Z^*\}.
 \end{split}
\]
Both terms are nonnegative because
$\widehat W -\widehat W \Zh \widehat H \Zh^\top \widehat W \succeq0$ and
$\widehat H -\Gh=\widehat H \Mh^{-1}\Gh=\widehat H (\Mh +\widehat H )^{-1}\widehat H \succeq0$.
Each nonnegative part has its own normalized limit:
\[
 \begin{split}
 N^{-1}\operatorname{tr}\{M_x^*Z^{*\top}(\widehat W -\widehat W \Zh \widehat H \Zh^\top \widehat W )Z^*\}&=o_p(1),\\
 N^{-1}\operatorname{tr}
   \{M_x^*Z^{*\top}\widehat W \Zh (\widehat H -\Gh )\Zh^\top \widehat W Z^*\}&=o_p(1).
 \end{split}
\]
For the second part, $\|\widehat H -\Gh \|\le C\|\widehat H \|^2\le C/N^2$ and
$\|\Zh^\top \widehat W Z^*\|_F\le CN$, so its unnormalized trace is $O_p(1)$.
This proves its normalized limit directly, without inferring eigenvalue
bounds for an unsymmetrized product.
The first term yields the matrix limit
\begin{equation}\label{eqfanzhengjiuyongyici}
 N^{-1}\{Z^{*\top}\widehat W Z^*-Z^{*\top}\widehat W \Zh \widehat H \Zh^\top \widehat W Z^*\}
 =N^{-1}R^\top \widehat W R=o_p(1).
\end{equation}
The candidate Gram condition directly supplies the stronger bounds
$\|\widehat H \|\le C/N$ and $\|\Gh \|\le C/N$; no inverse of a possibly singular
limiting alignment matrix is invoked here.
Let $H^*=(Z^{*\top}\Sigma_e^{*-1}Z^*)^{-1}$.
Lemma~\ref{lemmaforHerror}, the variance consistency, and
\eqref{eqfanzhengjiuyongyici} give
\begin{equation}\label{limitofH_rforfactmod}
 H^* Z^{*\top}\widehat W \Zh (I-\mathcal Z)^\top=I+o_p(1),\qquad
 (I-\mathcal Z) \Zh^\top \widehat W Z^*H^*=I+o_p(1).
\end{equation}
These identities keep $\widehat W $ in the cross-Gram. Replacing it there by
$\Sigma_e^{*-1}$ without a bound on the fitted rows is unnecessary.
The residual orthogonality gives, exactly,
\[
 R^\top \widehat W \Zh =0,\qquad
 R^\top \widehat W Z^*H^*
 =R^\top \widehat W \Zh (I-\mathcal Z)^\top H^*-R^\top \widehat W RH^*=-R^\top \widehat W RH^*.
\]
Thus
\begin{equation}\label{eq:Risreallysmall}
 R^\top \widehat W RH^*=o_p(1),
\end{equation}
and
\begin{equation}\label{eq:HhatandH*}
 \begin{split}
 (I-\mathcal Z) \Zh^\top \widehat W \Zh (I-\mathcal Z)^\top
 &=Z^{*\top}\widehat W Z^*-R^\top \widehat W R,\\
 (I-\mathcal Z)(\Zh^\top \widehat W \Zh /N)(I-\mathcal Z)^\top\,(NH^*)&=I+o_p(1).
 \end{split}
\end{equation}
The upper and lower spectral bounds on $NH^*$, together with the
bounded variance ratios, turn these into
$N^{-1}\operatorname{tr}(R^\top \widehat W R)=o_p(1)$ and
$N^{-1}\|R\|_F^2=o_p(1)$. In coordinates,
\[
 x_T=\frac1N\sum_r \widehat W_{rr}\|R_{r,\cdot}\|_2^2
 \ge d_+^{-1}\frac{\|R\|_F^2}{N}.
\]

\textbf{Step 3: }We next prove $I_d-\mathcal Z=O_p(1)$.
In \eqref{eq:HhatandH*}, both $\Zh^\top \widehat W \Zh /N$ and $Z^{*\top}\widehat W Z^*/N$
have fixed positive lower and upper spectral bounds, and
$R^\top \widehat W R/N=o_p(1)$. Hence $(I-\mathcal Z)^\top$ and $(I-\mathcal Z)^{-\top}$ are bounded in probability.
This does not use the score equations.

We retain the explicit covariance and $M_x$-score calculations here;
the score identities are used after their validity has been established
in Step 4. Put
\[
 \mathcal T=-(Z^*G_T+G_T^\top Z^{*\top}+E_T),\qquad
 \mathcal E_e=\Seh-\Sigma_e^*.
\]
Then $\widehat M_y=Z^*M_{X,T}Z^{*\top}+\Sigma_e^*-\mathcal T$ exactly, and
\begin{equation}\label{eqforerrorofsigmay}
 \begin{split}
 \widehat\Sigma_y-\widehat M_y
 &=\Zh \Mh \Zh^\top-Z^*M_{X,T}Z^{*\top}+\mathcal E_e+\mathcal T\\
 &=\Zh (\Mh -M_{X,T})\Zh^\top-(\Zh-Z^*)M_{X,T}(\Zh-Z^*)^\top\\
 &\quad+(\Zh-Z^*)M_{X,T}\Zh^\top+\Zh M_{X,T}(\Zh-Z^*)^\top
              +\mathcal E_e+\mathcal T.
 \end{split}
\end{equation}
Multiplying by $\widehat H \Zh^\top \widehat W $ and $\widehat W \Zh \widehat H $ and using the $\Mh $ score gives,
term by term,
\begin{equation}\label{errorforMx}
 \begin{split}
 \Mh -M_{X,T}
 ={}&\widehat H \Zh^\top \widehat W (\Zh-Z^*)M_{X,T}(\Zh-Z^*)^\top \widehat W \Zh \widehat H \\
 &-M_{X,T}(\Zh-Z^*)^\top \widehat W \Zh \widehat H -\widehat H \Zh^\top \widehat W (\Zh-Z^*)M_{X,T}\\
 &+\widehat H \Zh^\top \widehat W (Z^*G_T+G_T^\top Z^{*\top}+E_T)\widehat W \Zh \widehat H \\
 &+\widehat H \Zh^\top \widehat W (\Sigma_e^*-\Seh )\widehat W \Zh \widehat H .
 \end{split}
\end{equation}
Equivalently, with $B_H=\Zh \widehat H  \Zh^\top \widehat W $ and
$C_T=\widehat W \Zh \Gh $,
\[
 \Mh=(I-\mathcal Z)^\top M_{X,T}(I-\mathcal Z)-\widehat H \Zh^\top \widehat W (\mathcal T+\mathcal E_e)\widehat W \Zh \widehat H .
\]

\textbf{Step 4: }Finally, we prove $\mathcal Z=o_p(1)$ and
$N^{-1}\|\widehat Z-Z^*\|_F^2=o_p(1)$. Let $Q_A,Q_B$ be orthonormal
frames for the true column spaces, of dimensions $r_A,r_B\le d$.
Let $Q_*=Q_B\otimes Q_A$, and use the normalized true core
$Z_c^*=Q_*^\top Z^*/\sqrt N$ and the projected fitted core
$\widehat Z_c=Q_*^\top \Zh /\sqrt N$. Then
\begin{equation}\label{eq:Zr*isZhatrplusRr}
 Z_c^*=\widehat Z_c(I-\mathcal Z)^\top-R_c,\qquad
 R_c=Q_*^\top R/\sqrt N,\qquad \|R_c\|_F^2\le Cx_T.
\end{equation}
The normalization is $1/\sqrt N$ for orthonormal frames.
The two cores retain the Khatri--Rao form:
\[
 Z_c^*=(Q_B^\top B^*/\sqrt q)\odot(Q_A^\top A^*/\sqrt p).
\]
For the fitted core,
\[
 \widehat Z_c
 =(Q_B^\top\widehat B/\sqrt q)\odot
   (Q_A^\top\widehat A/\sqrt p).
\]
For row-pair and column-pair indices, the polarized minors are
\[
 \Omega(Z)_{mnkl,ij}
 =\frac12\left\{
 \det\begin{pmatrix}z_{mk,i}&z_{ml,j}\\z_{nk,i}&z_{nl,j}\end{pmatrix}
 +\det\begin{pmatrix}z_{mk,j}&z_{ml,i}\\z_{nk,j}&z_{nl,i}\end{pmatrix}
 \right\}.
\]
Let $\mathrm k(U)$ denote the induced map on quadratic coefficients.
Polarization and continuity give
\begin{equation}\label{errorforZ}
 \Omega(Z_c^*)\mathrm k((I-\mathcal Z)^{-\top})-\Omega(\widehat Z_c)=o_p(1).
\end{equation}
The diagonal columns vanish explicitly:
\[
 \Omega(Z_c^*)_{\cdot,(j,j)}=0,\qquad
 \Omega(\widehat Z_c)_{\cdot,(j,j)}=0.
\]
Equivalently,
for every column $u_j=(I-\mathcal Z)^{-\top}e_j$,
$\Psi_{Z_c^*}(u_j)=o_p(1)$.
The columns $u_j$ stay in a bounded annulus. Uniform identification on
the compact core sets implies that each approaches a coordinate ray.
In quadratic coordinates this gives
\[
 k_{\rm off}((I-\mathcal Z)^{-\top}e_j)=o_p(1),\qquad
 ((I-\mathcal Z)^{-\top})_{ij}((I-\mathcal Z)^{-\top})_{kj}=o_p(1)\quad(i\ne k).
\]
Their joint nonsingularity forces distinct rays, so $(I-\mathcal Z)^\top$ approaches a
nonsingular monomial matrix. The normalization of both product columns fixes the absolute monomial
scales. Indeed, $\Zh =(Z^*+R)(I-\mathcal Z)^{-\top}$, so if
$(I-\mathcal Z)^{-\top}e_j=\alpha_j e_{\pi(j)}+o_p(1)$, then
\[
 1=\|\Zh _{\cdot j}\|_2/\sqrt N=|\alpha_j|+o_p(1).
\]
Common permutations and compatible signs therefore give $\mathcal Z=o_p(1)$, and
\[
 \frac{\|\Zh -Z^*\|_F}{\sqrt N}
 \le\frac{\|\Zh \|}{\sqrt N}\|\mathcal Z\|_F+\frac{\|R\|_F}{\sqrt N}=o_p(1).
\]

For clarity, the separate-factor calculation is also retained.
After the product sign is fixed, write the normalized inner products
$\alpha_j=p^{-1}\widehat A_{\cdot j}^\top A_{\cdot j}^*$,
$\beta_j=q^{-1}\widehat B_{\cdot j}^\top B_{\cdot j}^*$.
Choose the remaining simultaneous factor sign $s_j$ so that
$s_j(\alpha_j+\beta_j)=|\alpha_j+\beta_j|$. Then
\begin{equation}\label{errorforB}
 \begin{split}
 &p^{-1}\|s_j\widehat A_{\cdot j}-A_{\cdot j}^*\|_2^2
  +q^{-1}\|s_j\widehat B_{\cdot j}-B_{\cdot j}^*\|_2^2\\
 &=4-2|\alpha_j+\beta_j|
 \le4(1-\alpha_j\beta_j)
 =2N^{-1}\|\Zh _{\cdot j}-Z_{\cdot j}^*\|_2^2.
 \end{split}
\end{equation}
This replaces an equality with differences of block norms, which
would otherwise discard signs.

The inequality uses $|\alpha+\beta|\ge2\alpha\beta$ for
$\alpha,\beta\in[-1,1]$, so it preserves the signs of both factors.
Summing over the columns proves their consistency. The product consistency
makes the candidate Gram constraint inactive with probability tending to one.
For fixed $\Zh,\Seh$, let
$C_{X,T}=\widehat H\Zh^\top\widehat W\widehat M_y\widehat W\Zh\widehat H$.
Then
\[
 \|\widehat W\Zh\widehat H\|=O(N^{-1/2}),\qquad
 \widehat H\Zh^\top\widehat W Z^*=(I-\mathcal Z)^\top=I+o_p(1),
\]
and
\[
 \begin{split}
 C_{X,T}={}&(I-\mathcal Z)^\top M_x^*(I-\mathcal Z)
       +\widehat H\Zh^\top\widehat W\Sigma_e^*\widehat W\Zh\widehat H\\
 &+\widehat H\Zh^\top\widehat W\cW_T\widehat W\Zh\widehat H
 =M_x^*+o_p(1).
 \end{split}
\]
The part of the negative likelihood depending on $M_x$ is
\begin{equation}\label{rev9:M-profile}
 \log|M_x+\widehat H |+\operatorname{tr}\{C_{X,T}(M_x+\widehat H )^{-1}\}.
\end{equation}
Its unique unrestricted minimizer is $C_{X,T}-\widehat H $; by the strict
spectral margin for $M_x^*$, this belongs to the interior of its box with
probability tending to one. Therefore global optimality gives
\begin{equation}\label{rev9:M-exact-profile}
 \Mh=C_{X,T}-\widehat H =M_x^*+o_p(1).
\end{equation}
This also justifies the $M_x$ score. For an arbitrary $U_A$, the path
\[
 \begin{gathered}
 D_A(t)=\operatorname{Diag}\{\sqrt p/\|\widehat A_{\cdot j}+tU_{A,\cdot j}\|_2\},\\
 A(t)=(\widehat A+tU_A)D_A(t),\qquad
 M_x(t)=D_A(t)^{-1}\Mh D_A(t)^{-1}
 \end{gathered}
\]
is feasible for both signs of sufficiently small $t$, and its covariance
is that of the unnormalized loading perturbation. Thus the unrestricted
$A$-direction score, and similarly the $B$-direction score, is zero.
Using $\Ph \Zh \Mh=\widehat W \Zh \Gh $, the $M_x$ score cancels the left Woodbury correction
and yields \eqref{gradcondforCP-A}--\eqref{gradcondforCP-B}.
The feasible segment $\Seh +t(\Sigma_e^*-\Seh )$ gives \eqref{gradcondforCP-e}.

\textbf{Step 5: (Proof of the convergence rate)}\label{step5forconvrate}
We start by bounding the rotation in terms of $R$. From
$\widehat Z_c=(Z_c^*+R_c)(I-\mathcal Z)^{-\top}$ and the local Lipschitz bound for the
quadratic minor map,
\begin{equation}\label{ineq:Omegasmoothandbounded}
 \|\Psi_{Z_c^*}((I-\mathcal Z)^{-\top}e_j)\|_2\le C\|R_c\|_F\le C\sqrt{x_T}.
\end{equation}
Writing $u_j=(I-\mathcal Z)^{-\top}e_j$ and using the polarized minor columns
$\psi_{ik}$, the same calculation is
\[
 \sum_{j=1}^d\left\|
       \sum_{i<k}(u_j)_i(u_j)_k\psi_{ik}(Z_c^*)\right\|_2^2
 \le C\|R_c\|_F^2,\qquad \psi_{ii}(Z_c^*)=0.
\]
More generally, if a local Hölder inverse of exponent $m$ is available
on the core family, this yields
$\|((I-\mathcal Z)^{-\top})_{\rm off}\|^{2m}\le Cx_T$.
In the uniformly transverse interior the direct expansion
\[
 \Psi_{Z_c^*}(e_j+h)=J_j(Z_c^*)h+\Psi_{Z_c^*}(h),
 \qquad \|\Psi_{Z_c^*}(h)\|_2\le C\|h\|_2^2
\]
gives $m=1$ uniformly: write $u_j/(u_j)_j=e_j+h$ with $h_j=0$.
The uniform lower bound for $\sigma_{\min}(J_j)$ in
\eqref{supp:uniform-interior} bounds $\|h\|_2$ by $C\sqrt{x_T}$
after absorbing $C\|h\|_2^2$ in a common consistency neighborhood.
The diagonal normalization can also be written without suppressing
any cross terms. With $Z^*=\Zh -\Zh \mathcal Z^\top-R$,
\[
 \begin{split}
 0={}&\operatorname{diag}(Z^{*\top}Z^*-\Zh^\top \Zh )\\
 ={}&-\operatorname{diag}\{
      \mathcal Z \Zh^\top \Zh +\Zh^\top \Zh \mathcal Z^\top
      -\mathcal Z \Zh^\top \Zh \mathcal Z^\top\}\\
 &+\operatorname{diag}\{
      -\Zh^\top R-R^\top \Zh +\mathcal Z \Zh^\top R
                  +R^\top \Zh \mathcal Z^\top+R^\top R\}.
 \end{split}
\]
Since $(\Zh^\top \Zh /N)_{jj}=1$, this gives
\begin{equation}\label{ineq:diagofZisoffdiagofZ}
 |\mathcal Z_{jj}|
 \le C\|\mathcal Z_{\rm off}\|+C\sqrt{x_T}+C\|\mathcal Z\|^2,
 \qquad j\le d.
\end{equation}
Absorbing the final term locally proves
$\|\mathcal Z\|^{2m}\le Cx_T$ whenever the stated Hölder inverse is valid.
For the interior rate, therefore,
\begin{equation}\label{rev9:inverse}
 \begin{gathered}
 \|\mathcal Z\|^2+N^{-1}\|\Zh -Z^*\|_F^2
 +p^{-1}\|\widehat A-A^*\|_F^2+q^{-1}\|\widehat B-B^*\|_F^2
 \le C_0x_T,\\
 x_T=N^{-1}\operatorname{tr}(R^\top \widehat W R).
 \end{gathered}
\end{equation}
Together with \eqref{errorforB}, the correctly normalized separate
loading calculation gives, whenever that uniform Hölder inverse holds,
\[
 \left(\frac{\|\widehat A-A^*\|_F^2}{p}\right)^m
 +\left(\frac{\|\widehat B-B^*\|_F^2}{q}\right)^m
 \le Cx_T.
\]

Secondly, we retain the explicit score calculation for $R$.
Substitute \eqref{errorforMx} into \eqref{eqforerrorofsigmay}.
The complete resulting expression is
\begin{equation}\label{orig9:profiled-covariance-error}
 \begin{split}
 \widehat\Sigma_y-\widehat M_y
 ={}&R M_{X,T}\Zh^\top+\Zh M_{X,T}R^\top\\
 &-R M_{X,T}\mathcal Z \Zh^\top-\Zh \mathcal Z^\top M_{X,T}R^\top
       -R M_{X,T}R^\top\\
 &+\mathcal T+\mathcal E_e
       -B_H\mathcal T B_H^\top-B_H\mathcal E_eB_H^\top.
 \end{split}
\end{equation}
Put $R_{bj}=J_jR$. Set $A_T=M_{X,T}(I-\mathcal Z) \widehat H^{-1}\Gh 
=M_{X,T}(I-\mathcal Z)(I-\Mh^{-1}\Gh )$.
The relations $R^\top \widehat W \Zh =0$ and $B_H^\top C_T=C_T$ imply
\[
 (\widehat\Sigma_y-\widehat M_y)C_T
 =RA_T+(I-B_H)(\mathcal T+\mathcal E_e)C_T.
\]
Hence the four-part A-score calculation is
\begin{equation}\label{eq:gradcondAexpressed}
 \begin{split}
 \sum_j \widehat\Sigma_{ej}^{-1}R_{bj}A_T\widehat b_j
 ={}&-\sum_j\widehat\Sigma_{ej}^{-1}J_j\mathcal T C_T\widehat b_j\\
 &+\sum_j\widehat\Sigma_{ej}^{-1}\widehat A\widehat b_j\widehat H \Zh^\top \widehat W \mathcal T C_T\widehat b_j\\
 &-\sum_j\widehat\Sigma_{ej}^{-1}J_j\mathcal E_e C_T\widehat b_j\\
 &+\sum_j\widehat\Sigma_{ej}^{-1}\widehat A\widehat b_j\widehat H \Zh^\top \widehat W \mathcal E_e C_T\widehat b_j.
 \end{split}
\end{equation}
The corresponding B-score is
\[
 \begin{split}
 \operatorname{diag}(\widehat A^\top \widehat\Sigma_{ej}^{-1}R_{bj}A_T)
 ={}&-\operatorname{diag}
     \{\widehat A^\top \widehat\Sigma_{ej}^{-1}(J_j-\widehat A\widehat b_j\widehat H \Zh^\top \widehat W )
                                     \mathcal T C_T\}\\
 &-\operatorname{diag}
     \{\widehat A^\top \widehat\Sigma_{ej}^{-1}(J_j-\widehat A\widehat b_j\widehat H \Zh^\top \widehat W )
                                      \mathcal E_e C_T\}.
 \end{split}
\]
Put $\delta A=\widehat A-A^*$ and
$\delta b_j=\widehat b_j-b_j^*$.
Expanding the left residual factor gives four terms:
\begin{equation}\label{eq:rsrequation}
 \begin{split}
 \sum_j R_{bj}^\top \widehat\Sigma_{ej}^{-1}R_{bj}A_T
 ={}&\sum_j\delta b_j\widehat A^\top \widehat\Sigma_{ej}^{-1}R_{bj}A_T\\
 &-\mathcal Z\sum_j\widehat b_j\widehat A^\top \widehat\Sigma_{ej}^{-1}R_{bj}A_T\\
 &+\sum_j(b_j^*-\widehat b_j)\delta A^\top \widehat\Sigma_{ej}^{-1}R_{bj}A_T\\
 &+\sum_j\widehat b_j\delta A^\top \widehat\Sigma_{ej}^{-1}R_{bj}A_T.
 \end{split}
\end{equation}
The second sum is zero by $\Zh^\top \widehat W R=0$.
Taking diagonals, commuting diagonal $\widehat b_j$ through the
diagonal operation, and inserting the A and B scores yields
\begin{equation}\label{eq:RSRisT}
 \begin{split}
 N^{-1}\operatorname{diag}(R^\top \widehat W RA_T)
 ={}&N^{-1}\operatorname{diag}
       \sum_j(b_j^*-\widehat b_j)\delta A^\top \widehat\Sigma_{ej}^{-1}R_{bj}A_T\\
 &-N^{-1}\operatorname{diag}\sum_j
       \delta A^\top \widehat\Sigma_{ej}^{-1}(J_j-\widehat A\widehat b_j\widehat H \Zh^\top \widehat W )
                   (\mathcal T+\mathcal E_e)C_T\widehat b_j\\
 &-N^{-1}\operatorname{diag}\sum_j
       \delta b_j\widehat A^\top \widehat\Sigma_{ej}^{-1}(J_j-\widehat A\widehat b_j\widehat H \Zh^\top \widehat W )
                   (\mathcal T+\mathcal E_e)C_T .
 \end{split}
\end{equation}
Lemma~\ref{lemma:eqRSRisT} below gives
\begin{equation}\label{ineq}
 \begin{split}
 N^{-1}|\operatorname{tr}(R^\top \widehat W RA_T)|
 \le{}& Cx_T^{1/m+1/2}\\
 &+O_p(x_T^{1/(2m)})
      \{T^{-1/2}+N^{-1/2}\|\mathcal E_e\|_F\}.
 \end{split}
\end{equation}
For $m=1$, $A_T=M_x^*+o_p(1)$ and the loading loss is coercive;
the exact spectral comparison is
\[
 \lambda_{\min}(M_x^*)\operatorname{tr}(R^\top \widehat W R)
 \le\operatorname{tr}(R^\top \widehat W RM_x^*)
 \le\lambda_{\max}(M_x^*)\operatorname{tr}(R^\top \widehat W R).
\]
Since $A_T=M_x^*+o_p(1)$, its trace differs from the middle expression by
$o_p(\operatorname{tr}(R^\top \widehat W R))$. Therefore
\[
\begin{gathered}
 cx_T\le o_p(x_T)+Cx_T^{3/2}
          +O_p(\sqrt{x_T})\{T^{-1/2}+N^{-1/2}\|\mathcal E_e\|_F\},\\
 x_T\le O_p(T^{-1})+O_p(N^{-1}\|\mathcal E_e\|_F^2).
\end{gathered}
\]
It does not close the rate by itself. For $m\ge2$ its nonlinear
term is not absorbable by this argument. We therefore use
the complete nuisance projection to obtain
the rate of $R$ first, as proved in Subsection~\ref{subsec:Est of R}.

Similarly, the diagonal variance condition can be expressed as follows.
Let $B_G=\Zh \Gh  \Zh^\top \widehat W $ and put
$\Gamma_D=\operatorname{Diag}(\widehat\sigma_r^4g_r)_{r=1}^N$.
Since $\Zh^\top \widehat W (\widehat\Sigma_y-\widehat M_y)\widehat W \Zh =0$, multiplication of the
covariance error by $\Seh \Ph=I-B_G$ on the left and $\Ph \Seh=I-B_G^\top$
on the right gives the following exact expanded identity:
\begin{equation}\label{eq:sigma_e-errorexpression}
 \begin{split}
 \Sigma_e^*-\Seh=\operatorname{Diag}\operatorname{diag}\big\{
 &R M_{X,T}(I-\mathcal Z)(I-\widehat H^{-1}\Gh )\Zh^\top\\
 &+\Zh (I-\Gh \widehat H^{-1})(I-\mathcal Z)^\top M_{X,T}R^\top\\
 &-R M_{X,T}R^\top
       +B_GR M_{X,T}R^\top+R M_{X,T}R^\top B_G^\top\\
 &+\mathcal T-B_H\mathcal T B_H^\top-B_H\mathcal E_eB_H^\top\\
 &-B_G\mathcal T(I-B_H^\top)
           -(I-B_H)\mathcal T B_G^\top\\
 &-B_G\mathcal E_e+B_G\mathcal E_eB_H^\top
           -\mathcal E_eB_G^\top+B_H\mathcal E_eB_G^\top
 \big\}-\Gamma_D .
 \end{split}
\end{equation}
The two $B_GR M_{X,T}R^\top$ terms vanish because $B_GR=0$;
they are displayed to retain the full expansion before cancellation.
The correction $-\Gamma_D$ is zero only on inactive coordinates.
By Subsection~\ref{subsec:Est of R}, $x_T=O_p(T^{-1})$. We now return to the variance calculation.

Write $s_T=\sqrt{Nx_T}$, $a_N=(1+s_T)/N$, and
$\ell_N=N^{-1/2}\max_r\|\Zh_{r,\cdot}\|_2=O_p(N^{-1/2}+\sqrt{x_T})=o_p(1)$;
these bounds follow from \eqref{rev9:inverse} and bounded true rows.
Since $I-\widehat H^{-1}\Gh=\Mh^{-1}\Gh$, the first linear term satisfies
\[
 \begin{split}
 &N^{-1}\sum_r\{R_{r,\cdot}M_{X,T}(I-\mathcal Z)
                       \Mh^{-1}\Gh\Zh_{r,\cdot}^\top\}^2\\
 &\quad\le CN^{-3}\sum_r\|R_{r,\cdot}\|_2^2\|\Zh_{r,\cdot}\|_2^2
 \le Ca_N^2s_T^2/N=o_p(x_T).
 \end{split}
\]
For $e_r=(\mathcal E_e)_{rr}$, $k_{ru}=(B_G)_{ru}$ and $b_{ru}=(B_H)_{ru}$,
$\|B_G\|+\|B_H\|\le C$ and
\[
 |k_{rr}|+|b_{rr}|\le C\ell_N^2,\qquad
 \max_r\sum_u(k_{ru}^2+b_{ru}^2)
 +\max_u\sum_r(k_{ru}^2+b_{ru}^2)\le C\ell_N^2.
\]
Thus the direct and projected variance terms give, respectively,
\[
 \begin{split}
 N^{-1}\sum_r\{\operatorname{diag}(\mathcal E_eB_G^\top)_r\}^2
 &=N^{-1}\sum_re_r^2k_{rr}^2\le C\ell_N^4N^{-1}\|\mathcal E_e\|_F^2,\\
 \|\operatorname{diag}(B_H\mathcal E_eB_H^\top)\|_2^2
 &=\sum_r\left(\sum_ub_{ru}^2e_u\right)^2\\
 &\le\left(\max_r\sum_ub_{ru}^2\right)
       \left(\max_u\sum_rb_{ru}^2\right)\sum_ue_u^2
 \le C\ell_N^4\|\mathcal E_e\|_F^2,\\
 \|\operatorname{diag}(B_G\mathcal E_eB_H^\top)\|_2^2
 &\le\left(\max_r\sum_u|k_{ru}b_{ru}|\right)
       \left(\max_u\sum_r|k_{ru}b_{ru}|\right)\sum_ue_u^2\\
 &\le C\ell_N^4\|\mathcal E_e\|_F^2.
 \end{split}
\]
The transpose terms have the same diagonal. For the quadratic term,
$R_{r,\cdot} M_{X,T}R_{r,\cdot}^\top\le\|M_{X,T}\|\|R_{r,\cdot}\|_2^2$, and the fourth-power
calculation retains its maximum-row factor:
\[
 \begin{split}
 N^{-1}\sum_r\|R_{r,\cdot}\|_2^4
 &=N^{-1}\sum_{\|R_{r,\cdot}\|\ge1}\|R_{r,\cdot}\|_2^4
   +N^{-1}\sum_{\|R_{r,\cdot}\|<1}\|R_{r,\cdot}\|_2^4\\
 &\le(1+\max_r\|R_{r,\cdot}\|_2^2)N^{-1}\sum_r\|R_{r,\cdot}\|_2^2.
 \end{split}
\]
To finish without assuming this maximum is bounded, use the exact
variance inequality \eqref{kkt9:variance-identity-inequality}. Put
$h=\sqrt{N/T}$, $d_D=\|\mathcal E_e\|_F$, and let $\mathcal F_0$ be the raw
$M_x$ residual defined in \eqref{kkt9:raw-M-projection}. The required
stochastic pairing is
\begin{equation}\label{kkt9:variance-stochastic-bound}
 |\langle\cW_T,\mathcal F_0(-\mathcal E_e)\rangle_{\Ph}|
 \le K_T h(d_D+s_T),\qquad K_T=O_p(1).
\end{equation}
Indeed, bounded true rows and the variance box imply
\[
 \|\mathcal E_e\Zh \|_F
 \le\|\mathcal E_eZ^*\|_F+\|\mathcal E_e\|\|\Zh -Z^*\|_F
 \le C(d_D+s_T).
\]
In the expansion $\Ph=\widehat W -Q_0$, write
\[
 Q_0=\widehat W \Zh \Gh  \Zh^\top \widehat W =C_T\Zh^\top \widehat W .
\]
Then
the pure diagonal term is bounded by
$C\|\operatorname{diag}\cW_T\|_2d_D\le K_Thd_D$.
Each cross term is bounded by
$\|\cW_TC_T\|_F\|\widehat W (-\mathcal E_e)\widehat W \Zh \|_F\le K_Th(d_D+s_T)$.
For the quadratic term, keep the exact identity
\[
 \operatorname{tr}(\cW_TQ_0(-\mathcal E_e)Q_0)
 =\operatorname{tr}\{C_T^\top\cW_TC_T \Zh^\top \widehat W (-\mathcal E_e)\widehat W \Zh \},
\]
whose absolute value is at most
$K_TT^{-1/2}\|\widehat W \Zh \|\|\mathcal E_e\widehat W \Zh \|_F\le K_Th(d_D+s_T)$.
The $M_x$ projection correction has coefficient norm at most $Cd_D$
and stochastic pairing at most $K_TT^{-1/2}d_D$.
This proves \eqref{kkt9:variance-stochastic-bound}.

For $RM_x^*R^\top$, preserve the sign of the quadratic term:
\begin{equation}\label{kkt9:quadratic-one-sided}
 \langle RM_x^*R^\top,-\mathcal E_e\rangle_{\Ph}
 =\operatorname{tr}(R^\top \widehat W (-\mathcal E_e)\widehat W RM_x^*)
 \ge-\operatorname{tr}(R^\top \widehat W RM_x^*)\ge-Cs_T^2,
\end{equation}
since $\widehat W (-\mathcal E_e)\widehat W =\widehat W \Sigma_e^*\widehat W -\widehat W \succeq-\widehat W $. Equivalently,
\[
 \langle RM_x^*R^\top,-\mathcal E_e\rangle_{\Ph}+\operatorname{tr}(R^\top \widehat W RM_x^*)
 =\sum_r\frac{\sigma_r^{*2}}{\widehat\sigma_r^4}R_{r,\cdot} M_x^*R_{r,\cdot}^\top\ge0.
\]
Combining \eqref{kkt9:variance-identity-inequality}, its coercivity bound,
and the linear diagonal bound from Subsection~\ref{subsec:Est of R} gives
\begin{equation}\label{kkt9:variance-only-absorption}
 cd_D^2\le Cs_T^2+Ca_Ns_Td_D+K_Th(d_D+s_T).
\end{equation}
The already proved $s_T=O_p(h)$ and Young's inequality yield
$d_D^2=O_p(N/T)$, hence $N^{-1}\operatorname{tr}(\mathcal E_e^2)=O_p(T^{-1})$.

Finally, return to \eqref{errorforMx}:
\begin{equation}\label{rev9:M-rate-identity}
 \begin{split}
 \Mh-M_x^*={}&(I-\mathcal Z)^\top M_x^*(I-\mathcal Z)-M_x^*
                   +(I-\mathcal Z)^\top F_T(I-\mathcal Z)\\
 &+\widehat H\Zh^\top\widehat W Z^*G_T\widehat W\Zh\widehat H\\
 &+\widehat H\Zh^\top\widehat W G_T^\top Z^{*\top}\widehat W\Zh\widehat H\\
 &+\widehat H\Zh^\top\widehat W E_T\widehat W\Zh\widehat H\\
 &+\widehat H\Zh^\top\widehat W(\Sigma_e^*-\Seh)\widehat W\Zh\widehat H.
 \end{split}
\end{equation}
The first term is $O_p(\|\mathcal Z\|+\|\mathcal Z\|^2)=O_p(T^{-1/2})$,
and $\|(I-\mathcal Z)^\top F_T(I-\mathcal Z)\|_F=O_p(T^{-1/2})$.
The two cross terms are bounded by $C\|G_T\|_F/\sqrt N$,
the pure error term by $C\|E_T\|_F/N$, and the last term by
$Cd_D/N$. Thus $\|\Mh -M_x^*\|_F=O_p(T^{-1/2})$, after the loading
and variance rates have been established.

\textbf{Step 6: (Fixed-dimension Case)}
For fixed $p$ and $q$, the parameter dimension is fixed. Assume the covariance identification condition in Proof condition~\ref{supp:fixed-covariance}(i). The population criterion is uniquely maximized at $\theta^*$ modulo the stated normalizations, since
\[
 \log|\Sigma|+\tr(\Sigma_y^*\Sigma^{-1})
 \ge \log|\Sigma_y^*|+pq,
\]
with equality only at $\Sigma=\Sigma_y^*$. Compactness and
\[
 \sup_{\theta\in\Theta_N}|\mathcal L_T(\theta)-\mathcal L(\theta)|
 \le \frac{\|\widehat M_y-\Sigma_y^*\|_F}{2pq}
       \sup_{\theta\in\Theta_N}\|\Sigma_y(\theta)^{-1}\|_F
 =O_p(T^{-1/2})
\]
therefore give consistency. For the rate, assume part (ii) and use smooth normalized local coordinates after permutation and sign alignment.

The bars denote the sample and population criteria in these local coordinates, with the common constant $\tfrac12+(2pq)^{-1}\log|\Sigma_y^*|$ added. Write
\[
    \bar{\mathcal L}_T(\theta)
    =
    \bar{\mathcal L}(\theta)+R_T(\theta),
\]
where
\[
    R_T(\theta)
    =
    -\frac{1}{2pq}
    \tr\left\{
        \mathcal W_T\Sigma_y(\theta)^{-1}
    \right\},
    \qquad
    \mathcal W_T=\widehat M_y-\Sigma_y^*,
\]
and let
\[
    P(\theta)=\Sigma_y(\theta)^{-1}.
\]

Under Proof condition~\ref{supp:likelihood-conditions},
the latent factor process is strongly mixing with the stated
$4+\delta$ moment condition, while the idiosyncratic errors are
independent over time and independent of the latent factors.
Consequently, for every fixed pair of coordinates $(i,j)$,
\[
    (\mathcal W_T)_{ij}
    =
    \frac1T\sum_{t=1}^T
    \left\{
        y_{t,i}y_{t,j}
        -
        E(y_{t,i}y_{t,j})
    \right\}-\bar y_i\bar y_j
    =
    O_p(T^{-1/2}).
\]
Indeed, the covariance series of each summand is absolutely summable
under Proof condition~\ref{supp:likelihood-conditions},
so its sample average has variance of order $T^{-1}$; the centering term is $O_p(T^{-1})$.
Since $pq$ is fixed,
\begin{equation}
    \|\mathcal W_T\|_F
    =
    O_p(T^{-1/2}).
    \label{S.1}
\end{equation}

For a direction $h$, write
\[
    \dot\Sigma
    =
    d\Sigma_y(\theta)[h].
\]
Using
\[
    dP(\theta)[h]
    =
    -P(\theta)\dot\Sigma P(\theta),
\]
we have
\[
    dR_T(\theta)[h]
    =
    \frac{1}{2pq}
    \tr\left\{
        \mathcal W_TP(\theta)\dot\Sigma P(\theta)
    \right\}.
\]
Hence, by the Cauchy--Schwarz inequality,
\[
\begin{aligned}
    |dR_T(\theta)[h]|
    &\leq
    \frac{1}{2pq}
    \|\mathcal W_T\|_F
    \|P(\theta)\dot\Sigma P(\theta)\|_F.
\end{aligned}
\]
Under Proof condition~\ref{supp:likelihood-conditions},
the covariance matrices are uniformly positive definite on a sufficiently small closed convex coordinate neighborhood $\mathcal U$ of the truth and the first derivatives of
$\Sigma_y(\theta)$ are uniformly bounded.  Therefore,
\[
    \sup_{\theta\in\mathcal U}
    \sup_{\|h\|=1}
    \|P(\theta)\dot\Sigma P(\theta)\|_F
    \leq C,
\]
and consequently
\begin{equation}
    \sup_{\theta\in\mathcal U}
    \|\nabla R_T(\theta)\|
    =
    O_p(T^{-1/2}).
    \label{S.2}
\end{equation}
In particular, since
$\nabla\bar{\mathcal L}(\theta^*)=0$,
\begin{equation}
    \nabla\bar{\mathcal L}_T(\theta^*)
    =
    \nabla R_T(\theta^*)
    =
    O_p(T^{-1/2}).
    \label{S.3}
\end{equation}

We next control the Hessian.  For two directions $h_1,h_2$, let
\[
    \dot\Sigma_i
    =
    d\Sigma_y(\theta)[h_i],
    \qquad i=1,2,
\]
and
\[
    \ddot\Sigma_{12}
    =
    d^2\Sigma_y(\theta)[h_1,h_2].
\]
Differentiating once more gives
\[
\begin{aligned}
    d^2R_T(\theta)[h_1,h_2]
    =
    \frac{1}{2pq}
    \tr\Big[
        \mathcal W_T\{
        &P\ddot\Sigma_{12}P
        -P\dot\Sigma_2P\dot\Sigma_1P\\
        &-P\dot\Sigma_1P\dot\Sigma_2P
        \}
    \Big],
\end{aligned}
\]
where $P=P(\theta)$.  Hence
\[
\begin{aligned}
    |d^2R_T(\theta)[h_1,h_2]|
    \leq
    \frac{\|\mathcal W_T\|_F}{2pq}
    \Big\|
        &P\dot\Sigma_2P\dot\Sigma_1P
        +
        P\dot\Sigma_1P\dot\Sigma_2P\\
        &-
        P\ddot\Sigma_{12}P
    \Big\|_F.
\end{aligned}
\]
By Proof condition~\ref{supp:likelihood-conditions},
$P(\theta)$ and the first two derivatives of
$\Sigma_y(\theta)$ are uniformly bounded on $\mathcal U$.  Therefore,
for $\|h_1\|=\|h_2\|=1$,
\[
    \sup_{\theta\in\mathcal U}
    \Big\|
        P\dot\Sigma_2P\dot\Sigma_1P
        +
        P\dot\Sigma_1P\dot\Sigma_2P
        -
        P\ddot\Sigma_{12}P
    \Big\|_F
    \leq C.
\]
Together with \eqref{S.1}, this yields
\begin{equation}
    \sup_{\theta\in\mathcal U}
    \|\nabla^2R_T(\theta)\|_{\op}
    =
    O_p(T^{-1/2})
    =
    o_p(1).
    \label{S.4}
\end{equation}

Since consistency has already been established,
\[
    \widehat\theta\xrightarrow{p}\theta^*.
\]
As $\theta^*$ is an interior point, the first-order condition gives
\[
    \nabla\bar{\mathcal L}_T(\widehat\theta)=0.
\]
An integrated Taylor expansion around $\theta^*$ yields
\[
    0=\nabla\bar{\mathcal L}_T(\theta^*)
      +\mathbb H_T(\widehat\theta-\theta^*),\qquad
    \mathbb H_T=\int_0^1\nabla^2\bar{\mathcal L}_T
      \{\theta^*+u(\widehat\theta-\theta^*)\}\,du.
\]
Moreover,
\[
\begin{split}
    \mathbb H_T={}&\int_0^1\nabla^2\bar{\mathcal L}
      \{\theta^*+u(\widehat\theta-\theta^*)\}\,du\\
    &+\int_0^1\nabla^2R_T
      \{\theta^*+u(\widehat\theta-\theta^*)\}\,du.
\end{split}
\]
By consistency, continuity of the population Hessian, and
\eqref{S.4},
\[
    \mathbb H_T
    =
    \nabla^2\bar{\mathcal L}(\theta^*)+o_p(1).
\]
By Proof condition~\ref{supp:fixed-covariance}(ii), the population Hessian is nonsingular, so
\[
    \left\|
        \{\mathbb H_T\}^{-1}
    \right\|_{\op}
    =
    O_p(1).
\]
Combining this with \eqref{S.3},
\[
\begin{aligned}
    \|\widehat\theta-\theta^*\|
    &\leq
    \left\|
        \{\mathbb H_T\}^{-1}
    \right\|_{\op}
    \|\nabla\bar{\mathcal L}_T(\theta^*)\|\\
    &=
    O_p(T^{-1/2}).
\end{aligned}
\]

The smooth coordinate map transfers this rate to the parameter errors in Theorem~9.

\end{proof}

\subsection{Lemma}\label{sec:lemma}

Write $e_{v,t}=e_t$ for the vectorized errors, and let
$G_T^0=T^{-1}\sum_t x_te_t^\top$ and
$E_T^0=T^{-1}\sum_t(e_te_t^\top-\Sigma_e^*)$ be the raw versions
of the centered primitives. The superscript $0$ is omitted only when
the displayed calculation already uses $P_T$.
In coordinate contractions, $e_r$ is the $r$th standard basis vector,
$z_r=Z^\top e_r$ and $z_r^*=Z^{*\top}e_r$ are column vectors, and $\Zh_{r,\cdot}=e_r^\top \Zh $
is a row vector.

\begin{lemma}\label{lemmaRformatrixfactor}
Uniformly over $\theta\in\Theta_N$, the following bounds hold:
\begin{align*}
 {\rm(a)}\quad&
 N^{-1}|\operatorname{tr}\{(Z^*G_T^0+G_T^{0\top}Z^{*\top})
                                      \Sigma_y(\theta)^{-1}\}|
       =O_p(T^{-1/2}),\\
 {\rm(b)}\quad&
 N^{-1}|\operatorname{tr}(E_T^0\Sigma_y(\theta)^{-1})|
       =O_p(T^{-1/2}),\\
 {\rm(c)}\quad&
 N^{-1}\operatorname{tr}(\bar e\bar e^\top\Sigma_y(\theta)^{-1})
       =O_p(T^{-1}),\\
 {\rm(d)}\quad&
 N^{-1}|\operatorname{tr}(Z^*F_TZ^{*\top}\Sigma_y(\theta)^{-1})|
       =O_p(T^{-1/2}).
\end{align*}
The same bounds in (a)--(b) hold for the centered versions.
\end{lemma}
\begin{proof}
For a candidate parameter write
$W_Z=\Sigma_e(\theta)^{-1}$,
$H_Z=(Z^\top W_ZZ)^{-1}$,
$G_Z=(M_x(\theta)^{-1}+H_Z^{-1})^{-1}$, and
\[
 Q(\theta)=W_ZZG_ZZ^\top W_Z,\qquad
 \Sigma_y(\theta)^{-1}=W_Z-Q(\theta).
\]
The matrix $Q(\theta)$ is positive semidefinite of rank at most $d$,
satisfies $Q(\theta)\preceq W_Z$, and has uniformly bounded Frobenius norm.

For (a), the two parts of the Woodbury calculation are
\[
 \begin{split}
 &N^{-1}\operatorname{tr}\{(Z^*G_T^0+G_T^{0\top}Z^{*\top})
                                \Sigma_y(\theta)^{-1}\}\\
 &=\frac2N\sum_{r=1}^N\sum_{\ell=1}^d
       (W_Z)_{rr}z^*_{r\ell}\left(T^{-1}\sum_t x_{\ell t}e_{rt}\right)
       -\frac2N\operatorname{tr}(Z^{*\top}Q(\theta)G_T^{0\top}).
 \end{split}
\]
Coordinatewise Cauchy--Schwarz followed by a second Cauchy--Schwarz gives
\[
 \begin{split}
 &\frac1N\sum_r\left|
 (W_Z)_{rr}\sum_\ell z^*_{r\ell}
                  T^{-1}\sum_tx_{\ell t}e_{rt}\right|\\
 &\le \frac C N\sum_r\|z_r^*\|_2\|(G_T^0)_{\cdot r}\|_2\\
 &\le C\left(N^{-1}\sum_r\|z_r^*\|_2^2\right)^{1/2}
       \left(N^{-1}\sum_r\|(G_T^0)_{\cdot r}\|_2^2\right)^{1/2}
 =O_p(T^{-1/2}).
 \end{split}
\]
For the second part, retain the normalized-Gram factorization
\[
 G_Z=H_Z^{1/2}
   (I+H_Z^{1/2}M_x(\theta)^{-1}H_Z^{1/2})^{-1}H_Z^{1/2},
 \quad H_Z^{1/2}Z^\top W_ZZH_Z^{1/2}=I.
\]
Consequently $\|W_ZZH_Z^{1/2}\|_F^2\le Cd$ and the middle
inverse has norm at most one. Hence, explicitly,
\[
 \begin{split}
 N^{-1}|\operatorname{tr}(Z^{*\top}Q(\theta)G_T^{0\top})|
 &\le N^{-1}\|Z^{*\top}W_ZZH_Z^{1/2}\|_F\\
 &\quad\cdot
 \|(I+H_Z^{1/2}M_x(\theta)^{-1}H_Z^{1/2})^{-1}\|\\
 &\quad\cdot\|H_Z^{1/2}Z^\top W_ZG_T^{0\top}\|_F\\
 &\le CN^{-1}\|Z^*\|_F\|G_T^0\|_F
 =O_p(T^{-1/2}).
 \end{split}
\]

For (b), split the expression into its diagonal and low-rank parts:
\[
 N^{-1}\operatorname{tr}(E_T^0\Sigma_y(\theta)^{-1})
 =N^{-1}\sum_r(W_Z)_{rr}(E_T^0)_{rr}
           -N^{-1}\operatorname{tr}(E_T^0Q(\theta)).
\]
The first term is bounded by
\[
 \left\{N^{-1}\sum_r
       [T^{-1}\sum_t(e_{rt}^2-\sigma_r^{*2})]^2\right\}^{1/2}
 \left\{N^{-1}\sum_r(W_Z)_{rr}^2\right\}^{1/2}
 =O_p(T^{-1/2}).
\]
Expanding the second term over both coordinates gives
\[
 \begin{split}
 \operatorname{tr}(E_T^0Q(\theta))
 =\sum_{r,u}(W_Z)_{rr}(W_Z)_{uu}
       \left[T^{-1}\sum_t
       \{e_{rt}e_{ut}-\mathbf1_{\{r=u\}}\sigma_r^{*2}\}\right]
       z_u^\top G_Zz_r.
 \end{split}
\]
Cauchy--Schwarz over $(r,u)$ therefore gives
\[
 \begin{split}
 N^{-1}|\operatorname{tr}(E_T^0Q(\theta))|
 &\le
 \left\{N^{-2}\sum_{r,u}
     [T^{-1}\sum_t\{e_{rt}e_{ut}
                    -\mathbf1_{\{r=u\}}\sigma_r^{*2}\}]^2\right\}^{1/2}
       \|Q(\theta)\|_F\\
 &\le C\|E_T^0\|_F/N=O_p(T^{-1/2}).
 \end{split}
\]

For (c), $\Sigma_y(\theta)^{-1}\preceq W_Z\preceq CI$ gives
\[
 N^{-1}\operatorname{tr}(\bar e\bar e^\top\Sigma_y(\theta)^{-1})
 \le CN^{-1}\|\bar e\|_2^2=O_p(T^{-1}).
\]
Finally,
\[
 N^{-1}|\operatorname{tr}(Z^*F_TZ^{*\top}\Sigma_y(\theta)^{-1})|
 \le CN^{-1}\|Z^*\|_F^2\|F_T\|=O_p(T^{-1/2}).
\]
The mixed centering term has normalized trace bounded by
$CN^{-1}\|Z^*\|_F\|\bar x\|_2\|\bar e\|_2=O_p(T^{-1})$;
the pure error centering term is (c). Thus no mean correction or
empirical factor-covariance term has been dropped.
\end{proof}

\begin{lemma}\label{lemmaforHerror}
Writing $\mathcal E^2=N^{-1}\sum_r(\widehat\sigma_r^2/\sigma_r^{*2}-1)^2$,
\[
 \|N^{-1}Z^{*\top}(\widehat W -\Sigma_e^{*-1})Z^*\|\le C\mathcal E.
\]
\end{lemma}
\begin{proof}
The coordinate calculation is
\[
 \begin{split}
 &\left\|N^{-1}\sum_r z_r^*z_r^{*\top}
                  \frac{\sigma_r^{*2}-\widehat\sigma_r^2}{\widehat\sigma_r^2\sigma_r^{*2}}\right\|\\
 &\le C\left(N^{-1}\sum_r\|z_r^*\|_2^4\right)^{1/2}
          \left(N^{-1}\sum_r(\widehat\sigma_r^2/\sigma_r^{*2}-1)^2\right)^{1/2}
 \le C\mathcal E.
 \end{split}
\]
\end{proof}

\begin{lemma}\label{lemmaforfactmodconsistency1}
Put $d_D=\|\Seh -\Sigma_e^*\|_F$, and let $e_r$ in this lemma
denote the $r$th coordinate basis vector, not an error observation.
For each deterministic coordinate
$r=r_{N,T}$, the following nine centered bounds hold:
\begin{align*}
 {\rm(a)}\quad&\widehat H \Zh^\top \widehat W  Z^*G_Te_r=O_p(T^{-1/2}),\\
 {\rm(b)}\quad&\widehat H \Zh^\top \widehat W  G_T^\top Z^{*\top}e_r=O_p(T^{-1/2}),\\
 {\rm(c)}\quad&\widehat H \Zh^\top \widehat W  E_Te_r=O_p(T^{-1/2}),\\
 {\rm(d)}\quad&\|\widehat H \Zh^\top \widehat W (\Sigma_e^*-\Seh )e_r\|_2
       \le C\|\widehat H \|\|\Zh_{r,\cdot}\|_2,\\
 {\rm(e)}\quad&\|\sqrt N \widehat H^{1/2}\|^2=O_p(1),\\
 {\rm(f)}\quad&\widehat H \Zh^\top \widehat W  Z^*G_T(\widehat W \Zh \widehat H )=O_p(T^{-1/2}),\\
 {\rm(g)}\quad&\widehat H \Zh^\top \widehat W  G_T^\top Z^{*\top}(\widehat W \Zh \widehat H )=O_p(T^{-1/2}),\\
 {\rm(h)}\quad&\widehat H \Zh^\top \widehat W  E_T(\widehat W \Zh \widehat H )=O_p(T^{-1/2}),\\
 {\rm(i)}\quad&\|\widehat H \Zh^\top \widehat W (\Sigma_e^*-\Seh )(\widehat W \Zh \widehat H )\|_F\le Cd_D/N.
\end{align*}
All probability bounds are uniform
over the choice of a deterministic coordinate, without asserting a
bound on the maximum over coordinates.
\end{lemma}
\begin{proof}
For (a), expand the weighted factor projection:
\[
 \widehat H \Zh^\top \widehat W  Z^*G_Te_r
 =\widehat H^{1/2}\sum_u \widehat W_{uu}\widehat H^{1/2}\Zh_{u,\cdot}^\top z_u^{*\top}\,G_Te_r.
\]
Cauchy--Schwarz gives
\[
 \|\widehat H^{1/2}\Zh^\top \widehat W Z^*\|
 \le C\{\operatorname{tr}(\widehat H^{1/2}\Zh^\top \widehat W \Zh \widehat H^{1/2})\}^{1/2}
        \|Z^*\|_F\le C\sqrt N.
\]
Since $\|G_Te_r\|_2=O_p(T^{-1/2})$ and
$\|\sqrt N \widehat H^{1/2}\|=O_p(1)$, this proves (a).
For (b), keep the coordinate expansion and its norm bound:
\[
\begin{gathered}
 \widehat H \Zh^\top \widehat W  G_T^\top Z^{*\top}e_r
 =\widehat H^{1/2}\sum_u \widehat W_{uu}\widehat H^{1/2}\Zh_{u,\cdot}^\top
          (G_Te_u)^\top z_r^*,\\
 \|\widehat H \Zh^\top \widehat W  G_T^\top z_r^*\|_2
 \le\|(\widehat W \Zh \widehat H )\|\|G_T\|_F\|z_r^*\|_2=O_p(T^{-1/2}).
\end{gathered}
\]
For (c), independence and the fourth moment bound give
$E\|E_T^0e_r\|_2^2\le CN/T$, while the centered correction satisfies
$\|\bar e\,\bar e_r\|_2=O_p(\sqrt N/T)$. Hence
\[
 \|\widehat H \Zh^\top \widehat W  E_Te_r\|_2
 \le\|(\widehat W \Zh \widehat H )\|\|E_Te_r\|_2
 \le O_p(N^{-1/2})O_p(\sqrt{N/T}).
\]
For (d), the exact coordinate product is
\[
 \widehat H \Zh^\top \widehat W (\Sigma_e^*-\Seh )e_r
 =\widehat H  \Zh_{r,\cdot}^\top \widehat W_{rr}(\sigma_r^{*2}-\widehat\sigma_r^2),
\]
which proves its stated bound. For (e), the candidate Gram gives
$cI\preceq N^{-1}\Zh^\top \widehat W \Zh \preceq CI$ and therefore
$\operatorname{tr}(N\widehat H )\le C$; equivalently the alignment
identity \eqref{eq:HhatandH*} gives the same bound after consistency.
For (f)--(h), the individual estimates are
\[
 \begin{split}
 \|\widehat H \Zh^\top \widehat W  Z^*G_T(\widehat W \Zh \widehat H )\|_F
 &\le\|I-\mathcal Z\|\|G_T\|_F\|(\widehat W \Zh \widehat H )\|=O_p(T^{-1/2}),\\
 \|\widehat H \Zh^\top \widehat W  G_T^\top Z^{*\top}(\widehat W \Zh \widehat H )\|_F
 &\le\|(\widehat W \Zh \widehat H )\|\|G_T\|_F\|I-\mathcal Z\|=O_p(T^{-1/2}),\\
 \|\widehat H \Zh^\top \widehat W  E_T(\widehat W \Zh \widehat H )\|_F
 &\le\|(\widehat W \Zh \widehat H )\|^2\|E_T\|_F=O_p(T^{-1/2}).
 \end{split}
\]
Finally the coordinate expansion in (i) is
\[
 \widehat H \Zh^\top \widehat W (\Sigma_e^*-\Seh )(\widehat W \Zh \widehat H )
 =\widehat H \left\{\sum_r \widehat W_{rr}^2(\sigma_r^{*2}-\widehat\sigma_r^2)\Zh_{r,\cdot}^\top \Zh_{r,\cdot}\right\}\widehat H .
\]
$\|\widehat H \Zh^\top \widehat W (\Sigma_e^*-\Seh )(\widehat W \Zh \widehat H )\|_F
 \le\|(\widehat W \Zh \widehat H )\|^2d_D\le Cd_D/N$.
\end{proof}

\begin{lemma}\label{lemmaforfactmodconsistency2}
The four further contractions satisfy
\begin{align*}
 {\rm(a)}\quad&N^{-1}\widehat H \Zh^\top \widehat W  Z^*G_TZ^*=O_p(T^{-1/2}),\\
 {\rm(b)}\quad&N^{-1}\widehat H \Zh^\top \widehat W  G_T^\top Z^{*\top}Z^*=O_p(T^{-1/2}),\\
 {\rm(c)}\quad&N^{-1}\widehat H \Zh^\top \widehat W  E_TZ^*=O_p(T^{-1/2}),\\
 {\rm(d)}\quad&
 N^{-1}\|\widehat H \Zh^\top \widehat W (\Sigma_e^*-\Seh )Z^*\|_F\le Cd_D/N^{3/2}.
\end{align*}
The last bound is $O_p((N\sqrt T)^{-1})$ after the variance rate is
proved; no rate for $\Seh $ is assumed at this stage.
\end{lemma}
\begin{proof}
For (a) the columnwise calculation is
\[
 \begin{split}
 N^{-1}\|\widehat H \Zh^\top \widehat W  Z^*G_TZ^*\|_F
 &\le
 \{N^{-1}\sum_r\|(I-\mathcal Z)^\top G_Te_r\|_2^2\}^{1/2}
 \{N^{-1}\sum_r\|z_r^*\|_2^2\}^{1/2}\\
 &\le C\|G_T\|_F/\sqrt N=O_p(T^{-1/2}).
 \end{split}
\]
For (b) and (c), respectively,
\[
 \begin{split}
 N^{-1}\|\widehat W\Zh\widehat H\|\|G_T\|_F\|Z^*\|^2&=O_p(T^{-1/2}),\\
 N^{-1}\|\widehat W\Zh\widehat H\|\|E_T\|_F\|Z^*\|&=O_p(T^{-1/2}).
 \end{split}
\]
For (d), bounded true rows imply
$\|(\Sigma_e^*-\Seh )Z^*\|_F^2
 =\sum_r(\sigma_r^{*2}-\widehat\sigma_r^2)^2\|z_r^*\|_2^2\le Cd_D^2$;
multiplication by $N^{-1}\|(\widehat W \Zh \widehat H )\|\le CN^{-3/2}$ proves the claim.
\end{proof}

\begin{lemma}\label{lemma:eqRSRisT}
Suppose a local Hölder inverse gives the normalized separate loading
errors $O(x_T^{1/m})$, with $m=1$ in the uniformly transverse interior.
The three terms on the right of \eqref{eq:RSRisT}, in their displayed
order, have norms bounded by
\[
 Cx_T^{1/m+1/2},\qquad
 O_p(x_T^{1/(2m)})\{T^{-1/2}+d_D/N\},\qquad
 O_p(x_T^{1/(2m)})\{T^{-1/2}+d_D/N\}.
\]
In particular they satisfy the weaker bounds displayed in
\eqref{ineq}.
\end{lemma}
\begin{proof}
Define the two separate loading directions and their quadratic product
\[
 E_A=\widehat B\odot\delta A,\quad
 E_B=\delta B\odot\widehat A,\quad
 Q=\delta B\odot\delta A.
\]
Column normalization and the local inverse give, separately,
\[
 \|E_A\|_F^2=q\|\delta A\|_F^2\le CNx_T^{1/m},\qquad
 \|E_B\|_F^2=p\|\delta B\|_F^2\le CNx_T^{1/m},
\]
\[
 \|Q\|_F^2
 =\sum_\ell\|\delta A_{\cdot\ell}\|_2^2
                \|\delta B_{\cdot\ell}\|_2^2
 \le\|\delta A\|_F^2\|\delta B\|_F^2
 \le CNx_T^{2/m}.
\]
The first block sum is $-Q^\top \widehat W RA_T$. Thus
\[
 N^{-1}\|\operatorname{diag}(Q^\top \widehat W RA_T)\|_2
 \le CN^{-1}\|Q\|_F\|R\|_F
 \le Cx_T^{1/m+1/2}.
\]
For the second term, retaining the direct and projected parts separately,
\[
 \begin{split}
 &N^{-1}\|\operatorname{diag}\{
                  E_A^\top \widehat W (I-B_H)(\mathcal T+\mathcal E_e)C_T\}\|_2\\
 &\le CN^{-1}\|E_A\|_F\|\mathcal T C_T\|_F
      +CN^{-1}\|E_A\|_F\|B_H\|\|\mathcal T C_T\|_F\\
 &\quad+CN^{-1}\|E_A\|_F\|\mathcal E_e C_T\|_F
      +CN^{-1}\|E_A\|_F\|B_H\|\|\mathcal E_e C_T\|_F .
 \end{split}
\]
The primitive bounds and $\|C_T\|\le C/\sqrt N$ imply
$\|\mathcal T C_T\|_F=O_p(\sqrt{N/T})$ and
$\|\mathcal E_eC_T\|_F\le Cd_D/\sqrt N$; also $\|B_H\|\le C$.
This gives its asserted bound. Replacing $E_A$ by $E_B$ gives the
third term with the same four explicit estimates. All inequalities
are norm inequalities for possibly random fitted directions.
\end{proof}

\subsection{Estimation of the term R}\label{subsec:Est of R}

Let $N=pq$.  The target is
\[
 x_T:=\frac1N\tr(R^\top\Seh^{-1}R),
 \qquad
 R=\Zh(I_d-\mathcal Z)^\top-Z^*.
\]
The proof uses the loading, $M_x$, and diagonal-$\Sigma_e$ score conditions in
\eqref{gradcondforCP-A}--\eqref{gradcondforCP-e}.  The key operation is orthogonally projecting a loading covariance direction onto the complement of the nuisance covariance space generated by diagonal $\Sigma_e$ changes and first-order $M_x$ changes.

\subsubsection{Likelihood inner product and score directions}

For matrices of the same size, define the Frobenius inner product
\[
 \ip{A}{B}_F=\tr(A^\top B).
\]
For symmetric matrices $K_1,K_2$ and a positive definite matrix $P$, define
\begin{equation}
 \ip{K_1}{K_2}_P=\tr(PK_1PK_2),
 \qquad
 \norm{K}_P^2=\ip{K}{K}_P.
 \label{eq:P-inner-product}
\end{equation}
This is the ordinary Frobenius inner product after whitening by $P^{1/2}$.

At the QMLE, write
\[
 \Ph=\Syh^{-1},
 \qquad
 \Sh=\Ph(\Syh-\widehat{M}_{y})\Ph.
\]

\begin{lemma}[Directional score equations]
\label{lem:score-equations}
Let
\[
 \dot Z=\widehat B\odot\dot A+\dot B\odot\widehat A
\]
be a feasible loading tangent.  Then
\begin{align}
 \ip{\Sh}{\dot Z\Mh\Zh^\top+\Zh\Mh\dot Z^\top}_F&=0,
 \label{eq:loading-score}\\
 \ip{\Sh}{\Zh U\Zh^\top}_F&=0,
 \qquad U=U^\top,
 \label{eq:M-score}\\
 \ip{\Sh}{\Sigma_e^*-\Seh }_F&\ge0.
 \label{eq:diagonal-score}
\end{align}
\end{lemma}

\begin{proof}
For
\[
 \mathcal{L}_T(\Sigma)=-\frac1{2N}\{\log|\Sigma|+\tr(\widehat{M}_{y}\Sigma^{-1})\},
\]
we have
\[
 \left.\frac{d}{ds}\mathcal{L}_T(\Sigma+sK)\right|_{s=0}
 =-\frac1{2N}\tr\{\Sigma^{-1}(\Sigma-\widehat{M}_{y})\Sigma^{-1}K\}.
\]
The normalized loading paths and inactive $M_x$ constraint were verified
in Step 4. Use these covariance derivatives and the feasible segment
$\Seh +t(\Sigma_e^*-\Seh )$ for the final inequality.
\end{proof}

\subsubsection{Primitive $\alpha$-mixing condition for the factor process}

Let $u_t=(M_x^*)^{-1/2}x_t$
so that $E(u_tu_t^\top)=I_d$.  For $h\geq1$, define the strong-mixing coefficient
\begin{equation}
 \alpha_x(h)
 =\sup_{A\in\sigma(u_s:s\leq0),\,B\in\sigma(u_s:s\geq h)}
 \left|P(A\cap B)-P(A)P(B)\right|.
 \label{eq:alpha-definition}
\end{equation}

When the factor law varies with $T$, subscripts $T$ and all bounds below
are understood uniformly over the triangular array. The factor process $\{u_t\}_{t\in\mathbb Z}$ is strictly stationary,
$E u_t=0$, $E(u_tu_t^\top)=I_d$, and the $\alpha$-mixing condition is
\begin{equation}
 \sup_{T,t} E\norm{u_{T,t}}^{4+\delta}<\infty,
 \qquad
 \sup_T\sum_{h=1}^{\infty}
 \alpha_{x,T}(h)^{\delta/(4+\delta)}<\infty.
 \label{eq:alpha-summability}
\end{equation}

\begin{lemma}
\label{lem:alpha-implies-quadratic}
Under Proof condition~\ref{supp:likelihood-conditions},
\[
 \norm{\frac1T\sum_{t=1}^T(u_tu_t^\top-I_d)}_F
 =O_p(T^{-1/2}).
\]
\end{lemma}
The strong-mixing covariance inequality gives
$|\operatorname{Cov}(u_{0i}u_{0j},u_{hi}u_{hj})|
\le C\alpha_{x,T}(h)^{\delta/(4+\delta)}$; summing over lags bounds
each sample-product variance by $C/T$, uniformly in the triangular array.
Since $d$ is fixed, the claim follows. See also \cite{Ibragimov1971IndependentAS}.

\begin{lemma}
\label{lem:sample-moment-bounds}\label{rev9:primitive}
Use the centered quantities in \eqref{rev9:primitives-definition}:
\[
 F_T=T^{-1}XP_TX^\top-M_x^*,\qquad
 G_T=T^{-1}XP_TE^\top,\qquad E_T=T^{-1}EP_TE^\top-\Sigma_e^*.
\]
Then
\begin{equation}
 \norm{F_T}_F=O_p(T^{-1/2}),
 \qquad
 \norm{G_T}_F=O_p\!\left(\sqrt{\frac NT}\right),
 \qquad
 \norm{E_T}_F=O_p\!\left(\frac N{\sqrt T}\right).
 \label{eq:sample-moment-bounds}
\end{equation}
For
\[
 \cW_T=\widehat{M}_{y}-\Sigma_y^*
 =Z^*F_TZ^{*\top}+Z^*G_T+G_T^\top Z^{*\top}+E_T,
\]
we also have
\begin{equation}
 \norm{\cW_T}_F=O_p\!\left(\frac N{\sqrt T}\right).
 \label{eq:WT-Frobenius}
\end{equation}
\end{lemma}

\begin{proof}
Lemma \ref{lem:alpha-implies-quadratic} proved the first bound.

For the raw quantities $G_T^0=T^{-1}XE^\top$ and
$E_T^0=T^{-1}EE^\top-\Sigma_e^*$, all cross terms with $t\ne s$ vanish because the errors are temporally
independent, centered, and independent of the factor process.  Hence
\[
 E\norm{G_T^0}_F^2
 =\frac1{T^2}\sum_{t=1}^T
   E\{\norm{x_t}_2^2\norm{e_t}_2^2\}
 \le C\frac NT.
\]
Similarly,
\[
 E\norm{E_T^0}_F^2
 =\frac1T E\norm{e_te_t^\top-\Sigma_e^*}_F^2
 \le C\frac{N^2}{T},
\]
where the last inequality follows from coordinate independence and the uniform fourth moments.
Centering subtracts $\bar x\bar x^\top$, $\bar x\bar e^\top$ and
$\bar e\bar e^\top$, whose Frobenius norms are respectively
$O_p(T^{-1})$, $O_p(\sqrt N/T)$ and $O_p(N/T)$.
Moreover, $E\bar e_r^4\le C/T^2$ and bounded true rows give
\begin{equation}\label{rev9:primitive-diagonals}
 \|\operatorname{diag}E_T\|_2+
 \|\operatorname{diag}\cW_T\|_2=O_p(\sqrt{N/T}).
\end{equation}
Finally, row boundedness and fixed $d$ imply
$\norm{Z^*}_{\op}=O(\sqrt N)$, so the displayed decomposition of $\cW_T$ proves \eqref{eq:WT-Frobenius}.
\end{proof}

\subsubsection{Diagonal and $M_x$ nuisance projections}

Let $\cD_N$ be the set of diagonal symmetric $N\times N$ matrices.

For positive definite $P$ and symmetric $K$, define
\begin{equation}
 D_P(K)=\Diag\left\{(P\circ P)^{-1}\diag(PKP)\right\},
 \qquad
 \cR_D^P(K)=K-D_P(K).
 \label{eq:diagonal-projection}
\end{equation}
Then,
\begin{lemma}[Diagonal projection]
\label{lem:diagonal-projection}
$D_P(K)$ is the $P$-orthogonal projection onto $\cD_N$,
\begin{align}
 &\diag\{P\cR_D^P(K)P\}=0,
 \label{eq:diagonal-zero}
\\
 \norm{D_P(K)}_P^2
 =&\diag(PKP)^\top(P\circ P)^{-1}\diag(PKP).
 \label{eq:diagonal-projection-norm}
\end{align}
\end{lemma}

\begin{proof}
For $d\in\R^N$,
$\diag\{P\Diag(d)P\}=(P\circ P)d$.  The result follows immediately.  The matrix
$P\circ P$ is positive definite by the Schur product theorem.
\end{proof}

At the estimator, define the $M_x$ covariance tangent space
\[
 \cM_{\Zh}=\{\Zh U\Zh^\top:U=U^\top\}
\]
and its diagonal-residualized image
\[
 \cW_M=\{\cR_D^{\Ph}(\Zh U\Zh^\top):U=U^\top\}.
\]
Let $\Pi_M^{\Ph}$ be the orthogonal projection onto $\cW_M$ under
$\ip{\cdot}{\cdot}_{\Ph}$.  Define the intrinsic complete nuisance residual
\begin{equation}
 \cR_{D,M}^{\Ph}(K)
 =(I-\Pi_M^{\Ph})\cR_D^{\Ph}(K).
 \label{eq:complete-residual}
\end{equation}
Then it is orthogonal to all diagonal matrices and to every matrix
$\Zh U\Zh^\top$.

\begin{remark}[Coordinate expression and the Gram matrix]
Let $E_1,\ldots,E_s$, $s=d(d+1)/2$, be an orthonormal basis of symmetric
$d\times d$ matrices and put
\[
 W_a=\cR_D^{\Ph}(\Zh E_a\Zh^\top),
 \qquad
 \mathbb G_{ab}=\ip{W_a}{W_b}_{\Ph}.
\]
Whenever $W_1,\ldots,W_s$ are linearly independent,
\[
 \Pi_M^{\Ph}(K_D)
 =\sum_{a,b=1}^s
 W_a(\mathbb G^{-1})_{ab}\ip{K_D}{W_b}_{\Ph}.
\]
Here $(\mathbb G^{-1})_{ab}$ is the $(a,b)$ entry of the inverse matrix.  If the image has
smaller dimension, one may instead take an orthonormal basis of $\cW_M$; the projection
itself is intrinsic and does not depend on a chosen basis.
\end{remark}

Since consistency and \eqref{rev9:inverse} hold, set
\[
 \begin{gathered}
 s_T=\sqrt{Nx_T},\qquad a_N=(1+s_T)/N,\\
 \ell_N=N^{-1/2}\max_r\|\Zh_{r,\cdot}\|_2\le C(N^{-1/2}+\sqrt{x_T})=o_p(1).
 \end{gathered}
\]
In particular $\max_r\|\Zh_{r,\cdot}\|_2\le C(1+s_T)$,
$a_Ns_T=o_p(1)$ and $a_N(1+s_T)=o_p(1)$.

\begin{lemma}
The Gram matrix is non-singular in the high-dimensional case.
\end{lemma}
\begin{proof}
Suppose $\Zh U\Zh^\top$ were nonzero and diagonal. Its support $I$ has size
at most $\operatorname{rank}(U)\le d$. For $r\notin I$ its entire row
vanishes, so $\Zh_{r,\cdot}U\Zh^\top=0$ and, since $\Zh $ has full column rank, $\Zh_{r,\cdot}U=0$.
For a unit $u\in\operatorname{ran}(U)$, symmetry gives $\Zh_{r,\cdot}u=0$ outside $I$.
The original pairing is therefore
\[
 u^\top(N^{-1}\Zh^\top \widehat W  \Zh )u
 =N^{-1}\sum_{r\in I}\widehat W_{rr}(\Zh_{r,\cdot}u)^2
 \le C|I|\max_r\|\Zh_{r,\cdot}\|_2^2/N\le Cd\ell_N^2=o_p(1),
\]
contradicting the candidate Gram lower bound. Thus the
 diagonal-residualized image has the same dimension as $\mathcal M_{\Zh}$.
\end{proof}

After the alignment used in the interior proof,
\begin{equation}
 R^\top\Seh^{-1}\Zh=0,
 \qquad
 x_T=o_p(1).
 \label{eq:R-orthogonality}
\end{equation}
With $\delta A=\widehat A-A^*$ and $\delta B=\widehat B-B^*$, put
$\dot Z=-\widehat B\odot\delta A-\delta B\odot\widehat A$ and
$Q_Z=\delta B\odot\delta A$. Then
\begin{equation}
    \Zh = (B^*+\delta B)\odot (A^*+\delta A),\quad Z^*-\Zh=\dot Z+Q_Z.
 \label{eq:local-expansion}
\end{equation}
Hence we have
\begin{align}
 \dot Z=&-R-\Zh\mathcal Z^\top-Q_Z,
 \label{eq:dotZ-decomposition}\\
 \norm{Q_Z}_F=O_p(\sqrt N\,x_T),
 \quad
 \norm{\mathcal Z}_{\op}&=O_p(x_T^{1/2}),
 \quad
 \norm{\Mh-M_x^*}_{\op}=o_p(1).
 \label{eq:local-rates}
\end{align}
To see $\norm{Q_Z}_F=O_p(\sqrt N\,x_T)$, compute
\[
 \|Q_Z\|_F^2=\sum_j\|\delta B_{\cdot j}\|_2^2\|\delta A_{\cdot j}\|_2^2
 \le\|\delta B\|_F^2\|\delta A\|_F^2=O_p(Nx_T^2).
\]
The rotation calculation is
\[
 \|\mathcal Z\|_F^2
 =\|(\Zh -Z^*)^\top \widehat W \Zh \widehat H \|_F^2
 \le\|\Zh -Z^*\|_F^2\|\widehat W \Zh \|^2\|\widehat H \|^2
 \le CN^{-1}\|\Zh -Z^*\|_F^2=O_p(x_T).
\]

\begin{lemma}
\label{lem:projection-stability}
\begin{align}
 \max_{r\le N}\norm{(\Ph\Zh)_r}_2\le Ca_N=O_p(N^{-1}+\sqrt{x_T/N}),
 \label{eq:PZ-row}\\
 \norm{(\Ph\circ\Ph)^{-1}}_{\op}=O_p(1).
 \label{eq:Hadamard-inverse}
\end{align}
If $N\to\infty$, the Gram matrix $\mathbb G$ above satisfies
\begin{equation}
 cI_s\preceq\mathbb G\preceq CI_s
 \label{eq:Gram-stability}
\end{equation}
with probability tending to one.
\end{lemma}

\begin{proof}
$\norm{\Gh}_{\op}=O_p(N^{-1})$ and
$\Ph\Zh=\widehat W \Zh \Gh  \Mh^{-1}$ give the row bound, together with
$\|\Ph \Zh \|_F\le C/\sqrt N$. Write $\Ph=\widehat W -Q_0$, where
$Q_0=\widehat W \Zh \Gh  \Zh^\top \widehat W $ and $0\preceq Q_0\preceq \widehat W $.
Then
\[
 (Q_0)_{rr}\le C\ell_N^2,\qquad
 \sum_{u\ne r}\Ph _{ru}^2\le (Q_0^2)_{rr}\le C\ell_N^2.
\]
Thus $\Ph _{rr}$ is bounded away from zero and the off-diagonal row sums
of $\Ph \circ \Ph $ are $o_p(1)$; Gershgorin's theorem yields
$cI\preceq \Ph \circ \Ph \preceq CI$ and \eqref{eq:Hadamard-inverse}.

For $U=U^\top$, put $M(U)=\Zh U\Zh^\top$ and
$C_Z=\Zh^\top\Ph\Zh$.  Then
\[
 \norm{M(U)}_{\Ph}^2=\tr(C_ZUC_ZU)\asymp\norm U_F^2,
\]
because
$C_Z=\Mh^{-1}-\Mh^{-1}\Gh\Mh^{-1}$ has eigenvalues bounded above and away from zero.
Furthermore,
\[
 \{\diag(\Ph M(U)\Ph)\}_r=(\Ph\Zh)_{r,\cdot}U(\Ph\Zh)_{r,\cdot}^\top,
\]
so
\[
 \norm{\diag(\Ph M(U)\Ph)}_2^2
 \le\max_r\|(\Ph \Zh )_r\|_2^2\|\Ph \Zh \|_F^2\|U\|_F^2\le C a_N^2N^{-1}\|U\|_F^2.
\]
By \eqref{eq:diagonal-projection-norm},
$\norm{D_{\Ph}\{M(U)\}}_{\Ph}^2\le Ca_N^2N^{-1}\norm U_F^2$.  Since diagonal
projection is orthogonal,
\[
 \norm{\cR_D^{\Ph}\{M(U)\}}_{\Ph}^2
 =\norm{M(U)}_{\Ph}^2-\norm{D_{\Ph}\{M(U)\}}_{\Ph}^2
 \asymp\norm U_F^2.
\]
Writing $U=\sum_a u_aE_a$, because $E_1,\ldots,E_s$ are Frobenius orthonormal,
\begin{equation}
    \norm{U}_F^2
    =
    \sum_{a=1}^s u_a^2
    =
    \norm{u}_2^2.
    \label{eq:U-u-norm}
\end{equation}
By linearity,
\[
    W(U)=\sum_{a=1}^s u_aW_a.
\]
Therefore,
\begin{align}
    u^\top\mathbb G u
    &=
    \sum_{a=1}^s\sum_{b=1}^s
    u_au_b
    \langle W_a,W_b\rangle_{\Ph} \notag\\
    &=
    \left\langle
        \sum_{a=1}^s u_aW_a,
        \sum_{b=1}^s u_bW_b
    \right\rangle_{\Ph} \notag\\
    &=
    \norm{W(U)}_{\Ph}^2.
    \label{eq:Gram-quadratic-form}
\end{align}
Combining
\eqref{eq:U-u-norm}, and
\eqref{eq:Gram-quadratic-form}, we obtain \(c\norm{u}_2^2
    \leq
    u^\top\mathbb G u
    \leq
    C\norm{u}_2^2\) uniformly over $u\in\mathbb R^s$. Hence
\(c
    \leq
    \lambda_{\min}(\mathbb G)
    \leq
    \lambda_{\max}(\mathbb G)
    \leq
    C\). Since $\mathbb G$ is symmetric positive definite,
\(\norm{\mathbb G}_{\op}
    =
    \lambda_{\max}(\mathbb G)
    \leq C\), and \(\norm{\mathbb G^{-1}}_{\op}
    =
    \lambda_{\min}(\mathbb G)^{-1}
    \leq c^{-1}\).
This proves \eqref{eq:Gram-stability}.

\end{proof}

\subsubsection{The profiled loading direction and its information}

For $H\in\R^{N\times d}$ define
\begin{equation}
 \mathcal K(H)=H\Mh\Zh^\top+\Zh\Mh H^\top.
 \label{eq:K-H}
\end{equation}
Put
\[
 K_R=\mathcal K(R),
 \qquad
 K_Q=\mathcal K(Q_Z),
\]
and
\[
 K_{\mathcal Z}
 =\Zh(\mathcal Z^\top\Mh+\Mh\mathcal Z)\Zh^\top.
\]
The feasible loading covariance tangent is
\[
 K_L=\mathcal K(\dot Z)
 =-K_R-K_{\mathcal Z}-K_Q.
\]
Define
\begin{equation}
 V_T=\cR_{D,M}^{\Ph}(K_L).
 \label{eq:VT-definition}
\end{equation}
Since $K_{\mathcal Z}\in\cM_{\Zh}$,
\begin{equation}
 V_T=-\cR_{D,M}^{\Ph}(K_R)-\cR_{D,M}^{\Ph}(K_Q).
 \label{eq:VT-decomposition}
\end{equation}

\begin{lemma}
\label{lem:zero-score}
There are $U_V=U_V^\top$ and $v_V\in\mathbb R^N$ such that
\begin{equation}\label{kkt9:residual-compensation}
 V_T=K_L-\Zh U_V\Zh^\top-\operatorname{Diag}(v_V),
 \qquad\|v_V\|_2\le Ca_Ns_T,
\end{equation}
and
\[
 \ip{\Sh}{V_T}_F=-g^\top v_V.
\]
The direction is $\Ph$-orthogonal to all diagonal and $M_x$ covariance directions.
\end{lemma}
\begin{proof}
The complete projection subtracts a diagonal direction and an $M_x$
direction with its diagonal projection removed. Grouping their diagonal pieces
gives the displayed representation. First remove $K_{\mathcal Z}$,
which lies in the raw $M_x$ space. The direct diagonal corrections
of $K_R+K_Q$ have norm at most $Ca_Ns_T(1+\sqrt{x_T})$ by
\eqref{eq:diagonal-correction}. The stable Gram bounds the remaining
coefficient matrix by $Cs_T$, since $\|K_R\|_{\Ph}+\|K_Q\|_{\Ph}\le Cs_T$.
Its attached diagonal part has norm at most $Ca_NN^{-1/2}s_T$.
Norm equivalence on diagonal matrices proves the bound for $v_V$.
The loading and raw $M_x$ scores vanish, which gives the score identity.
Orthogonality follows from the two projections.
\end{proof}

Recall that
\begin{equation}
 \Ph R=\Seh^{-1}R,
 \qquad
 R^\top\Ph\Zh=0,
 \label{eq:P-R-orthogonality}
\end{equation}
and
\begin{equation}
 \norm{K_R}_{\Ph}^2
 =2\tr\!\left\{
 R^\top\Ph R\,\Mh(\Zh^\top\Ph\Zh)\Mh
 \right\}
 \asymp Nx_T.
 \label{eq:KR-information}
\end{equation}
\begin{lemma}
\label{lem:information}
For every $H$,
\begin{equation}
 \norm{D_{\Ph}\{\mathcal K(H)\}}_{\Ph}
 \le C a_N\norm{\Ph H}_F.
 \label{eq:diagonal-correction}
\end{equation}
If $N\to\infty$, then
\begin{equation}
 \norm{\Pi_M^{\Ph}\cR_D^{\Ph}(K_R)}_{\Ph}
 \le C\norm{D_{\Ph}(K_R)}_{\Ph}
 =O_p(a_N\sqrt{Nx_T}).
 \label{eq:M-projection-KR}
\end{equation}
\end{lemma}

\begin{proof}

For $K=\mathcal K(H)$,
\[
 \{\diag(\Ph K\Ph)\}_r
 =2(\Ph H)_{r,\cdot}\Mh(\Ph\Zh)_{r,\cdot}^\top.
\]
Use the row bound in Lemma~\ref{lem:projection-stability} and
\eqref{eq:diagonal-projection-norm} to obtain \eqref{eq:diagonal-correction}.

For the last statement, let $K_{R,D}=\cR_D^{\Ph}(K_R)$ and
$W(U)=\cR_D^{\Ph}(\Zh U\Zh^\top)$.  Since $K_{R,D}$ is orthogonal to diagonal
matrices and $K_R$ is orthogonal to the raw $M_x$ space,
\[
 \ip{K_{R,D}}{W(U)}_{\Ph}
 =\ip{K_R-D_{\Ph}(K_R)}{\Zh U\Zh^\top}_{\Ph}
 =-\ip{D_{\Ph}(K_R)}{\Zh U\Zh^\top}_{\Ph}.
\]
The norm equivalence and
Lemma~\ref{lem:projection-stability} show that the norm of this linear functional on
$\cW_M$ is bounded by $C\norm{D_{\Ph}(K_R)}_{\Ph}$:
\begin{align*}
    \ip{D_{\Ph}(K_R)}{\Zh U\Zh^\top}_{\Ph}&\leq \|D_{\Ph}(K_R)\|_{\Ph}\|\Zh U\Zh^{\top}\|_{\Ph}, \\c\|U\|_F\leq \|W(U)\|_{\Ph}, \quad \|\Zh U\Zh^{\top}\|_{\Ph}&=\{\operatorname{tr}(C_ZUC_ZU)\}^{1/2}\leq C\|U\|_F,
\end{align*}therefore, together with equation \eqref{eq:diagonal-correction}, we proved equation \eqref{eq:M-projection-KR}.
\end{proof}

\begin{lemma}
\label{lem:coercivity}
Under the preceding assumptions,
\begin{equation}
 \norm{V_T}_{\Ph}^2\asymp Nx_T.
 \label{eq:coercivity}
\end{equation}
\end{lemma}

\begin{proof}
When $N\to\infty$, Lemma~\ref{lem:information} gives
\[
 \norm{\cR_{D,M}^{\Ph}(K_R)-K_R}_{\Ph}
 =O_p(a_N\sqrt{Nx_T}).
\]

Therefore, $\norm{\cR_{D,M}^{\Ph}(K_R)}_{\Ph}^2\asymp Nx_T$.  Since the complete residual operator
is contractive and
\[
 \norm{K_Q}_{\Ph}\le C\norm{Q_Z}_F
 =O_p(\sqrt N x_T)=o_p(\sqrt{Nx_T}),
\]
\eqref{eq:VT-decomposition} proves the claim.
\end{proof}

Let
\[
 \Delta_Z=Z^*-\Zh,
 \qquad
 \Delta_M=M_x^*-\Mh
 .
\]
Since $\Delta_Z=\dot Z+Q_Z$, direct expansion yields
\begin{align}
 \Sigma_y^*-\Syh
 ={}&K_L+\Zh\Delta_M\Zh^\top-\mathcal E_e+\mathfrak R_T,
 \label{eq:covariance-expansion}\\
 \mathfrak R_T={}&K_Q
 +\Delta_Z\Mh\Delta_Z^\top
 +\Delta_Z\Delta_M\Zh^\top
 +\Zh\Delta_M\Delta_Z^\top
 +\Delta_Z\Delta_M\Delta_Z^\top.
 \label{eq:remainder}
\end{align}

\begin{lemma}
\label{lem:exact-identity}
The direction $V_T$ satisfies
\begin{equation}
 \norm{V_T}_{\Ph}^2
 =-\ip{\Ph\cW_T\Ph}{V_T}_F
  -\ip{\mathfrak R_T}{V_T}_{\Ph}+g^\top v_V.
 \label{eq:exact-identity}
\end{equation}
\end{lemma}

\begin{proof}
Since $\Syh-\widehat{M}_{y}=-(\Sigma_y^*-\Syh)-\cW_T$, substitute equation \eqref{eq:covariance-expansion} into the score identity $\ip{\Sh}{V_T}_F=-g^\top v_V$.  The terms $-\mathcal E_e$ and $\Zh\Delta_M\Zh^\top$ vanish exactly by orthogonality. Finally, $\ip{K_L}{V_T}_{\Ph}=\norm{V_T}_{\Ph}^2$ because $V_T$ is the orthogonal residual of $K_L$.
\end{proof}

\subsubsection{Nonlinear covariance terms}

\begin{lemma}
\label{lem:quadratic-orthogonality}
Let $Q_R=R\Mh R^\top$.  Then for every $H\in\R^{N\times d}$ and every symmetric $U$,
\begin{equation}
 \ip{Q_R}{\mathcal K(H)}_{\Ph}=0,
 \qquad
 \ip{Q_R}{\Zh U\Zh^\top}_{\Ph}=0.
 \label{eq:QR-orthogonality}
\end{equation}
The same statements hold with $R\Mh R^\top$ replaced by
$R\Delta_MR^\top$.
\end{lemma}

\begin{proof}
Expand each trace cyclically.  Every term contains either
$R^\top\Ph\Zh$ or $\Zh^\top\Ph R$, both of which are zero.
\end{proof}

\begin{lemma}[Remainder bound]
\label{lem:remainder-bound}
Under the preceding assumptions,
\begin{equation}
 \ip{\mathfrak R_T}{V_T}_{\Ph}=o_p(Nx_T).
 \label{eq:remainder-bound}
\end{equation}
\end{lemma}

\begin{proof}
The linear term $K_Q$ has norm $O_p(\sqrt N x_T)$, so its pairing with $V_T$ is
$O_p(Nx_T^{3/2})$.

Next use the exact relation $\Delta_Z=-R-\Zh\mathcal Z^\top$ to write
\begin{align*}
 \Delta_Z\Mh\Delta_Z^\top
 ={}&R\Mh R^\top
 +R\Mh\mathcal Z\Zh^\top
 +\Zh\mathcal Z^\top\Mh R^\top
 +\Zh\mathcal Z^\top\Mh\mathcal Z\Zh^\top.
\end{align*}
The last term is an $M_x$ direction and therefore has zero pairing with $V_T$.  The two
middle terms have $\Ph$-norm $O_p(\sqrt N x_T)$ and hence pairing
$O_p(Nx_T^{3/2})$.

For $Q_R=R\Mh R^\top$, Lemma~\ref{lem:quadratic-orthogonality} shows that it is
orthogonal to $K_R$, $K_Q$, and every raw $M_x$ correction.  Thus its pairing with $V_T$
comes only from diagonal pieces: the direct diagonal projections of $K_R$ and $K_Q$, and the diagonal pieces attached to the residualized $M_x$ projections $D_{\Ph}(\Zh U_{R/Q}\Zh^\top)$ in $\Pi_M^{\Ph}(K_{R/Q,D})$.  When $N\to\infty$,
the latter are smaller than the former by the bound
$\norm{D_{\Ph}\{\Zh U\Zh^\top\}}_{\Ph}\le Ca_NN^{-1/2}\norm U_F$ from the proof of
Lemma~\ref{lem:projection-stability}.  Consequently, by equation \eqref{eq:diagonal-correction}:
\begin{align*}
 |\ip{Q_R}{V_T}_{\Ph}|
 &\le C\norm{Q_R}_{\Ph}
 \left\{
 a_N\norm{\Ph R}_F
 +a_N\norm{\Ph Q_Z}_F
 \right\}+o_p(Nx_T)\\
 &\le C(Nx_T)
 a_N\{\sqrt{Nx_T}+\sqrt N x_T\}
 +o_p(Nx_T)
 =o_p(Nx_T),
\end{align*}
since $a_Ns_T=\sqrt{x_T/N}+x_T=o_p(1)$.

The terms linear in $\Delta_M$ are, after substituting
$\Delta_Z=-R-\Zh\mathcal Z^\top$, a loading-type direction with norm
$o_p(\sqrt{Nx_T})$ plus an $M_x$ direction, whose pairing is zero.  Their total pairing is
$o_p(Nx_T)$.  Finally, the leading component of
$\Delta_Z\Delta_M\Delta_Z^\top$ is $R\Delta_MR^\top$; by
Lemma~\ref{lem:quadratic-orthogonality}, only diagonal corrections contribute, and the
same displayed bound multiplied by $\norm{\Delta_M}_{\op}=o_p(1)$ applies.  Every
remaining component contains at least one factor
$\mathcal Z=O_p(x_T^{1/2})$ and is smaller.
\end{proof}

\subsubsection{Stochastic bound}

Define \(C_T=\Ph\Zh\Mh\). Woodbury gives \begin{equation}
 C_T=\Seh^{-1}\Zh\Gh.
 \label{eq:CT-Woodbury}
\end{equation}

We have proved that, $$\norm{C_T}_{\op}=O_p(N^{-1/2}),
 \qquad
 \norm{Z^{*\top}C_T}_{\op}=O_p(1).$$Hence we have:

\begin{lemma}[Sample covariance applied to the factor direction]
\label{lem:WCT}

\begin{equation}
  \norm{\cW_TC_T}_F=O_p\!\left(\sqrt{\frac NT}\right).
 \label{eq:WCT}
\end{equation}

\end{lemma}

\begin{proof}
Apply
deterministic norm inequalities on the high-probability event where the first two bounds
hold:
\begin{align*}
 \norm{Z^*F_TZ^{*\top}C_T}_F
 &\le \norm{Z^*}_{\op}\norm{F_T}_F
       \norm{Z^{*\top}C_T}_{\op}
 =O_p\!\left(\sqrt{\frac NT}\right),\\
 \norm{Z^*G_TC_T}_F
 &\le \norm{Z^*}_{\op}\norm{G_T}_F\norm{C_T}_{\op}
 =O_p\!\left(\sqrt{\frac NT}\right),\\
 \norm{G_T^\top Z^{*\top}C_T}_F
 &\le \norm{G_T}_F\norm{Z^{*\top}C_T}_{\op}
 =O_p\!\left(\sqrt{\frac NT}\right),\\
 \norm{E_TC_T}_F
 &\le \norm{E_T}_F\norm{C_T}_{\op}
 =O_p\!\left(\sqrt{\frac NT}\right).
\end{align*}
No conditioning on $C_T$ is used.
\end{proof}

\begin{lemma}
\label{lem:structured-scores}
For $H\in\R^{N\times d}$,
\begin{equation}
 \left|\ip{\Ph\cW_T\Ph}{\mathcal K(H)}_F\right|
 \le 2\norm{\Ph H}_F\norm{\cW_TC_T}_F.
 \label{eq:structured-score}
\end{equation}
For symmetric $U$,
\begin{equation}
 \left|\ip{\Ph\cW_T\Ph}{\Zh U\Zh^\top}_F\right|
 =O_p\!\left(\frac{\norm U_F}{\sqrt T}\right).
 \label{eq:M-score-stochastic}
\end{equation}
For diagonal $D$,
\begin{equation}
 \left|\ip{\Ph\cW_T\Ph}{D}_F\right|
 \le C\norm{\cW_T}_F\norm D_{\Ph}.
 \label{eq:diagonal-score-stochastic}
\end{equation}
\end{lemma}

\begin{proof}
Equation \eqref{eq:structured-score} follows from
\[
 \ip{\Ph\cW_T\Ph}{\mathcal K(H)}_F
 =2\tr(H^\top\Ph\cW_TC_T)
\]
and Cauchy--Schwarz. For the $M_x$ direction, use $\Ph\Zh=C_T\Mh^{-1}$.
Then $\norm{(\Ph\Zh)}_{\op}=O_p(N^{-1/2})$ and the proof of
Lemma~\ref{lem:WCT} also gives
$\norm{\cW_T(\Ph\Zh)}_F=O_p(\sqrt{N/T})$.  Therefore
\[
 \left|\tr\{U (\Ph\Zh)^\top\cW_T(\Ph\Zh)\}\right|
 \le \norm U_F\norm{(\Ph\Zh)}_{\op}\norm{\cW_T(\Ph\Zh)}_F
 =O_p(\norm U_F/\sqrt T).
\]
The diagonal bound follows from Frobenius Cauchy--Schwarz, boundedness of
$\norm{\Ph}_{\op}$, and norm equivalence on diagonal matrices.
\end{proof}

\begin{lemma}
\label{lem:stochastic-score}
Under the preceding assumptions,
\begin{equation}
 \left|\ip{\Ph\cW_T\Ph}{V_T}_F\right|
 =O_p\!\left(N\sqrt{\frac{x_T}{T}}\right)
  +O_p\!\left(\frac{Nx_T}{\sqrt T}\right)
 =O_p\!\left(N\sqrt{\frac{x_T}{T}}\right).
 \label{eq:stochastic-score}
\end{equation}
\end{lemma}

\begin{proof}
For the raw directions, \eqref{eq:structured-score} gives
\begin{align*}
 \left|\ip{\Ph\cW_T\Ph}{K_R}_F\right|
 &=O_p\!\left(N\sqrt{\frac{x_T}{T}}\right),\\
 \left|\ip{\Ph\cW_T\Ph}{K_Q}_F\right|
 &=O_p\!\left(\frac{Nx_T}{\sqrt T}\right).
\end{align*}

Consider next the nuisance pieces subtracted by the complete projection.
A diagonal piece generated by $\mathcal K(H)$ has norm at most
$Ca_N\|\Ph H\|_F$. Combining \eqref{eq:WT-Frobenius} and
\eqref{eq:diagonal-score-stochastic} gives, for $H=R,Q_Z$,
\[
 O_p(N/\sqrt T)a_Ns_T,\qquad
 O_p(N/\sqrt T)a_Ns_T\sqrt{x_T}.
\]
These are no larger than the corresponding raw bounds because
$a_N\sqrt N=N^{-1/2}+\sqrt{x_T}=o_p(1)$.
The stable Gram bounds the raw $M_x$ coefficient for $R$ by $Ca_Ns_T$
and for $Q_Z$ by $Cs_T\sqrt{x_T}$. Their stochastic pairings are
$O_p(a_Ns_T/\sqrt T)$ and $O_p(s_T\sqrt{x_T}/\sqrt T)$.
The attached diagonal pieces have norm at most $Ca_NN^{-1/2}$ times
these coefficient norms and obey the same or smaller target bounds.
\end{proof}

\subsubsection{Rate for R}

The diagonal score term in \eqref{eq:exact-identity} remains to be
controlled. Put $d_D=\|\mathcal E_e\|_F$, and
$h=\sqrt{N/T}$. Let $\Pi_0$ be the $\Ph $-orthogonal projection onto the
raw $M_x$ space, distinct from the residualized projection $\Pi_M^{\Ph}$:
\begin{equation}\label{kkt9:raw-M-projection}
 \Pi_0(Y)=\Zh \{C_Z^{-1}\Zh^\top \Ph Y\Ph  \Zh C_Z^{-1}\}\Zh^\top,
 \qquad\mathcal F_0=I-\Pi_0.
\end{equation}
The projection-stability calculation gives
\begin{equation}\label{kkt9:diagonal-coercivity}
 \|\Pi_0(-\mathcal E_e)\|_{\Ph}\le Ca_NN^{-1/2}d_D,\qquad
 cd_D^2\le\|\mathcal F_0(-\mathcal E_e)\|_{\Ph}^2\le Cd_D^2.
\end{equation}
Define $\mathsf B=M_x^*(I-\mathcal Z) \Mh^{-1}$.
Using $Z^*=\Zh (I-\mathcal Z)^\top-R$ and the $M_x$ score, direct expansion gives
\begin{equation}\label{kkt9:profiled-score}
 \Sh=-\Ph \{\mathcal F_0(-\mathcal E_e)-\mathcal K(R\mathsf B)
              +RM_x^*R^\top+\mathcal F_0(\cW_T)\}\Ph .
\end{equation}
Indeed, before the raw $M_x$ projection the expansion is
\[
 \Sigma_y^*-\widehat\Sigma_y
 =\Zh ((I-\mathcal Z)^\top M_x^*(I-\mathcal Z)-\Mh )\Zh^\top-\mathcal K(R\mathsf B)+RM_x^*R^\top-\mathcal E_e,
\]
and both $\mathcal K(R\mathsf B)$ and $RM_x^*R^\top$ are orthogonal to the raw
$M_x$ space. For $C_T=\Ph \Zh \Mh $ and $Q_0=\widehat W \Zh \Gh  \Zh^\top \widehat W $, the primitive
bounds imply
\[
 \begin{gathered}
 \|\operatorname{diag}\cW_T\|_2\le K_Th,\qquad
 \|\cW_TC_T\|_F\le K_Th,\\
 \|C_T^\top\cW_TC_T\|_F+\|(\Ph \Zh )^\top\cW_T\Ph \Zh \|_F
 \le K_T/\sqrt T,\qquad K_T=O_p(1).
 \end{gathered}
\]
Expand
\[
 \Ph \cW_T\Ph=\widehat W \cW_T\widehat W -\widehat W \cW_TQ_0-Q_0\cW_T\widehat W +Q_0\cW_TQ_0.
\]
The first diagonal norm is at most $K_Th$; each cross term obeys
\[
 \|\operatorname{diag}(\widehat W \cW_TQ_0)\|_2
 \le C\max_r\|\Zh_{r,\cdot}\|_2\|\cW_TC_T\|_F\le K_Th(1+s_T).
\]
Since $Q_0\cW_TQ_0=\widehat W \Zh (C_T^\top\cW_TC_T)\Zh^\top \widehat W $, its diagonal
norm is at most
\[
 C\max_r\|\Zh_{r,\cdot}\|_2\|\Zh \|_F K_T/\sqrt T\le K_Th(1+s_T).
\] The coefficient of $\Pi_0(\cW_T)$ is
$O_p(T^{-1/2})$ and its diagonal correction is smaller. Consequently
\[
 \|\operatorname{diag}\{\Ph \mathcal F_0(\cW_T)\Ph \}\|_2
 \le K_Th(1+s_T).
\]
The other diagonal vectors in \eqref{kkt9:profiled-score} satisfy
\[
 \begin{split}
 \|\operatorname{diag}\{\Ph \mathcal K(R\mathsf B)\Ph \}\|_2&\le Ca_Ns_T,\\
 \|\operatorname{diag}(\Ph RM_x^*R^\top\Ph)\|_2&\le Cs_T^2,\\
 \|\operatorname{diag}\{\Ph \mathcal F_0(-\mathcal E_e)\Ph \}\|_2&\le Cd_D.
 \end{split}
\]
The variance KKT condition $\langle\Sh,-\mathcal E_e\rangle_F\ge0$ now gives
\begin{equation}\label{kkt9:variance-identity-inequality}
 \|\mathcal F_0(-\mathcal E_e)\|_{\Ph}^2
 \le\langle\mathcal K(R\mathsf B),-\mathcal E_e\rangle_{\Ph}
 -\langle RM_x^*R^\top,-\mathcal E_e\rangle_{\Ph}
 -\langle\cW_T,\mathcal F_0(-\mathcal E_e)\rangle_{\Ph}.
\end{equation}
Taking absolute values in this preliminary bound and dividing by $d_D$
when positive yields
\begin{equation}\label{kkt9:weak-noise}
 d_D+\|g\|_2\le C(a_Ns_T+s_T^2)+K_Th(1+s_T).
\end{equation}
This does not assume a variance rate. By \eqref{kkt9:residual-compensation},
\begin{equation}\label{kkt9:residual-score-defect}
 |g^\top v_V|\le Ca_Ns_T(a_Ns_T+s_T^2)
       +K_Th a_Ns_T(1+s_T)=o_p(s_T^2)+O_p(hs_T),
\end{equation}
because $a_Ns_T=o_p(1)$ and $a_N(1+s_T)=o_p(1)$.

\begin{proposition}
\label{prop:R-rate}
Under Proof condition~\ref{supp:likelihood-conditions}
and \eqref{supp:uniform-interior},
\begin{equation}
 x_T=O_p(T^{-1}).
 \label{eq:R-rate}
\end{equation}
\end{proposition}

\begin{proof}
From Lemmas~\ref{lem:exact-identity}, \ref{lem:remainder-bound}, and
\ref{lem:stochastic-score}, and \eqref{kkt9:residual-score-defect},
\[
 \norm{V_T}_{\Ph}^2
 \le O_p\!\left(N\sqrt{\frac{x_T}{T}}\right)+o_p(Nx_T).
\]
By Lemma~\ref{lem:coercivity},
$\norm{V_T}_{\Ph}^2\asymp Nx_T$.  Absorb the last term and divide by $N$:
\[
 x_T\le O_p\!\left(\sqrt{\frac{x_T}{T}}\right).
\]
On $\{x_T>0\}$, division by $x_T^{1/2}$ gives
$x_T^{1/2}=O_p(T^{-1/2})$; the conclusion is trivial on $\{x_T=0\}$.
\end{proof}

\subsection{A sharper rate under Gaussian short-memory factors}
\label{sec:gaussian-sharp-rate}

Let \(N=pq\), and recall
\[
x_T
=
\frac1N
\operatorname{tr}
\left(
R^\top\widehat\Sigma_e^{-1}R
\right),
\qquad
R=\widehat Z(I_d-\mathcal Z)^\top-Z^*.
\]

We work under Proof condition~\ref{supp:likelihood-conditions} and
\eqref{supp:uniform-interior}. For the loading rate in
Theorem~10, the following additional conditions
specify the scope of this proof.

\paragraph*{Assumption}
\label{ass:gaussian-short-memory}
The latent factor process \(\{x_t\}_{t\in\mathbb Z}\) is
jointly sub-Gaussian with constants uniform along the triangular array,
and satisfies the uniform moment and short-memory conditions in
Proof condition~\ref{supp:likelihood-conditions}. A stationary Gaussian
factor process satisfying these conditions is included.

The idiosyncratic errors satisfy
\begin{equation}
    e_t\stackrel{\mathrm{i.i.d.}}{\sim}
N(0,\Sigma_e^*),
\qquad
\Sigma_e^*
=
\operatorname{Diag}
(\sigma_1^{*2},\ldots,\sigma_N^{*2}),
\label{G.2}
\end{equation}
are independent of the entire factor process, and
\begin{equation}
0<c_e
\le
\min_{r\le N}\sigma_r^{*2}
\le
\max_{r\le N}\sigma_r^{*2}
\le C_e<\infty .
\label{G.3}
\end{equation}
We additionally require a uniform positive margin from the variance
constraints: for some fixed $\epsilon_D>0$,
\begin{equation}\label{supp:variance-margin}
 d_-+\epsilon_D\le\sigma_r^{*2}\le d_+-\epsilon_D,
 \qquad r\le N.
\end{equation}
This is a supplementary condition used to justify the diagonal-score
localization below.
Throughout this subsection, \(p,q,T\to\infty\), and
\begin{equation}
\log N=o(T).
\label{G.4}
\end{equation}

We first record the coarse results proved in the preceding
subsections.  In particular,
\begin{equation}
x_T=O_p(T^{-1}),\quad
\|R\|_F^2\asymp Nx_T=O_p(NT^{-1}).
\label{G.6,7}
\end{equation}
Furthermore,
\begin{equation}
\|\mathcal Z\|_{\mathrm{op}}
=
O_p(x_T^{1/2})
=
O_p(T^{-1/2}),\quad
\frac1p\|\widehat A-A^*\|_F^2
+
\frac1q\|\widehat B-B^*\|_F^2
=
O_p(x_T)=O_p(T^{-1}).
\label{G.8,9}
\end{equation}
Finally,
\begin{equation}
\frac1N
\|\widehat\Sigma_e-\Sigma_e^*\|_F^2
=
O_p(T^{-1}).
\label{G.10}
\end{equation}
Also, \(\|\widehat M_x-M_x^*\|=O_p(T^{-1/2})\).
Recall \(s_T=\sqrt{Nx_T}\). To keep the exact equations consistent with the
centered likelihood, let \(X=(x_1,\ldots,x_T)\),
\(E=(e_1,\ldots,e_T)\), and \(P_T=I_T-\mathbf1\mathbf1^\top/T\), and use
\begin{equation}\label{min10:centered}
\begin{aligned}
 M_{X,T}&=T^{-1}XP_TX^\top,& F_T&=M_{X,T}-M_x^*,\\
 G_T&=T^{-1}XP_TE^\top,& E_T&=T^{-1}EP_TE^\top-\Sigma_e^*,\\
 \mathcal T&=-E_T-Z^*G_T-G_T^\top Z^{*\top},&
 \mathcal W_T&=Z^*F_TZ^{*\top}+Z^*G_T+G_T^\top Z^{*\top}+E_T.
\end{aligned}
\end{equation}
In the radial calculations,
\(e_r=(e_{r,1},\ldots,e_{r,T})P_T\) denotes the centered error row;
its squared standardized length has distribution \(\chi^2_\nu\),
\(\nu=T-1\).

\subsubsection{Proof Strategy}
To achieve the sharper rate, we consider the following trace:
\begin{equation}\label{min10:trace-split}
\begin{aligned}
 \langle\Ph\mathcal W_T\Ph,V_T\rangle_F
 ={}&-\tr\{\Ph\mathcal W_T\Ph
             [R\Mh\Zh^\top+\Zh\Mh R^\top]\}\\
 &+\tr\bigl(\Ph\mathcal W_T\Ph
    \{[I-\cR_{D,M}^{\Ph}]K_R
                    -\cR_{D,M}^{\Ph}(K_Q)\}\bigr).
\end{aligned}
\end{equation}
The second term is \(O_p(s_T/\sqrt T)+o_p(s_T^2)\), as verified in the
final proof below. We retain its linear part while studying the first term.

By \(\mathcal W_T=Z^*F_TZ^{*\top}+Z^*G_T+G_T^\top Z^{*\top}+E_T\) and $Z^*=\Zh(I-\mathcal{Z})^{\top}-R$, the trace can be written as 
\begin{equation}\label{eq:VTsquare}
\begin{aligned}
    &\frac{1}{2}\tr(\Ph \mathcal{W}_T\Ph[R\Mh\Zh^{\top}+\Zh\Mh R^{\top}])\\
=&-\tr(R^\top\widehat PRF_TB_T)-\tr(R^\top\widehat PRG_TC_T)\\
&+\tr(R^\top\widehat WG_T^\top B_T)+\tr(R^\top\widehat WE_TC_T)\, ,\\
&B_T=(I-\mathcal{Z})\widehat H^{-1}\widehat G_H; C_T=\widehat W\Zh \Gh.
\end{aligned}
\end{equation}
The first two terms are small: \(\tr(R^{\top}\Ph RF_TB_T)=o_p(Nx_T)\), \(\tr
(R^\top\widehat PRG_TC_T)=o_p(Nx_T)\). Therefore, we only need to estimate the rest:
\[
T_G=\tr(R^\top\widehat WG_T^\top B_T),\qquad
T_E=\tr(R^\top\widehat WE_TC_T).
\] 

Write \(W^*=(\Sigma_e^*)^{-1}\). To understand where the sharper rate comes from, we rewrite the error term $R=\Zh(I-\mathcal{Z})^{\top}-Z^*=(\widehat B-B^*)\odot \widehat A+\widehat B\odot (\widehat A-A^*)-\Zh \mathcal{Z}^{\top}-(\widehat B-B^*)\odot (\widehat A-A^*)$. By decomposing  $T_G=\tr(R^{\top} W^*G_T^{\top}B_T)+\tr[R^{\top}(\widehat{W}-W^*)G_T^{\top}B_T]$, the first term is the sum, weighted by $(B_T)_{m',m}$, of
\[
\widehat a_i^m\mathcal G_{ij,T}^{m'}\Delta b_j^m
+\Delta a_i^m\mathcal G_{ij,T}^{m'}\widehat B_{jm}
-\sum_\ell\mathcal Z_{m,\ell}\widehat a_i^\ell
                 \mathcal G_{ij,T}^{m'}\widehat B_{j\ell}
-\Delta a_i^m\mathcal G_{ij,T}^{m'}\Delta b_j^m,
\]
where $\mathcal G_{ij,T}^{m}=\sigma_{ij}^{*-2}T^{-1}
\sum_t(e_{ij,t}-\bar e_{ij})(x_{t,m}-\bar x_m)$ is the $(i,j)$ entry of the $p\times q$ matrix $\mathcal G^m$ for each $m\in[d]$. Conditional on the factor path, $\mathcal{G}^m$ has independent, mean-zero, Gaussian entries, by standard operator norm inequality, we know their operator norms are $O_p(\sqrt{(p+q)/T})$. This produces the main contribution of the sharper rate. To control the fitted-weight difference, we consider $\widetilde W=\{\operatorname{Diag}[\diag(T^{-1}EP_TE^\top)]\}^{-1}$ and $T_G=\tr(R^{\top} \widetilde WG_T^{\top}B_T)+\tr[R^{\top}(\widehat{W}-\widetilde W)G_T^{\top}B_T]$ instead. 

In the following, the remaining terms are kept as $O_p(s_T/\sqrt T)+o_p(s_T^2)$ until the final rate inequality. We firstly prove that each loading error is entry-wise $o_p(1)$. 

\subsubsection{Uniform rowwise loading localization}

Write
\[
\Delta A=\widehat A-A^*,
\qquad
\Delta B=\widehat B-B^*.
\]
In the row calculations, \(\widehat a_i=\widehat A_{i,\cdot}^{\top}\),
\(a_i^*=(A_{i,\cdot}^*)^{\top}\), \(\Delta a_i=\widehat a_i-a_i^*\), and
\(\Delta b_j=(\Delta B_{j,\cdot})^{\top}\) are column vectors.
The symbols \(\widehat b_j=\operatorname{Diag}(\widehat B_{j,\cdot})\)
and \(b_j^*=\operatorname{Diag}(B_{j,\cdot}^*)\) retain their original
diagonal-matrix meaning.

\begin{lemma}[Uniform rowwise loading localization]
\label{lem:gaussian-rowwise}
Under the additional conditions of this subsection,
\begin{equation}
\max_{i\le p}
\|\widehat a_i-a_i^*\|
=
O_p\left(
T^{-1/2}
+
\sqrt{\frac{\log p}{qT}}
\right),
\quad
\max_{j\le q}
\|\widehat B_{j,\cdot}-B_{j,\cdot}^*\|
=
O_p\left(
T^{-1/2}
+
\sqrt{\frac{\log q}{pT}}
\right).
\label{G.12}
\end{equation}
\end{lemma}

\begin{proof}
We prove \eqref{G.12} for $A$; the second result follows by symmetry.

Put
\begin{equation}
A_T
=
M_{X,T}(I_d-\mathcal Z)
\left(
I_d-\widehat M_x^{-1}\widehat{G}_H
\right).
\label{G.13}
\end{equation}
Since
\(\|\mathcal Z\|_{\mathrm{op}}=O_p(T^{-1/2})\), \(\|\widehat G_H\|_{\mathrm{op}}=O_p(N^{-1})\), \(\|M_{X,T}-M_x^*\|=O_p(T^{-1/2})\), and \(\widehat M_x^{-1}\) is uniformly bounded, \(\|A_T-M_x^*\|_{\mathrm{op}}=o_p(1).\)

Recall the exact loading score equation \eqref{eq:gradcondAexpressed}:
\begin{equation}
\sum_{j=1}^q
\widehat\Sigma_{ej}^{-1}
R_{bj}A_T\widehat b_j
=
\mathcal E_A,
\label{G.15}
\end{equation}
where \(\mathcal E_A\) denotes the sum of the four terms on the
right-hand side of \eqref{eq:gradcondAexpressed}.

Fix \(i\in[p]\).  Let \(r_{ij}\in\mathbb R^d\) be the transpose
of the \(i\)-th row of \(R_{bj}\), and let
\(\varepsilon_{A,i}\) be the transpose of the \(i\)-th row of
\(\mathcal E_A\).  Since
\(\widehat\Sigma_e\) is diagonal, taking the \(i\)-th row of
\eqref{G.15} gives
\begin{equation}
    \sum_{j=1}^q
\widehat w_{ij}
\widehat b_j A_T^\top r_{ij}
=
\varepsilon_{A,i},
\qquad
\widehat w_{ij}=\widehat\sigma_{ij}^{-2}.
\label{G.16}
\end{equation}

By the definition of \(R\), the column vector \(r_{ij}\) satisfies the exact identity \(r_{ij}=(I-\mathcal{Z})\widehat b_j\widehat a_i-b_j^*a_i^*\).

Hence
\begin{equation}
    r_{ij}
=
\widehat b_j\Delta a_i
+
\operatorname{Diag}(a_i^*)\Delta b_j
-
\mathcal Z\widehat z_{ij},
\qquad
\widehat z_{ij}
=
\widehat b_j\widehat a_i.
\label{G.17}
\end{equation}

Substitution of \eqref{G.17} into \eqref{G.16} yields
\(\mathcal H_{A,i}\Delta a_i+\mathcal C_{AB,i}-\mathcal C_{AZ,i}
=\varepsilon_{A,i}\), where
\[
\begin{aligned}
\mathcal H_{A,i}
 &=\sum_{j=1}^q\widehat w_{ij}\widehat b_j A_T^\top \widehat b_j,\\
\mathcal C_{AB,i}
 &=\sum_{j=1}^q\widehat w_{ij}\widehat b_j A_T^\top \operatorname{Diag}(a_i^*)\Delta b_j,\\
\mathcal C_{AZ,i}
 &=\sum_{j=1}^q\widehat w_{ij}\widehat b_j A_T^\top\mathcal Z\widehat z_{ij}.
\end{aligned}
\]

We prove that the coefficient matrix
\(\mathcal H_{A,i}\) cannot become singular: 
Define its principal part
\begin{equation}
\mathcal H_{A,i}^{(0)}
=
\sum_{j=1}^q
\widehat w_{ij}
\widehat b_j
M_x^*
\widehat b_j.
\label{G.22}
\end{equation}
For arbitrary \(v\in\mathbb R^d\),
\begin{align}
v^\top\mathcal H_{A,i}^{(0)}v
=
\sum_{j=1}^q
\widehat w_{ij}
(\widehat b_j v)^\top
M_x^*
(\widehat b_j v)                                     
\ge
c
\sum_{j=1}^q
\|\widehat b_j v\|^2                                  
=
c
\sum_{\ell=1}^d
v_\ell^2
\sum_{j=1}^q\widehat B_{j\ell}^2.
\end{align}
Therefore
\begin{equation}
\lambda_{\min}
(\mathcal H_{A,i}^{(0)})
\ge cq
\label{G.23}
\end{equation}
uniformly in \(i\).

On the other hand, column normalization implies
\begin{align}
\|
\mathcal H_{A,i}
-
\mathcal H_{A,i}^{(0)}
\|_{\mathrm{op}}
\le
C
\|A_T^\top-M_x^*\|_{\mathrm{op}}
\sum_{j=1}^q
\|\widehat b_j\|_{\mathrm{op}}^2
\le Cqd\|A_T^\top-M_x^*\|_{\mathrm{op}}
=o_p(q)\label{G.24}
\end{align}
uniformly in \(i\).  Hence, \(\inf_{i\le p}
s_{\min}(\mathcal H_{A,i})
\ge c_0q\) with probability tending to one.

The coupling term involving \(B\) satisfies, by \eqref{G.8,9},
\begin{align}
\|\mathcal C_{AB,i}\|
&\le C\sum_{j=1}^q\|\widehat B_{j,\cdot}\|\,\|\Delta b_j\|
 \le C\sqrt{qd}\,\|\Delta B\|_F
 =O_p(qT^{-1/2}) .\label{G.26}
\end{align}
uniformly in \(i\). Similarly,
\(\|\mathcal C_{AZ,i}\|\le Cq\|\mathcal Z\|\,\|\widehat a_i\|
=O_p(qT^{-1/2})\|\widehat a_i\|\).
The estimated-row factor is retained until the final absorption.

It remains to bound the right-hand side of \eqref{G.16}. Recall \(\mathcal T=-E_T-G_T^\top Z^{*\top}-Z^*G_T,\, C_T=\widehat\Sigma_e^{-1}\widehat Z\widehat G_H\). 
Easy to see
\(\|C_T\|_{\mathrm{op}}
=
O_p(N^{-1/2})\) and \(\|Z^{*\top}C_T\|_{\mathrm{op}}=O_p(1)\).

For \(g_{ij}
=
\frac1T\sum_{t=1}^T(x_t-\bar x)(e_{ij,t}-\bar e_{ij})\), conditional on the factor path,
\begin{equation}
g_{ij}\mid\{x_t\}_{t=1}^T
\sim
N\left(
0,
\frac{\sigma_{ij}^{*2}}{T}M_{X,T}
\right).
\label{G.34}
\end{equation}
Hence, conditional on the factor path, the variables
\(\sqrt T\,\|g_{ij}\|\) possess uniformly bounded Gaussian tails.
On an event where \(M_{X,T}\) is bounded, independent Gaussian
quadratic-form concentration over \(j\), followed by a union bound
over \(i\), gives
\[
\max_{i\le p}\left\{\frac1q\sum_{j=1}^q\|g_{ij}\|^2\right\}^{1/2}
=O_p\left(T^{-1/2}+\sqrt{\frac{\log p}{qT}}\right).
\]
Consequently, we also have
\[
\max_{i\le p}
\frac1q\sum_{j=1}^q\|g_{ij}\|
=
O_p\left(
T^{-1/2}+\sqrt{\frac{\log p}{qT}}
\right).
\]

For the pure-noise covariance term,
\begin{equation}
\max_r\|(E_TC_T)_r\|
\le \max_r\|(E_T)_{r,\cdot}\|_2\,\|C_T\|_{\mathrm{op}}
=O_p(T^{-1/2}).
\label{G.37}
\end{equation}
Indeed, conditional on a centered error row, the off-diagonal
entries in its covariance row are independent Gaussian variables.
The chi-square concentration calculation in the proof of
Lemma~\ref{lemma:BGET} below gives
\(\max_r\|(E_T)_{r,\cdot}\|_2=O_p(\sqrt{N/T})\).
It also implies the original absolute-row bound
\[
\max_r\frac1N\sum_u|(E_T)_{ru}|
\le N^{-1/2}\max_r\|(E_T)_{r,\cdot}\|_2=O_p(T^{-1/2}).
\]
The first inequality in \eqref{G.37} holds for the fitted \(C_T\).

Regarding the term \(\mathcal T C_T\), \(\|(\mathcal T C_T)_{ij}\|
\le
\|(E_TC_T)_{ij}\|
+
C\|g_{ij}\|
+
C\|G_TC_T\|_{\mathrm{op}}\). Moreover
\(\|G_TC_T\|_{\mathrm{op}}
\le
\|G_T\|_F\|C_T\|_{\mathrm{op}}
=
O_p(T^{-1/2})\), because \(\|G_T\|_F=O_p(\sqrt{N/T})\).
Therefore,
\[
\max_i
\frac1q
\sum_{j=1}^q
\|(\mathcal T C_T)_{ij}\|
=
O_p\left(
T^{-1/2}
+
\sqrt{\frac{\log p}{qT}}
\right).
\]

For the first score term, the required weighted bound is
\[
\begin{aligned}
\frac1q\left\|\sum_j\widehat w_{ij}\widehat b_j
                      (\mathcal T C_T)_{ij,\cdot}^\top\right\|
&\le\frac Cq\left(\sum_j\|\widehat B_{j,\cdot}\|^2\right)^{1/2}
             \left(\sum_j\|(\mathcal T C_T)_{ij,\cdot}\|^2\right)^{1/2}\\
&=O_p\left(T^{-1/2}+\sqrt{\frac{\log p}{qT}}\right)
\end{aligned}
\]
uniformly in \(i\). The second stochastic term contains
\(\widehat H\widehat Z^\top
\widehat\Sigma_e^{-1}
\mathcal T
\widehat\Sigma_e^{-1}
\widehat Z\widehat G_H\), where \(\|\mathcal T\|_F
=
O_p(NT^{-1/2})\). 
Using
\(\|\widehat H\|_{\mathrm{op}}
+
\|\widehat G_H\|_{\mathrm{op}}
=
O_p(N^{-1})\) and \(\|\widehat Z\|_{\mathrm{op}}
=
O_p(N^{1/2})\), we obtain
\[
\left\|
\widehat H\widehat Z^\top
\widehat\Sigma_e^{-1}
\mathcal T
\widehat\Sigma_e^{-1}
\widehat Z\widehat G_H
\right\|_{\mathrm{op}}
=
O_p(T^{-1/2}).
\]
Therefore its contribution to the \(i\)-th row of \eqref{eq:gradcondAexpressed} is
\(O_p(qT^{-1/2})\|\widehat a_i\|\), uniformly in \(i\).

Finally consider the two terms involving
\(\mathcal E_e=\widehat\Sigma_e-\Sigma_e^*\). The third term satisfies
\[
\begin{aligned}
\frac1q\left\|\sum_j\widehat w_{ij}\widehat b_j
                      (\mathcal E_eC_T)_{ij,\cdot}^\top\right\|
&\le\frac{C\|\widehat a_i\|}{qN}
                 \sum_j| (\mathcal E_e)_{ij,ij}|\,\|\widehat B_{j,\cdot}\|^2\\
&\le\frac{C\|\widehat a_i\|}{qN}\|\mathcal E_e\|_F
                        \left(\sum_j\|\widehat B_{j,\cdot}\|^4\right)^{1/2}\\
&\le C\frac{\|\mathcal E_e\|_F}{N}\|\widehat a_i\|
 =O_p((NT)^{-1/2})\|\widehat a_i\|,
\end{aligned}
\]
where \(\sum_j\|\widehat B_{j,\cdot}\|^2=qd\) was used. For the fourth term,
\[
\left\|\widehat H\widehat Z^\top\widehat W\mathcal E_e
                  \widehat W\widehat Z\widehat G_H\right\|
\le C\|\mathcal E_e\|_F/N=O_p((NT)^{-1/2}),
\]
so its row contribution divided by \(q\) has the same bound.
Therefore, uniformly in \(i\),
\[
\frac1q\|\varepsilon_{A,i}\|
\le O_p\left(T^{-1/2}+\sqrt{\frac{\log p}{qT}}\right)
       +O_p(T^{-1/2})\|\widehat a_i\|.
\]
The row equation and \eqref{G.23}--\eqref{G.26} give, for
\(\alpha_A=\max_i\|\Delta a_i\|\),
\[
\alpha_A\le O_p\left(T^{-1/2}+\sqrt{\frac{\log p}{qT}}\right)
                       +O_p(T^{-1/2})(C+\alpha_A).
\]
The last term is absorbed because \(T^{-1/2}\to0\).

In summary, we have 
\[
\max_{i\le p}
\|\Delta a_i\|
=
O_p\left(
T^{-1/2}
+
\sqrt{\frac{\log p}{qT}}
\right).
\]
Interchanging \(p,A\) with \(q,B\) proves \eqref{G.12}.
The established row bounds now also give
\(\max_r\|(C_T)_r\|=O_p(N^{-1})\) and hence the original
entrywise inequality
\[
\|(E_TC_T)_r\|\le\frac C N\sum_u|(E_T)_{ru}|=O_p(T^{-1/2})
\]
uniformly in \(r\).
\end{proof}

Define \(\rho_T=T^{-1/2}+\sqrt{\log p/(qT)}+\sqrt{\log q/(pT)}\).
Lemma~\ref{lem:gaussian-rowwise} and \eqref{G.8,9} imply
\begin{equation}
\max_{i,j}\|r_{ij}\|
=
O_p\left(
T^{-1/2}
+
\sqrt{\frac{\log p}{qT}}
+
\sqrt{\frac{\log q}{pT}}
\right)
=o_p(1).
\label{G.46}
\end{equation}

\subsubsection{The quadratic loading term in the diagonal profile}

\begin{lemma}
\label{lem:gaussian-QR-diagonal}
Recall that \(Q_R=R\widehat M_xR^\top\). Then
\begin{equation}
\|D_{\widehat P}(Q_R)\|_{\max}
=
O_p(\rho_T^2).
\label{G.47}
\end{equation}
\end{lemma}

\begin{proof}
Since \(\widehat M_x\) is uniformly bounded,
\begin{equation}
\max_{ij}
r_{ij}^\top\widehat M_xr_{ij}
=
O_p(\rho_T^2).
\label{G.48}
\end{equation}
The alignment relation gives
\(\widehat PR
=
\widehat\Sigma_e^{-1}R\), and therefore
\begin{equation}
    \left\{
\operatorname{diag}
(\widehat P Q_R\widehat P)
\right\}_r
=
\widehat w_r^2
R_{r,\cdot}\widehat M_xR_{r,\cdot}^\top.
\label{G.49}
\end{equation}
By Woodbury's identity,
\(\widehat P_{rr}
=
\widehat w_r+O_p(N^{-1}),
\,
\widehat P_{rs}
=
O_p(N^{-1}),
\quad r\ne s.\)
Hence
\(\widehat P\circ\widehat P=D_w+E_N,\,D_w=\operatorname{Diag}(\widehat w_1^2,\ldots,\widehat w_N^2)\), where
\(\|E_N\|_{\infty\to\infty}=O_p(N^{-1}).\)
Since \(D_w\) is uniformly nonsingular, the Neumann expansion gives \((\widehat P\circ\widehat P)^{-1}D_w=I_N+O_p(N^{-1})\) in the \(\ell_\infty\)-operator norm. By the definition of
\(D_{\widehat P}\),
\[
D_{\widehat P}(Q_R)
=
\operatorname{Diag}
\{R_{r,\cdot}\widehat M_xR_{r,\cdot}^\top\}_{r=1}^N
+
O_p(N^{-1}\rho_T^2),
\]
which proves the assertion. The same calculation holds with
\(\widehat M_x\) replaced by \(M_x^*\) or \(M_{X,T}\).
\end{proof}

\begin{lemma}[Uniform consistency of the diagonal covariance estimator]
\label{lem:variance-max}
Under Proof condition~\ref{supp:likelihood-conditions} and the additional conditions of this subsection,
\[
\|\widehat\Sigma_e-\Sigma_e^*\|_{\max}
=
O_p\left(
\sqrt{\frac{\log N}{T}}
+\rho_T^2+\frac{\rho_T}{N}
\right)
=o_p(1).
\]
\end{lemma}

\begin{proof}
Let \(\mathcal E_e=\widehat\Sigma_e-\Sigma_e^*\). We first justify
the diagonal score equation. Put
\(g_r=[\widehat P(\widehat\Sigma_y-\widehat M_y)\widehat P]_{rr}\).
The row bound already proved gives
\(\max_r\|(\widehat P\widehat Z)_r\|=O_p(N^{-1})\) and
\(\max_r\|(\widehat PZ^*)_r\|=O_p(\rho_T+N^{-1})\).
Together with the primitive covariance bounds, these imply
\[
\begin{aligned}
 \max_r|[\widehat P\mathcal E_e\widehat P]_{rr}
                   -\widehat w_r^2(\mathcal E_e)_{rr}|
       &\le O_p(N^{-1})\|\mathcal E_e\|_{\max},\\
 \max_r|[\widehat PE_T\widehat P]_{rr}|&=O_p(\sqrt{\log N/T}+T^{-1/2}),\\
 \max_r\|(G_T\widehat P)_{\cdot,r}\|&=O_p(\sqrt{\log N/T}+T^{-1/2}),\\
 \max_r|g_r-\widehat w_r^2(\mathcal E_e)_{rr}|&=o_p(1).
\end{aligned}
\]
The variance box bounds \(\|\mathcal E_e\|_{\max}\). At a lower
endpoint the last display makes \(g_r<0\), and at an upper endpoint
it makes \(g_r>0\), uniformly, because the true variances have a
positive margin. Both contradict the corresponding one-sided KKT
condition for the maximizer. Thus every variance constraint is
inactive with probability tending to one. The diagonal covariance
score equation now yields
\begin{equation}
\mathcal{E}_e=\mathcal T^{(d)}+\mathcal B_T(R)+\mathcal K_T(\mathcal{E}_e),
\label{eq:variance-expansion}
\end{equation}
where \(\mathcal T^{(d)}\), \(\mathcal B_T(R)\), and
\(\mathcal K_T(\mathcal{E}_e)\) collect respectively the primitive covariance
fluctuation, loading-error terms, and the terms linear in \(\mathcal{E}_e\).

The Gaussian error assumption and conditioning on the factor path give
\(\|E_T\|_{\max}=O_p(\sqrt{\log N/T})\) and \(\max_r\|(G_T)_{\cdot,r}\|=O_p(\sqrt{\log N/T})\).
The entries of \(B_H,B_G\), defined below, are \(O_p(N^{-1})\),
with bounded row \(\ell_1\) norms. The four explicit products in
\(\mathfrak T_T\) therefore imply
\[
\|\mathcal T^{(d)}\|_{\max}=O_p(\sqrt{\log N/T}).
\]

The terms in \(\mathcal B_T(R)\) that are linear in \(R\) contain
either \(\widehat M_x^{-1}\widehat G_H\) or
\(\widehat G_H\widehat M_x^{-1}\), and are therefore
\(O_p(\rho_T/N)\) entrywise. The remaining dominant loading term is
\(R M_x^*R^\top\), whose entrywise maximum is
\(O_p(\rho_T^2)\) by the rowwise localization result. Hence
\[
\|\mathcal B_T(R)\|_{\max}
=
O_p(\rho_T^2+\rho_T/N).
\]

Finally, every term in \(\mathcal K_T(\mathcal{E}_e)\) contains at least one
low-rank factor
\(\widehat Z\widehat G_H\widehat Z^\top\widehat\Sigma_e^{-1}\)
or
\(\widehat Z\widehat H\widehat Z^\top\widehat\Sigma_e^{-1}\).
For a diagonal input, the single-factor diagonal coefficients are
\(O_p(N^{-1})\), and the double-factor diagonal coefficients have
row absolute sums \(O_p(N^{-1})\). Consequently, on an event of
arbitrarily high probability, we have
\[
\|\mathcal K_T(D)\|_{\max}
\le
\frac{C}{N}\|D\|_{\max}
\]
for any diagonal matrix \(D\).

Combining the above bounds gives
\[
\|\mathcal{E}_e\|_{\max}
\le
O_p\left(
\sqrt{\frac{\log N}{T}}
+\rho_T^2+\frac{\rho_T}{N}
\right)
+
\frac{C}{N}\|\mathcal{E}_e\|_{\max}.
\]
Since \(N\to\infty\), the last term can be absorbed into the left
hand side, which proves the result.
\end{proof}

\subsubsection{Explicit control of the diagonal-score remainder}

We next use the explicit diagonal covariance score equation to
control the terms generated by estimating \(\Sigma_e\).  Throughout
this subsection, write
\[
\widehat W=\widehat\Sigma_e^{-1},
\qquad
\mathcal{E}_e=\widehat\Sigma_e-\Sigma_e^*.
\]
and let \(\widehat z_r^\top\) and \(r_r^\top\) denote the \(r\)-th
rows of \(\widehat Z\) and \(R\), respectively. For a diagonal matrix, its Frobenius norm equals the Euclidean norm of its diagonal vector. Recall that
\(\|R\|_F^2\asymp Nx_T,\,\max_{r\le N}\|r_r\|=O_p(\rho_T)\) and \(\|\widehat H\|_{\mathrm{op}}+\|\widehat G_H\|_{\mathrm{op}}=O_p(N^{-1})\).

Recall equation \eqref{eq:sigma_e-errorexpression}:
\begin{equation}
        \begin{aligned}
            \Sigma_{e}^*-\widehat{\Sigma}_{e}=&\operatorname{Diag}\operatorname{diag}\big\{RM_{X,T}(I_d-\mathcal{Z})(I_d-\widehat{H}^{-1}\widehat{G}_H)\widehat{Z}^{\top}\\&+\widehat
            {Z}(I_d-\widehat{G}_H\widehat{H}^{-1})(I_d-\mathcal{Z})^{\top}M_{X,T}R^{\top}\\&-RM_{X,T}R^{\top}+\widehat{Z}\widehat{G}_H\widehat{Z}^{\top}\widehat{\Sigma}_e^{-1}RM_{X,T}R^{\top}\\&+RM_{X,T}R^{\top}\widehat{\Sigma}_e^{-1}\widehat{Z}\widehat{G}_H\widehat{Z}^{\top}\\&+[\mathcal{T}-\widehat{Z}\widehat{H}\widehat{Z}^{\top}\widehat{\Sigma}_e^{-1}\mathcal{T}\widehat{\Sigma}_{e}^{-1}\widehat{Z}\widehat{H}\widehat{Z}^{\top}\\&-\widehat{Z}\widehat{G}_H\widehat{Z}^{\top}\widehat{\Sigma}_{e}^{-1}\mathcal{T}(I_{pq}-\widehat{\Sigma}_e^{-1}\widehat{Z}\widehat{H}\widehat{Z}^{\top})\\&-(I_{pq}-\widehat{Z}\widehat{H}\widehat{Z}^{\top}\widehat{\Sigma}_e^{-1})\mathcal{T}\widehat{\Sigma}_{e}^{-1}\widehat{Z}\widehat{G}_H\widehat{Z}^{\top}]\\&-\widehat{Z}\widehat{H}\widehat{Z}^{\top}\widehat{\Sigma}_{e}^{-1}\mathcal{E}_e\widehat{\Sigma}_{e}^{-1}\widehat{Z}\widehat{H}\widehat{Z}^{\top}\\&-\widehat{Z}\widehat{G}_H\widehat{Z}^{\top}\widehat{\Sigma}_e^{-1}\mathcal{E}_e+\widehat{Z}\widehat{G}_H\widehat{Z}^{\top}\widehat{\Sigma}_e^{-1}\mathcal{E}_e\widehat{\Sigma}_e^{-1}\widehat{Z}\widehat{H}\widehat{Z}^{\top}\\&-\mathcal{E}_e\widehat{\Sigma}_e^{-1}\widehat{Z}\widehat{G}_H\widehat{Z}^{\top}+\widehat{Z}\widehat{H}\hat{Z}^{\top}\widehat{\Sigma}_e^{-1}\mathcal{E}_e\widehat{\Sigma}_e^{-1}\widehat{Z}\widehat{G}_H\widehat{Z}^{\top}\big\}.
        \end{aligned}
    \end{equation} collect all terms
that are linear in the primitive covariance fluctuation
\(\mathcal T\) as
\[
\begin{aligned}
\mathfrak T_T
={}&
\mathcal T
-
\widehat Z\widehat H\widehat Z^\top
\widehat W\mathcal T\widehat W
\widehat Z\widehat H\widehat Z^\top
\\
&-
\widehat Z\widehat G_H\widehat Z^\top
\widehat W\mathcal T
\left(
I_N-\widehat W\widehat Z
\widehat H\widehat Z^\top
\right)
\\
&-
\left(
I_N-\widehat Z\widehat H
\widehat Z^\top\widehat W
\right)
\mathcal T\widehat W
\widehat Z\widehat G_H\widehat Z^\top .
\end{aligned}
\]

This is the leading stochastic component of the diagonal score and
will be treated by the Gaussian profile-score lemma.  We now control
all remaining terms in \eqref{eq:sigma_e-errorexpression} directly.

Define
\[
\begin{aligned}
\ell_{R,T}=\operatorname{diag}\bigl\{
 &RM_{X,T}(I_d-\mathcal Z)
       (I_d-\widehat H^{-1}\widehat G_H)\widehat Z^\top\\
 &+\widehat Z(I_d-\widehat G_H\widehat H^{-1})
       (I_d-\mathcal Z)^\top M_{X,T}R^\top\bigr\},
\end{aligned}
\]
\[
q_{R,T}
=
-\operatorname{diag}(RM_{X,T}R^\top),
\]
and let \(r_{E,T}\) denote the sum of the diagonal vectors of all terms in \eqref{eq:sigma_e-errorexpression} containing
\(\mathcal{E}_e\).  Then, using the alignment relation
\(R^\top\widehat W\widehat Z=0\), the two remaining terms in
\eqref{eq:sigma_e-errorexpression} containing \(RM_{X,T}R^\top\) vanish exactly, and hence
it becomes
\begin{equation}
-\operatorname{diag}(\mathcal{E}_e)
=
\operatorname{diag}(\mathfrak T_T)
+
\ell_{R,T}
+
q_{R,T}
+
r_{E,T}.
\label{eq:explicit-variance-decomposition}
\end{equation}

Recall that, 
\(|(\ell_{R,T})_r|
\le
C
\|r_r\|
\|\widehat G_H\|_{\mathrm{op}}
\|\widehat z_r\|
\le
\frac{C}{N}\|r_r\|\), where row boundedness of \(\widehat Z\) and boundedness of
\(M_{X,T},\widehat M_x^{-1}\), and \(I_d-\mathcal Z\) have been used.
Consequently,
\begin{equation}
\|\ell_{R,T}\|_2
\le
\frac{C}{N}\|R\|_F
=
O_p\left(
\sqrt{\frac{x_T}{N}}
\right).
\label{eq:linear-loading-diagonal-bound}
\end{equation}

For the quadratic loading term, \(|(q_{R,T})_r|
=
|r_r^\top M_{X,T}r_r|
\le
C\|r_r\|^2\).
It follows that
\[
\begin{aligned}
\|q_{R,T}\|_2^2
\le
C\sum_{r=1}^N\|r_r\|^4
\le
C
\left(
\max_{r\le N}\|r_r\|^2
\right)
\sum_{r=1}^N\|r_r\|^2.
\end{aligned}
\]
Hence
\begin{equation}
\|q_{R,T}\|_2
=
O_p\left(
\rho_T\sqrt{Nx_T}
\right).
\label{eq:quadratic-loading-diagonal-bound}
\end{equation}

We next turn to the terms containing \(\mathcal{E}_e\). Previous discussion in the proof of Theorem 9 showed: 
\begin{equation}
\|r_{E,T}\|_2
\le
\frac{C}{N}\|\mathcal{E}_e\|_F
=
O_p\left(
\frac1{\sqrt{NT}}
\right).
\label{eq:variance-feedback-l2-bound}
\end{equation}

Combining
\eqref{eq:explicit-variance-decomposition},
\eqref{eq:linear-loading-diagonal-bound},
\eqref{eq:quadratic-loading-diagonal-bound}, and
\eqref{eq:variance-feedback-l2-bound}, we obtain the explicit
decomposition
\begin{equation}
-\operatorname{diag}(\mathcal{E}_e)
=
\operatorname{diag}(\mathfrak T_T)
+
q_{R,T}
+
r_{e,T}^{\mathrm{rem}},
\label{eq:variance-leading-remainder}
\end{equation}
where
\begin{equation}
\|r_{e,T}^{\mathrm{rem}}\|_2
=
O_p\left(
\sqrt{\frac{x_T}{N}}
+
\frac1{\sqrt{NT}}
\right),
\label{eq:variance-remainder-l2}
\end{equation}
and
\[
\|q_{R,T}\|_2
=
O_p\left(
\rho_T\sqrt{Nx_T}
\right).
\]

For the term $\mathfrak{T}_T$, we rewrite as 
\begin{equation}
    \begin{aligned}
        \mathfrak T_T
=&(I-\Zh\Gh \Zh^{\top}\widehat W)\mathcal{T}(I-\Zh\Gh\Zh^{\top}\widehat W)^{\top}-\Zh(\widehat H-\Gh) \Zh^{\top}\widehat W\mathcal{T}\widehat W\Zh(\widehat H-\Gh)\Zh^{\top}. 
    \end{aligned}
\end{equation}
Take $B_G=\Zh\Gh \Zh^{\top}\widehat W$, $B_H=\Zh\widehat H \Zh^{\top}\widehat W$. Denoting $L=I-B_G,\, D=B_H-B_G$, and \(J_T
=
\widehat Z\widehat G_H
\widehat M_x^{-1}
(I-\mathcal Z)^\top\), we have 
\begin{align*}
\diag(\mathfrak T_T)
={}&
-\diag(E_T)
+
2\diag(B_GE_T)
-
\diag(B_GE_TB_G^\top)
\\
&+
2\diag(RG_TL^\top)
-
2\diag(J_TG_TL^\top)-\diag(D\mathcal TD^\top).
\end{align*}
Define the diagonal matrix
\[
\delta=\mathcal{E}_e-\operatorname{Diag}\{\diag(E_T)\}.
\]
Therefore, \begin{equation}\label{eq:delta estimation}
    \begin{aligned}
        \diag(\delta)={}&-2\diag(B_GE_T)
+\diag(B_GE_TB_G^\top)
\\
&-2\diag(RG_TL^\top)
+2\diag(J_TG_TL^\top)
\\
&+\diag(D\mathcal TD^\top)
-\ell_{R,T}-q_{R,T}-r_{E,T}.
    \end{aligned}
\end{equation}

Take a decomposition $T_G=\tr(R^{\top}\widetilde{W}G_T^{\top}B_T)+\tr[R^{\top}(\widehat W-\widetilde W)G_T^{\top}B_T]$, where $\widetilde W=\{\operatorname{Diag}[\diag(E_T)]+\Sigma_e^*\}^{-1}$. Easy to see that, $\widehat W-\widetilde W=-\widetilde W\delta\widehat W$. For the following lemmas write
\[
\begin{aligned}
G(Z,M,W)&=(M^{-1}+Z^\top WZ)^{-1},\\
B(Z,M,W)&=ZG(Z,M,W)Z^\top W,\qquad
C(Z,M,W)=WZG(Z,M,W).
\end{aligned}
\]
The symbols \(G_a\) denote the inverse cores corresponding to \(B_a\),
unless a separate \(C\)-chain core is specified. Then we have the following lemmas: 

\begin{lemma}\label{lemma:BGET}
\[
\|\diag(B_GE_T)\|_2
\le O_p(T^{-1/2})
+
O_p(T^{-1})+
O_p\left(
\sqrt{\frac{Nx_T}{T}}
\right)+O_p(T^{-1/2})\|\delta\|_F.
\]
\end{lemma}
\begin{proof}
    We decompose the $B_GE_T$ as following: 
    \begin{align*}
        &B_GE_T=B_G^*E_T+(B_1-B_G^*)E_T+(B_2-B_1)E_T+(B_3-B_2)E_T+(B_G-B_3)E_T,\\
        &G_H^*=\{(M_x^*)^{-1}+Z^{*\top}W^*Z^*\}^{-1},\quad B_G^*=Z^*G_H^*Z^{*\top}W^*,\\
        &B_1=B(Z^*,M^*_x, \widetilde W),\, B_2=B(\Zh, M^*_x, \widetilde W), \,B_3=B(\Zh, \Mh, \widetilde W). 
    \end{align*}
    Firstly, we consider the $B_G^*E_T$. Since each entry in $B_G^*$ is bounded by $C/N$, 
    \begin{equation}
        \|\diag(B_G^*E_T)\|_2^2
=\sum_s\left[\frac1T\sum_t[B_G^*(e_t-\bar e)]_s
 (e_{s,t}-\bar e_s)-(B_G^*\Sigma_e^*)_{ss}\right]^2=O_p(T^{-1}). 
    \end{equation}

    For the second term, we write $B_1-B_G^*=Z^*G_H^*Z^{*\top}(\widetilde W-W^*)(I_N-Z^*G_1Z^{*\top}\widetilde W)$. Put $\Xi=\widetilde W-W^*=\operatorname{Diag}(\xi_r)$. The squared-column calculation is
\[
\begin{aligned}
\|Z^*G_H^*Z^{*\top}\Xi\|_F^2
&=\sum_r\|(Z^*G_H^*Z^{*\top})_{\cdot,r}\|^2\xi_r^2\\
&\le\frac C N\sum_r\xi_r^2=O_p(T^{-1}).
\end{aligned}
\]
Since $\|I_N-B_1\|_{\mathrm{op}}=O_p(1)$,
$\|B_1-B_G^*\|_F=O_p(T^{-1/2})$. Denote $(B_1-B_G^*)$ as $\Delta B_1$. Then $(\Delta B_1E_T)_{ss}=(\Delta B_1)_{ss}(E_T)_{ss}+\sum_{r\neq s} (\Delta B_1)_{sr}(E_T)_{rs}$. 
    
    To understand the expectation
$\mathbb{E}[(\Delta B_1E_T)_{ss}-(\Delta B_1)_{ss}(E_T)_{ss}]^2$,
identify each centered error row with its column of coordinates in
an orthonormal basis of the centering subspace, and write $e_s=R_su_s$,
where $R_s=\|e_s\|_2$ and $u_s$ is uniform on the unit sphere in the
$\nu=T-1$ dimensional centering subspace. Conditional on
$\mathcal R=(R_1,\ldots,R_N)$,
\[
\mathbb E\{(E_T)_{rs}^2\mid\mathcal R\}
=\frac{R_r^2R_s^2}{T^2\nu},\qquad
\mathbb E\{(E_T)_{rs}(E_T)_{us}\mid\mathcal R\}=0
\quad(r\ne u,\ r,u\ne s).
\]
For $U_s=R_s^2/\sigma_s^{*2}\sim\chi^2_\nu$, both radius tails satisfy
\[
\begin{aligned}
\Pr(\max_s U_s>\nu+2\sqrt{\nu t}+2t)&\le Ne^{-t},\\
\Pr(\min_s U_s<\nu-2\sqrt{\nu t})&\le Ne^{-t}.
\end{aligned}
\]
Taking $t=T/64$ and using $\log N=o(T)$ shows, with probability
tending to one,
\[
0<c\le\min_s R_s/\sqrt T\le\max_s R_s/\sqrt T\le C.
\]
The entries of $\Delta B_1$ depend only on these radii. On this event,
\[
\sum_s\mathbb E\left[
 \{(\Delta B_1E_T)_{ss}-(\Delta B_1)_{ss}(E_T)_{ss}\}^2
 \mid\mathcal R\right]
\le\frac C T\|\Delta B_1\|_F^2.
\]
Conditional Markov's inequality gives
\[
\|\diag(\Delta B_1E_T)-\diag(\Delta B_1)\circ\diag(E_T)\|_F=O_p(T^{-1}).
\]
For the diagonal part, the same resolvent gives
\[
| (\Delta B_1)_{ss}|\le\frac C N|\xi_s|
                              +\frac C{N^2}\sum_r|\xi_r|.
\]
The inverse chi-square calculations below yield
\[
\mathbb E[\xi_s^2(E_T)_{ss}^2]=O(T^{-2}),\qquad
\sum_r|\xi_r|=O_p(NT^{-1/2}),
\]
and $\|\diag(E_T)\|_F=O_p(\sqrt{N/T})$. Therefore
\[
\|\diag(\Delta B_1)\circ\diag(E_T)\|_F=O_p((\sqrt N T)^{-1}).
\]
In summary,
\[\|\diag[(B_1-B_G^*)E_T]\|_F=O_p(T^{-1}).\]

    Then we look at the third term. We write
\[
\begin{aligned}
B_2-B_1={}&(\Zh-Z^*)G_1Z^{*\top}\widetilde W\\
 &+\Zh G_1(Z^{*\top}\widetilde WZ^*-\Zh^\top\widetilde W\Zh)
                               G_2Z^{*\top}\widetilde W\\
 &+\Zh G_2(\Zh-Z^*)^\top\widetilde W
\end{aligned}
\]
and estimate the third term via
\[
\|\diag[(B_2-B_1)E_T]\|_F^2
\le\bigl(\max_s\|(E_T)_{\cdot,s}\|_2^2\bigr)\|B_2-B_1\|_F^2.
\] 
    
    Under Gaussian assumption, \(e_t=(e_{1t},\ldots,e_{Nt})^\top\stackrel{\mathrm{iid}}{\sim} N(0,\Sigma_e^*)\), $e_s=R_su_s$ in the $\nu$-dimensional centering subspace. For a fixed \(s\), conditional on \(e_s\), we have \(e_r^\top e_s\big| e_s\sim N\left(0,\sigma_r^{*2}\|e_s\|_2^2\right)\) and \((E_T)_{rs}=\frac{\sigma_r^*\| e_s\|_2}{T}Z_{rs},\,r\neq s\). Therefore, the off-diagonal part of column norm is \(\sum_{r\neq s}(E_T)_{rs}^2=\frac{\| e_s\|_2^2}{T^2}\sum_{r\neq s}\sigma_r^{*2}Z_{rs}^2\). 

For the $\|e_s\|_2^2$ part, since \(\frac{\| e_s\|_2^2}{\sigma_s^{*2}}
\sim\chi_\nu^2\), normal \(\chi^2\) concentration gives,
\[\forall x>0,\quad
P\left(\frac{\|e_s\|_2^2}{\sigma_s^{*2}}>\nu+2\sqrt{\nu x}+2x\right)
\le e^{-x}.\] By taking \(x=3\log N\), we have: 
\[
P\left(\max_{s\le N}\frac{\| e_s\|_2^2}{\sigma_s^{*2}}>\nu+2\sqrt{3\nu\log N}+6\log N\right)\le Ne^{-3\log N}=N^{-2}.\]
By \(\log N=o(T)\), we can find a constant $C_1$ independent of $N,T$ such that \(\max_{s\le N}
\frac{\| e_s\|_2^2}{T}
\le C_1\) with probability tending to one. 

Then regarding the $Z_{rs}$ part, conditional on \(e_s\), we have
\(\sum_{r\neq s}\sigma_r^{*2}Z_{rs}^2\leq C_e\sum_{r\neq s}Z_{rs}^2\), with \(\sum_{r\neq s}Z_{rs}^2
\sim\chi_{N-1}^2\). Similarly, by \(\chi^2\) concentration, 
\[P\left(\sum_{r\ne s}Z_{rs}^2>(N-1)+2\sqrt{(N-1)x}+2x\;\middle|\; e_s\right)\le e^{-x}.\] Again, we take \(x=3\log N\). As \(N\) go to infinity, \((N-1)+2\sqrt{3(N-1)\log N}+6\log N\le C_2N\) for some $C_2$, hence, \(P(\sum_{r\ne s}\sigma_r^{*2}Z_{rs}^2>C_2N\;\big|\;e_s)\le N^{-3}\). Then we have 
\[
P\left(\max_{s\le N}\sum_{r\ne s}\sigma_r^{*2}Z_{rs}^2>C_2N\right)\le N\cdot N^{-3}=N^{-2}.
\]This give the upper bound of the off-diagonal part: 
\[
\max_{s\le N}
\sum_{r\ne s}(E_T)_{rs}^2
=
\max_{s\le N}
\frac{\| e_s\|_2^2}{T^2}
\sum_{r\ne s}\sigma_r^{*2}Z_{rs}^2\le
\frac{C_1T}{T^2}\,C_2N=C\frac NT.
\]
For diagonal entries, 
\((E_T)_{ss}=R_s^2/T-\sigma_s^{*2}\). Similarly we have
\[\max_{s\le N}|E_{T,ss}|=O_p\left(\sqrt{\frac{\log N}{T}}+\frac{\log N}{T}\right).\]
In summary, 
\[
\max_s\|(E_T)_{\cdot s}\|_2^2=\max_s\left\{\sum_{r\ne s}(E_T)_{rs}^2+(E_T)_{ss}^2\right\}=O_p\left(\frac NT\right).
\]

For the Frobenius norm of $B_2-B_1$, we only need to collect the square norms of each component and get $\|B_2-B_1\|_F^2=O_p(x_T)$. In summary, we have $\|\diag[(B_2-B_1)E_T]\|_F=O_p(\sqrt{\frac{Nx_T}{T}})$. 

The fourth term is easy, and we only present the result: \[\|\diag[(B_3-B_2)E_T]\|_F=O_p\left((\sqrt{N}T)^{-1}\right).\] 

Finally, we deal with the term $\|\diag[(B_G-B_3)E_T]\|_F^2\leq(\max_s\|(E_T)_{\cdot,s}\|_2^2)\|B_G-B_3\|_F^2$. Similar to $B_1-B_G^*$, the exact identity is
\[
B_G-B_3=\Zh G_3\Zh^\top(\widehat W-\widetilde W)(I_N-B_G).
\]
The last factor has bounded operator norm. Since each row of $\Zh$ and each entry of $NG_3$ are uniformly bounded (with high probability), we can find a sufficiently large constant $C$ such that $\|B_G-B_3\|_F\leq \frac{C}{\sqrt{N}}\|\delta\|_F$. Therefore, the fifth term is at most $\|\diag[(B_G-B_3)E_T]\|_F=O_p(\frac{1}{\sqrt{T}}\|\delta\|_F)$. 
\end{proof}

The above lemma controls the first term $B_GE_T$. Following a similar route, we have: 
\begin{lemma}\label{lemma:BEB}
    \[\|\diag\left(B_GE_TB_G^{\top}\right)\|_F=O_p\left(\frac{1}{\sqrt{NT}}\right)+O_p(\sqrt{Nx_T/T})+O_p(T^{-\frac{1}{2}})\|\delta\|_F. \]
\end{lemma}
\begin{proof}
    We decompose it as follows: 
    \begin{align*}
        B_GE_TB_G^{\top}={}&B_1E_TB_1^{\top}+(B_G-B_1)E_TB_1^{\top}\\
 &+B_1E_T(B_G-B_1)^{\top}+(B_G-B_1)E_T(B_G-B_1)^{\top}.
    \end{align*}
    For the first term, we consider
    \begin{align*}
        &\|\diag(B_1E_TB_1^{\top})\|_F^2\\
&\quad=\sum_s\Bigl[2\sum_{u<r}(B_1)_{sr}(B_1)_{su}(R_uR_r/T)u_r^\top u_u\\
&\hspace{85pt}+\sum_r(B_1)_{sr}^2(R_r^2/T-\sigma_r^{*2})\Bigr]^2. 
    \end{align*}
    Conditional on $\mathcal{R}$, we consider the expectation of each component on the right-hand side (off-diagonal part): 
    \begin{align*}
\operatorname{Var}
\left(\left.b_{1,s}^\top E_Tb_{1,s}\right|\mathcal R\right)
&=4\sum_{u<r}(B_1)_{sr}^2(B_1)_{su}^2\frac{R_u^2R_r^2}{T^2\nu}\\
&\le\frac CT\|b_{1,s}\|_2^4=O_p\left(\frac1{N^2T}\right); 
    \end{align*}
    and the diagonal part: 
    \begin{align*}
        \sum_r(B_1)_{sr}^2|R_r^2/T-\sigma_r^{*2}|\leq \frac{C}{N^2}\sum_r|(E_T)_{rr}|\leq_p \frac{2C}{N\sqrt{T}}. 
    \end{align*}
    Therefore, we have $\|\diag(B_1E_TB_1^{\top})\|_F=O_p[(NT)^{-\frac{1}{2}}]$. 

    For the second and third terms, let $b_{1,s}^\top$ and
$\delta_{b,s}^\top$ be row $s$ of $B_1$ and $B_G-B_1$. We consider
\[
\begin{aligned}
\|\diag((B_G-B_1)E_TB_1^\top)\|_2^2
&=\sum_s(\delta_{b,s}^\top E_Tb_{1,s})^2\\
&\le\max_s\|E_Tb_{1,s}\|_2^2\|B_G-B_1\|_F^2.
\end{aligned}
\]
Following the calculation in the proof of Lemma~\ref{lemma:BGET},
\[
\begin{aligned}
&\|\diag((B_G-B_1)E_TB_1^\top)\|_2^2\\
&\quad\le C\max_r\|(E_T)_{\cdot r}\|^2
 \{\|B_2-B_1\|_F^2+\|B_3-B_2\|_F^2+\|B_G-B_3\|_F^2\}.
\end{aligned}
\] 

    Put $D_B=B_G-B_1$. Both $B_1$ and $D_B$ have bounded row
$\ell_1$ norms, since their entries are $O_p(N^{-1})$. Thus
\[
\max_s\|E_T(D_B)_{s,\cdot}^\top\|_2
\le\max_s\|(D_B)_{s,\cdot}\|_1\max_r\|(E_T)_{\cdot r}\|_2
=O_p(\sqrt{N/T}).
\]
The last term therefore satisfies the explicit bound
\[
\begin{aligned}
\|\diag(D_BE_TD_B^\top)\|_2^2
&\le\max_s\|E_T(D_B)_{s,\cdot}^\top\|_2^2\|D_B\|_F^2\\
&\le O_p(N/T)
 \{O_p(x_T)+O_p((N^2T)^{-1})+O_p(N^{-1})\|\delta\|_F^2\}.
\end{aligned}
\]
The same maximum-column argument applies to the two cross terms,
because the rows of $B_1$ have bounded $\ell_1$ norm. This proves
the stated bound for all four original terms. 
\end{proof}

\begin{lemma}\label{min10:delta}
\[
\|\delta\|_F=
O_p(T^{-1/2})
+
O_p\left(
\sqrt{\frac{Nx_T\log N}{T}}
\right)
+
O_p(\rho_T\sqrt{Nx_T}).
\]
\[
\begin{aligned}
T_{G,\delta}&:=\tr[R^{\top}\widetilde W\delta \widehat WG_T^{\top}B_T],\\
|T_{G,\delta}|&\le C\|\delta\|_F O_p\left(\sqrt{\frac{Nx_T\log N}{T}}\right)\\
&=O_p(s_T\sqrt{\log N}/T)
  +O_p(\log N/T+\rho_T\sqrt{\log N/T})s_T^2.
\end{aligned}
\]
\end{lemma}

\begin{proof}
    For the third term in equation \eqref{eq:delta estimation}, we have \[\|\diag(RG_TL^\top)\|_2\le\max_s\|(G_TL^\top)_{\cdot s}\|\|R\|_F\le O_p(\sqrt{\log N/T})\sqrt{Nx_T}.\] For the fourth term, $\|J_T\|_F=O_p(N^{-1/2})$ and the
bounded row $\ell_1$ norms of $L=I_N-B_G$ give
$\|\diag(J_TG_TL^\top)\|_2=O_p(\sqrt{\log N/(NT)})=o_p(T^{-1/2})$.
For the fifth term, the small matrix is
$\widehat H-\Gh=\Gh\Mh^{-1}\widehat H=O_p(N^{-2})$, while
$B_H-B_G=\Zh(\widehat H-\Gh)\Zh^\top\widehat W$ has row norms
$O_p(N^{-3/2})$. Hence
\[
\begin{aligned}
\|\diag(D\mathcal T D^\top)\|_2
&\le\sqrt N\max_r\|D_{r,\cdot}\|_2^2\|\mathcal T\|_F\\
&=O_p((N^{3/2}\sqrt T)^{-1}),\qquad D=B_H-B_G,
\end{aligned}
\]
which also implies the original coarser $O_p((N\sqrt T)^{-1})$ bound. Therefore, the estimation of $\|\delta\|_F$ is given by equations \eqref{eq:linear-loading-diagonal-bound}--\eqref{eq:variance-feedback-l2-bound}, Lemmas \ref{lemma:BGET}--\ref{lemma:BEB} and the above discussions, after absorbing
$O_p(T^{-1/2})\|\delta\|_F$ into the left side. Here
$s_T/N=O_p((NT)^{-1/2})$ follows from the coarse rate; no comparison
between $N^{-1}$ and $\sqrt{\log N/T}$ is required. 

    For $T_{G,\delta}$, we estimate it via: $|T_{G,\delta}|\leq C\|\delta\|_F
\max_{r\in[N]}\|(G_T)_{\cdot,r}\|\,\|R\|_F$ and prove the result. 
\end{proof}

Put $C^*=C(Z^*,M_x^*,W^*)=W^*Z^*G_H^*$. Regarding the term $T_E=\tr[R^{\top}\widehat WE_TC^*]+\tr[R^{\top}\widehat WE_T(C_T-C^*)]$, use $|\tr[R^{\top}\widehat WE_TC^*]|\leq C\|R\|_F\|E_TC^*\|_F$ and $C_T-C^*=C_T-C_3+C_3-C_2+C_2-C_1+C_1-C^*$, where $C_1=C(Z^*,M^*_x,\widetilde W)$, $C_2=C(Z^*,M_x^*,\widehat W)$ and $C_3=C(\Zh,M_x^*,\widehat W)$. 

\begin{lemma}\label{min10:C1-action}
    \[\|E_TC^*\|_F=O_p(T^{-\frac{1}{2}}),\quad
       \|E_T(C_1-C^*)\|_F=O_p(T^{-1}).\]
\end{lemma}
\begin{proof}
 \begin{align*}
\mathbb E\|E_TC^*\|_F^2
&=\frac{\nu}{T^2}\tr\!\left[
C^{*\top}\mathbb E\{(e_te_t^\top-\Sigma_e^*)^2\}C^*\right]
 +\frac1{T^2}\|\Sigma_e^*C^*\|_F^2\\
&=\frac{\nu}{T^2}\tr\!\left[
C^{*\top}\{\tr(\Sigma_e^*)\Sigma_e^*+(\Sigma_e^*)^2\}C^*\right]
 +\frac1{T^2}\|\Sigma_e^*C^*\|_F^2\\
&\le\frac{\nu}{T^2}(NC_e^2+C_e^2)\|C^*\|_F^2
 +\frac1{T^2}\|\Sigma_e^*C^*\|_F^2=O(T^{-1}).
 \end{align*}
This proves the first equation. 

Regarding the term $C_1-C^*$, we write $C_1-C^*=(I-C^*Z^{*\top})(\widetilde W- W^*)Z^*G_H^*[I-Z^{*\top}(\widetilde W-W^*)Z^*G_1]$. Consider $\|E_T(\widetilde W-W^*)Z^*G_H^*\|_F$. With
$\xi_r=\widetilde w_r-w_r^*$ and $\Xi=\operatorname{Diag}(\xi_r)$,
\[
\|E_T\Xi Z^*G_H^*\|_F^2
=\sum_{s,\ell}\left|\sum_r(E_T)_{sr}\xi_r(Z^*G_H^*)_{r\ell}\right|^2,
\qquad \max_{r,\ell}|(Z^*G_H^*)_{r\ell}|\le C/N.
\]

Take \(\eta_r:=(E_T)_{rr}=R_r^2/T-\sigma^{*2}_r\), then \(\xi_r:=(\widetilde W-W^*)_{rr}=\frac{T}{R_r^2}-\frac1{\sigma_r^{*2}}=-\frac{T\eta_r}{\sigma^{*2}_rR_r^2}\). Suppose \(U_r=R_r^2/\sigma_r^{*2}\). Under the Gaussian assumption,
\(U_r\sim\chi_\nu^2\), \(\nu=T-1\) and \(\xi_r=\frac{T-U_r}{\sigma_r^{*2}U_r}\). 
Hence for \(T>5\),
\[
\mathbb E U_r^{-1}=\frac1{T-3},\qquad
\mathbb E U_r^{-2}=\frac1{(T-3)(T-5)},\qquad
\mathbb E\xi_r=\frac3{\sigma_r^{*2}(T-3)},
\]
and
\[
\begin{aligned}
\mathbb E(\xi_r^2)
&=\frac1{\sigma_r^{*4}}
  \mathbb E\left[\frac{(T-U_r)^2}{U_r^2}\right]\\
&=\frac1{\sigma_r^{*4}}
 \left\{\frac{T^2}{(T-3)(T-5)}-\frac{2T}{T-3}+1\right\}\\
&=\frac{2T+15}{\sigma_r^{*4}(T-3)(T-5)}=O(T^{-1}).
\end{aligned}
\]
Therefore, \(\mathbb E\xi_r^2=O(T^{-1})\).
For \(r\neq s\), conditional on the centered error row \(e_r\),
independence gives
\(\mathbb E[(E_T)_{sr}^2\mid e_r]=\sigma_s^{*2}R_r^2/T^2\).
Hence
\[
\begin{aligned}
\mathbb E[(E_T)_{sr}^2\xi_r^2\mid e_r]
 &=\frac{\sigma_s^{*2}}{T^2}R_r^2\xi_r^2,\\
\mathbb E[R_r^2\xi_r^2]
 &=\frac1{\sigma_r^{*2}}\mathbb E\frac{(T-U_r)^2}{U_r}\\
 &=\frac1{\sigma_r^{*2}}
       \left\{\frac{T^2}{T-3}-2T+(T-1)\right\}
   =\frac{2T+3}{\sigma_r^{*2}(T-3)}=O(1).
\end{aligned}
\]
Then, by the chi-square calculation,
\[
\mathbb E[(E_T)_{sr}^2\xi_r^2]=O(T^{-2}),\qquad r\ne s.
\]
For the cross term
$\mathbb E[(E_T)_{sr}(E_T)_{su}\xi_r\xi_u]$,
the conditional means vanish by sign symmetry, while the distinct
error rows are independent given $e_s$. The cases with one of
$r,u$ equal to $s$ vanish by sign symmetry of the other row as well.

When \(r=s\), represent the diagonal part via \(U_r\): 
\[
(E_T)_{rr}
=
\sigma_r^{*2}\frac{U_r-T}{T},
\quad
\xi_r
=
-\frac{U_r-T}{\sigma_r^{*2}U_r}.
\]Hence \((E_T)_{rr}^2\xi_r^2=\frac{(U_r-T)^4}{T^2U_r^2}\). Directly compute gives
\[
\begin{aligned}
\mathbb E
\left[
\frac{(U_r-T)^4}{U_r^2}
\right]
&=
\mathbb E U_r^2
-4T\mathbb EU_r
+6T^2
\\
&\quad
-4T^3\mathbb EU_r^{-1}
+T^4\mathbb EU_r^{-2}
\\
&=
\frac{12T^2+68T-15}{(T-3)(T-5)}.
\end{aligned}
\]Hence
\[
\mathbb E
\left[
(E_T)_{rr}^2\xi_r^2
\right]
=
\frac{12T^2+68T-15}
{T^2(T-3)(T-5)}
=
O(T^{-2}).
\]
Therefore, the vanishing cross terms and the row bound on $Z^*G_H^*$ give
\[
\mathbb E\|E_T\Xi Z^*G_H^*\|_F^2
=\sum_{s,\ell,r}|(Z^*G_H^*)_{r\ell}|^2
                    \mathbb E[(E_T)_{sr}^2\xi_r^2]
\le\frac C{N^2}\sum_{s,r}T^{-2}=O(T^{-2}).
\]
Thus $\|E_T\Xi Z^*G_H^*\|_F=O_p(T^{-1})$. 

Similarly, let $S_\Xi=Z^{*\top}\Xi Z^*G_H^*$. Its entries are sums of
independent $\xi_r$ with coefficients $O(N^{-1})$, so
\[
\|\mathbb ES_\Xi\|_F=O(T^{-1}),\qquad
\mathbb E\|S_\Xi-\mathbb ES_\Xi\|_F^2=O((NT)^{-1}),
\]
and consequently
\[
\begin{aligned}
\|S_\Xi\|_F&=O_p(T^{-1}+(NT)^{-1/2}),\\
\|E_TC^*S_\Xi\|_F
&\le O_p(T^{-1/2})O_p(T^{-1}+(NT)^{-1/2})=O_p(T^{-1}).
\end{aligned}
\]
Finally,
\[
I_d-Z^{*\top}\Xi Z^*G_1=(G_H^*)^{-1}G_1,
\qquad\|(G_H^*)^{-1}G_1\|_{\mathrm{op}}=O_p(1).
\]
Multiplying the original $C_1-C^*$ identity by $E_T$ now proves
$\|E_T(C_1-C^*)\|_F=O_p(T^{-1})$.
\end{proof}

Regarding the term $C_2-C_1$, by exact diagonal resolvent: 
\(\widehat W-\widetilde W=-\widetilde W\delta\widehat W\), we have
\[C_2-C_1=-(I-C_1Z^{*\top})\widetilde W\delta C_2.\] Easy to see that, the rows of $C_2$ satisfy \(\max_{s\in[N]}\|(C_2)_{s,\cdot}\|=O_p(N^{-1})\). Hence
\[
\begin{aligned}
\|
\delta C_2
\|_F^2
&=
\sum_s
\delta_{ss}^2
\|(C_2)_{s,\cdot}\|^2
\\
&\le
\frac C{N^2}
\|\delta\|_F^2.
\end{aligned}
\]Then we have
\[
\| C_2-C_1\|_{\op}
\le
\|C_2-C_1\|_F
\le
C\frac{\|\delta\|_F}{N}.
\]By coarse bound of \(E_T\): 
\[
\begin{aligned}
\|E_T(C_2-C_1)\|_F
&\le
\|E_T\|_F
\|C_2-C_1\|_{\op}
\\
&=
O_p\left(\frac N{\sqrt T}\right)
O_p\left(\frac{\|\delta\|_F}{N}\right)
\\
&=
O_p\left(
\frac{\|\delta\|_F}{\sqrt T}
\right).
\end{aligned}
\]

For the terms $C_T-C_3$ and $C_3-C_2$, direct computation gives $\|C_3-C_2\|_F=O_p(\sqrt{x_T/N})$ and $\|C_T-C_3\|_F=O_p(\frac{1}{N\sqrt{NT}})$. 

In summary, we proved 
\begin{lemma}\label{min10:TE}
    \[T_E=O_p\left(\sqrt{\frac{{Nx_T}}{{T}}}\right)+o_p(Nx_T).\]
\end{lemma}
\begin{proof}
In the order $C^*$, $C_1-C^*$, $C_2-C_1$, $C_3-C_2$, $C_T-C_3$,
the five trace bounds are
\[
\begin{gathered}
O_p(s_T/\sqrt T),\quad O_p(s_T/T),\quad O_p(s_T/\sqrt T)\|\delta\|_F,\\
O_p(s_T^2/\sqrt T),\quad O_p(s_T/(\sqrt N T)).
\end{gathered}
\]
Lemma~\ref{min10:delta} gives
\[
O_p(s_T/\sqrt T)\|\delta\|_F
=O_p(s_T/T)+O_p(\sqrt{\log N}/T+\rho_T/\sqrt T)s_T^2.
\]
The coefficients of $s_T^2$ tend to zero, while all linear terms are
$O_p(s_T/\sqrt T)$, proving the assertion.
\end{proof}

Finally, we go back to the leading term in $T_G$: 
\begin{lemma}\label{lemma:estimation of TG leading term}
    \[\tr(R^{\top}\widetilde{W}G_T^{\top}B_T)=
O_p\left(
\sqrt{\frac{p+q}{T}}
\sqrt{Nx_T}
\right)
+
o_p(Nx_T).\]
\end{lemma}
Before proving this Lemma, we estimate the operator norm: 
\begin{lemma}
\label{lem:cross-noise-op}

For $m\in[d]$, define
\(g_{m,ij}=\frac1T\sum_{t=1}^T
 (x_{t,m}-\bar x_m)(e_{ij,t}-\bar e_{ij})\), and \(\widetilde\Gamma_{m,T}=\left[
\widetilde w_{ij} g_{m,ij}\right]_{i\in[p],\,j\in[q]}\), with entries \(\widetilde\gamma_{ij,T}^m=\widetilde w_{ij}g_{m,ij}\).
Then
\[
\max_{m\in[d]}
\left\|
\widetilde\Gamma_{m,T}
\right\|_{\mathrm{op}}
=
O_p\left(
\frac{\sqrt p+\sqrt q}{\sqrt T}
\right).
\]
\end{lemma}

\begin{proof}

Let
\[
\mathcal F_x
=
\sigma(x_1,\ldots,x_T).
\]
For each spatial coordinate $s=(i,j)$, define the event
\[
\mathcal V_s
=
\left\{
\frac{\sigma_s^{*2}}{2}
\le
R_s^2/T
\le
2\sigma_s^{*2}
\right\}.
\]
Since \(R_s^2/\sigma_s^{*2}\sim\chi_\nu^2\), \(\nu=T-1\), standard concentration for $\chi^2$-random variables yields \(\mathbb P(\mathcal V_s^c)\le2\exp(-cT)\) for some constant $c>0$. Hence, with \(\mathcal V=\bigcap_{s=1}^N\mathcal V_s\), the union bound gives
\[
\mathbb P(\mathcal V^c)
\le
2N\exp(-cT)
=
o(1),
\]
where the last equality follows from $\log N=o(T)$.

Define the truncated entries \(\bar\gamma_{m,s}=\widetilde w_s g_{m,s}\mathbf 1_{\mathcal V_s}\), and let
\[
\bar\Gamma_{m,T}
=
\left[
\bar\gamma_{m,ij}
\right]_{i\in[p],\,j\in[q]}.
\]

Conditional on $\mathcal F_x$, the variables
$\{\bar\gamma_{m,s}:s\in[N]\}$ are independent, since each
$\bar\gamma_{m,s}$ depends only on the error sequence
$\{e_{s,t}\}_{t=1}^T$, and these error sequences are independent
across $s$.

Moreover,
\[
\mathbb E
\left(
\bar\gamma_{m,s}
\mid
\mathcal F_x
\right)
=
0.
\]
Indeed, both $\widetilde w_s$ and $\mathbf 1_{\mathcal V_s}$ are even functions of $(e_{s,1},\ldots,e_{s,T})$, whereas \(g_{m,s}\) is linear in the centered error row and hence odd. Since the Gaussian vector $(e_{s,1},\ldots,e_{s,T})$ is symmetric about the origin, the
conditional expectation vanishes.

Next, define the factor event
\[
\mathcal X
=
\left\{
\max_{m\in[d]}
\frac1T\sum_{t=1}^T x_{t,m}^2
\le C_x
\right\}.
\]
Under Proof condition~\ref{supp:likelihood-conditions}, 
\(\mathbb P(\mathcal X)\to1\). Conditional on $\mathcal F_x$, $g_{m,s}$ is Gaussian with variance \(\operatorname{Var}\left(g_{m,s}\mid\mathcal F_x\right)
=\frac{\sigma_s^{*2}}{T^2}\sum_{t=1}^T(x_{t,m}-\bar x_m)^2\). Hence, on $\mathcal X$, \(\left\|g_{m,s}\right\|_{\psi_2\mid\mathcal F_x}\le\frac{C}{\sqrt T}\). On $\mathcal V_s$ we also have \(|\widetilde w_s|\le\frac{2}{\sigma_s^{*2}}\le C\). Consequently, \(|\bar\gamma_{m,s}|
\le
C|g_{m,s}|\), and therefore
\[
\left\|
\bar\gamma_{m,s}
\right\|_{\psi_2\mid\mathcal F_x}
\le
\frac{C}{\sqrt T}.
\]

We now apply the standard operator-norm inequality for rectangular
matrices with independent, mean-zero, sub-Gaussian entries. For any
$u>0$, conditional on $\mathcal F_x$ and on $\mathcal X$,
\[
\mathbb P
\left(
\left.
\|\bar\Gamma_{m,T}\|_{\mathrm{op}}
>
\frac{C}{\sqrt T}
\left(
\sqrt p+\sqrt q+u
\right)
\right|
\mathcal F_x
\right)
\le
2e^{-u^2}.
\]
Taking, for example,
\[
u=\sqrt p+\sqrt q
\]
shows that
\[
\|\bar\Gamma_{m,T}\|_{\mathrm{op}}
=
O_p\left(
\frac{\sqrt p+\sqrt q}{\sqrt T}
\right).
\]
Since $d$ is fixed, a union bound over $m\in[d]$ gives
\[
\max_{m\in[d]}
\|\bar\Gamma_{m,T}\|_{\mathrm{op}}
=
O_p\left(
\frac{\sqrt p+\sqrt q}{\sqrt T}
\right).
\]
Finally, since
$\mathbb P(\mathcal V)\to1$ and
$\bar\Gamma_{m,T}=\widetilde\Gamma_{m,T}$ on $\mathcal V$, the same
bound holds for $\widetilde\Gamma_{m,T}$.
\end{proof}

Then we go back to the proof of Lemma \ref{lemma:estimation of TG leading term}: 
\begin{proof}[Proof of Lemma \ref{lemma:estimation of TG leading term}]
Recall the original expansion
\[
R=(\widehat B-B^*)\odot\widehat A
  +\widehat B\odot(\widehat A-A^*)
  -\widehat Z\mathcal Z^\top
  -(\widehat B-B^*)\odot(\widehat A-A^*).
\]
Then
\[
\begin{aligned}
\tr(R^\top\widetilde WG_T^\top B_T)
={}&\sum_{m,m'}(B_T)_{m',m}\sum_{i,j}
 \bigg\{\widehat a_i^m\widetilde\gamma_{ij,T}^{m'}\Delta b_j^m
       +\Delta a_i^m\widetilde\gamma_{ij,T}^{m'}\widehat B_{jm}\\
 &\qquad-\sum_\ell\mathcal Z_{m,\ell}\widehat a_i^\ell
                  \widetilde\gamma_{ij,T}^{m'}\widehat B_{j\ell}
       -\Delta a_i^m\widetilde\gamma_{ij,T}^{m'}\Delta b_j^m\bigg\}.
\end{aligned}
\]
By the preceding operator-norm lemma, normalization, and
\eqref{G.8,9}, the four contributions satisfy
\[
\begin{aligned}
\left|\tr(R^\top\widetilde WG_T^\top B_T)\right|
&\le O_p\left(\sqrt{\frac{p+q}{T}}\right)
 \Bigl\{\sqrt p\|\Delta B\|_F+\sqrt q\|\Delta A\|_F\\
&\hspace{95pt}+\sqrt N\|\mathcal Z\|+\|\Delta A\|_F\|\Delta B\|_F\Bigr\}\\
&=O_p\left(\sqrt{\frac{p+q}{T}}\right)
        \{\sqrt{Nx_T}+\sqrt N x_T\}
 =O_p\left(\sqrt{\frac{(p+q)Nx_T}{T}}\right).
\end{aligned}
\]
This inequality holds for the realized fitted loadings.
\end{proof}

Finally, we prove the sharper rate:
\begin{theorem}\label{min10:sharp-result}
Under the assumptions at the beginning of this subsection,
$x_T=O_p((p+q)/(pqT))$.
\end{theorem}
\begin{proof}
We first supply the nuisance bound retained in the original
Proof Strategy. Write
\(\mathcal K(H)=H\widehat M_x\widehat Z^\top
                 +\widehat Z\widehat M_xH^\top\),
\(\mathcal M(U)=\widehat ZU\widehat Z^\top\),
and \(Q_Z=\Delta B\odot\Delta A\). Then
\(K_R=\mathcal K(R)\), \(K_Q=\mathcal K(Q_Z)\), and
\(V_T=-\cR_{D,M}^{\Ph}(K_R+K_Q)\). The original projection construction gives
\[
\cR_{D,M}^{\Ph}\mathcal K(H)
=\mathcal K(H)-D_{\widehat P}\mathcal K(H)
              -\mathcal M(U_H)+D_{\widehat P}\mathcal M(U_H).
\]
By the proved row bounds and normalized Grams,
\[
\begin{gathered}
\|D_{\widehat P}\mathcal K(H)\|_{\widehat P}
       \le O_p(N^{-1})\|H\|_F,\qquad
\|D_{\widehat P}\mathcal M(U)\|_{\widehat P}
       \le O_p(N^{-3/2})\|U\|_F,\\
\|\mathcal M(U)\|_{\widehat P}^2
=\tr\{U(\widehat Z^\top\widehat P\widehat Z)
       U(\widehat Z^\top\widehat P\widehat Z)\}
\asymp\|U\|_F^2.
\end{gathered}
\]
Thus the finite factor-projection Gram remains nonsingular. The
orthogonality $\widehat Z^\top\widehat PR=0$ and
$\|Q_Z\|_F=O_p(s_T\sqrt{x_T})$ yield
\[
\|U_R\|_F\le O_p(s_T/N),\qquad
\|U_Q\|_F\le O_p(s_T\sqrt{x_T}).
\]
The four primitive terms in \eqref{min10:centered} give, pathwise,
$\|\mathcal W_TC_T\|_F=O_p(\sqrt{N/T})$ and
$\|\mathcal W_T\|_F=O_p(N/\sqrt T)$. Consequently, for any realized
$H,U$ and diagonal $D'$,
\[
\begin{aligned}
|\langle\mathcal W_T,\mathcal K(H)\rangle_{\widehat P}|
 &\le2\|\widehat PH\|_F\|\mathcal W_TC_T\|_F,\\
|\langle\mathcal W_T,\mathcal M(U)\rangle_{\widehat P}|
 &\le O_p(T^{-1/2})\|U\|_F,\qquad
|\langle\mathcal W_T,D'\rangle_{\widehat P}|
 \le O_p(N/\sqrt T)\|D'\|_{\widehat P}.
\end{aligned}
\]
Applying these inequalities to the raw, diagonal, and factor pieces
in the displayed projection gives
\begin{equation}\label{min10:nuisance-bound}
\begin{aligned}
&|\langle\mathcal W_T,(I-\cR_{D,M}^{\Ph})K_R-\cR_{D,M}^{\Ph}K_Q
                                             \rangle_{\widehat P}|\\
&\quad\le O_p(N/\sqrt T)\{s_T/N+s_T\sqrt{x_T}/N
                  +N^{-3/2}(s_T/N+s_T\sqrt{x_T})\}\\
&\qquad\quad+O_p(T^{-1/2})(s_T/N+s_T\sqrt{x_T})
                +O_p(\sqrt{N/T})s_T\sqrt{x_T}\\
&\quad=O_p(s_T/\sqrt T)+O_p(s_T^2/\sqrt T)=O_p(s_T/\sqrt T)+o_p(s_T^2).
\end{aligned}
\end{equation}
In particular, the linear $O_p(s_T/\sqrt T)$ term is kept.

The information bound follows from the exact identity
\[
\|K_R\|_{\widehat P}^2
=2\tr\{R^\top\widehat WR\,\widehat M_x
          (\widehat Z^\top\widehat P\widehat Z)\widehat M_x\}
\asymp s_T^2,
\]
since the projection corrections and $\cR_{D,M}^{\Ph}K_Q$ have norm
$o_p(s_T)$. Hence $\|V_T\|_{\widehat P}^2\asymp s_T^2$.
For completeness, with
$\mathsf B=M_x^*(I_d-\mathcal Z)\widehat M_x^{-1}=I_d+o_p(1)$,
the exact population-covariance expansion has projected remainder
\[
\mathcal E_T^{\mathrm{nl}}
=RM_x^*R^\top+\mathcal K\{R(I_d-\mathsf B)\}+K_Q.
\]
The last two terms have norm $o_p(s_T)$. The first is orthogonal to
every raw $\mathcal K(H)$ and $\mathcal M(U)$ because
$\widehat Z^\top\widehat PR=0$. Only the diagonal projection
corrections can pair with it, and
\[
\|D_{\widehat P}(RM_x^*R^\top)\|_{\widehat P}=O_p(\rho_Ts_T),
\qquad
|\langle RM_x^*R^\top,V_T\rangle_{\widehat P}|
\le O_p(\rho_Ts_T)O_p(s_T/N)=o_p(s_T^2).
\]
The CP loading score and the factor score vanish; Lemma~\ref{lem:variance-max}
has established the same for the diagonal score. Therefore the
original projected-score identity is
\[
\|V_T\|_{\widehat P}^2
=-\langle\widehat P\mathcal W_T\widehat P,V_T\rangle_F
  -\langle\mathcal E_T^{\mathrm{nl}},V_T\rangle_{\widehat P}.
\]
Combining \eqref{eq:VTsquare}, \eqref{min10:nuisance-bound}, and the
preceding $T_E$, $T_G$, and $T_{G,\delta}$ bounds gives
\[
\begin{aligned}
cs_T^2\le{}&O_p\left(s_T\sqrt{(p+q)/T}\right)+O_p(s_T/\sqrt T)\\
         &+O_p(s_T\sqrt{\log N}/T)\\
         &+O_p\bigl(T^{-1/2}+\log N/T+\rho_T\sqrt{\log N/T}\\
         &\hspace{45pt}+\sqrt{\log N}/T+\rho_T/\sqrt T\bigr)s_T^2+o_p(s_T^2).
\end{aligned}
\]
Since $T^{-1/2},\sqrt{\log N/T},\rho_T\to0$, the coefficients multiplying $s_T^2$ on the
right can be absorbed. The remaining linear terms are
$O_p(s_T\sqrt{(p+q)/T})$. Division by $s_T$ on $\{s_T>0\}$, with
$s_T=0$ already satisfying the result, yields
\[
s_T=O_p(\sqrt{(p+q)/T}),\qquad
x_T=O_p\left(\frac{p+q}{NT}\right).
\]
Finally, \eqref{G.8,9} and
$\widehat Z-Z^*=R+\widehat Z\mathcal Z^\top$ give the loading bound in Theorem~10.
\end{proof}

The conditional Gaussian calculations use the factor path through
its centered covariance bounds and its independence from the Gaussian
errors. They therefore apply to the jointly sub-Gaussian factor
processes satisfying the uniform short-memory conditions stated in this supplement.

\end{document}